%% file: main.tex
\documentclass[11pt, a4paper]{article}
\usepackage{subfiles}  
\usepackage{titling}

\usepackage[a4paper,top=2cm,bottom=2cm,left=2.5cm,right=2.5cm,marginparwidth=1.75cm]{geometry}

\usepackage{amsmath}
\usepackage{graphicx}
\usepackage{kotex}
\usepackage[colorlinks=true, allcolors=blue]{hyperref}
\usepackage{orcidlink}
\usepackage[title]{appendix}
\usepackage{mathrsfs}
\usepackage{amsfonts}
\usepackage{amssymb}
\usepackage{amsthm}
\usepackage{booktabs} % For \toprule, \midrule, \botrule
\usepackage{threeparttable} % For table footnotes
\usepackage{algpseudocode}
\usepackage{listings}
\usepackage{enumitem}
\usepackage{chngcntr}
\usepackage{lipsum}
\usepackage{subcaption}
\usepackage{authblk}
\usepackage[T1]{fontenc}    % Font encoding
\usepackage{csquotes}       % Include csquotes
\usepackage{diagbox}
\usepackage{natbib}
\usepackage{bbm}
\usepackage{xspace}

\newtheorem{theorem}{Theorem}

\usepackage{setspace}
\usepackage{titlesec}
\titleformat{\section} % Redefine section numbering format
  {\normalfont\Large\bfseries}{\thesection.}{1em}{}

\newcommand{\GJRM}{{\fontfamily{qcr}\selectfont GJRM}\xspace}

\usepackage{float}   % for customizing caption position
\usepackage{caption} % for customizing caption format


\title{Bayesian Nonparametric Modeling for Multivariate Conditional Copula Regression with Varying Coefficients}
\author[1]{Yujin Jeong}
\author[1,2]{Seonghyun Jeong\thanks{ Corresponding author: \texttt{sjeong@yonsei.ac.kr}}}
\affil[1]{Department of Statistics and Data Science, Yonsei University}
\affil[2]{Department of Applied Statistics, Yonsei University}

\date{}  % Remove date

\begin{document}
%---------------
\maketitle
\begin{abstract}
Multivariate mixed-type outcomes are difficult to model jointly, and additional complexity arises when both marginal effects and dependence structures vary with a covariate such as age or time. Existing approaches often impose restrictive dependence assumptions or lack sufficient flexibility to accommodate heterogeneous response types in a unified framework. To address this issue, we propose a Bayesian nonparametric framework for multivariate conditional copula regression with varying coefficients. The proposed model combines adaptive spline-based marginal regressions with an infinite mixture of Gaussian copulas whose weights vary with the covariate through a probit stick-breaking process. This construction provides flexible covariate-dependent dependence modeling while avoiding explicit global constraints on functional correlation matrices. We further establish approximation results for the proposed copula representation and develop a Markov chain Monte Carlo algorithm for posterior inference. Simulation studies show accurate recovery under correct specification and robust performance under copula misspecification. In an analysis of the BRFSS 2023 data, the proposed model reveals age-varying marginal effects and dependence patterns among multiple health outcomes, providing a coherent joint view of multimorbidity beyond separate marginal analyses. %\lipsum[1]
\end{abstract}


%-------------------------------------------
% Paper Body
%-------------------------------------------
%--- Section ---%
\section{Introduction}\label{sec1}

\subsection{Background and Motivation}
\label{sec:mot}

In many scientific fields, datasets often include multivariate response variables observed simultaneously, comprising both discrete and continuous measurements. Modeling such multi-response data is challenging because no natural multivariate distribution exists for mixed-type observations.
In this respect, copula regression models serve as a useful compromise between sufficient flexibility and accessible parameterization \citep{kolev2009copula}.
Within the multivariate regression framework, for $i=1,\dots,n$, let $\mathbf y_i=(y_{i1},\dots,y_{im})^\top\in\mathbb R^m$ denote the $m$-dimensional response vector. 
Denote its marginal distributions by $F_1(\cdot\,;\mathbf x_i),\dots,F_m(\cdot\,;\mathbf x_i)$, each depending on the $p$-dimensional fixed explanatory variables $\mathbf x_i=(x_{i1},\dots,x_{ip})^\top\in\mathbb R^p$.
A copula $C:[0,1]^m\rightarrow[0,1]$, defined as an $m$-dimensional distribution function with uniform marginals on $[0,1]$ for each coordinate, specifies the joint distribution of $\mathbf y_i$ as $C(F_1(y_{i1};\mathbf x_i),\dots,F_m(y_{im};\mathbf x_i))$.
If all marginals are continuous distributions, the joint distribution is uniquely specified by a copula $C$ according to Sklar's theorem \citep{Skla59}. When some marginals are discrete, Sklar's theorem does not hold, but a copula still provides a fully valid construction for the joint distribution \citep{geenens2020copula}.

The copula regression framework typically employs a parametric linear predictor for each marginal distribution, namely $\mathbf x_i^\top \boldsymbol{\beta}_k$ with coefficients $\boldsymbol{\beta}_k\in\mathbb R^p$ for $k=1,\dots,m$. However, analyses of real datasets often reveal a more dynamic structure, where the effects of $\mathbf x_i$ vary with another fixed covariate $t_i\in\mathcal T $ for some $\mathcal T \subset\mathbb R$, such as observation time or subject age.
Additionally, the dependence structure may also change with $t_i$, necessitating a copula function $C$ that accommodates such varying effects. 
We illustrate this using the BRFSS 2023 data, a large-scale health survey system conducted by the U.S. Centers for Disease Control and Prevention (CDC) in collaboration with state health departments (\url{https://www.cdc.gov/brfss/annual_data/annual_2023.html}).
We examine the relationship between an individual's health status and socio-demographic variables.
For the explanatory variables $\mathbf x_i$, we include socio-demographic characteristics of each respondent, such as gender, race, education level, smoking and drinking behavior, and insurance coverage, along with other relevant factors and an intercept term (see Table~\ref{tab:predictors} in Section~\ref{sec6} for the complete list of variables). For the response vector $\mathbf y_i$, we use eight variables indicating health status measures: diabetes ($y_{i1}$; $1=\text{no}$, $2=\text{prediabetes}$, $3=\text{yes}$), high blood pressure ($y_{i2}$; $0=\text{no}$, $1=\text{yes}$), high cholesterol ($y_{i3}$; $0=\text{no}$, $1=\text{yes}$), stroke ($y_{i4}$; $0=\text{no}$, $1=\text{yes}$), heart disease ($y_{i5}$; $0=\text{no}$, $1=\text{yes}$), arthritis ($y_{i6}$; $0=\text{no}$, $1=\text{yes}$), asthma ($y_{i7}$; $0=\text{no}$, $1=\text{yes}$), and the log-transformed body mass index (BMI) measure ($y_{i8}$; continuous).
Given that the influence of socio-demographic factors on an individual’s health status is expected to vary nonlinearly with age \citep{lynch2005life}, we allow the relationship between $\mathbf x_i$ and $\mathbf y_i$ to change with respondent’s age denoted by $t_i$. For illustration, the age variable $t_i$ is divided into five groups
($<30$, $30$--$44$, $45$--$59$, $60$--$74$, $\ge75$), and each response variable $y_{ik}$ is  separately regressed on the explanatory variables $\mathbf x_i$ within each age group using appropriate models: cumulative logit regression for $y_{i1}$, logistic regression for $y_{i2},\dots,y_{i7}$, and linear regression for $y_{i8}$. In total, this yields 40 separate analyses. Figure~\ref{fig:motiv_marginal} shows the estimated regression coefficients and their 95\% confidence intervals from these analyses, revealing distinct patterns across the age groups. To assess pairwise dependencies among the response variables in their joint distributions conditional on the explanatory variables, we calculate the correlations of standardized residuals within each age group and construct 95\% confidence intervals using the Fisher transformation; the results are shown in Figure~\ref{fig:motiv_corr}. The figure demonstrates not only that the health status conditions are correlated but also that the strength of these dependencies varies across the age groups, consistent with findings in the multimorbidity literature \citep{barnett2012epidemiology,chmiel2014spreading}.

\begin{figure}[p!]
    \centering
    \includegraphics[width=\textwidth]{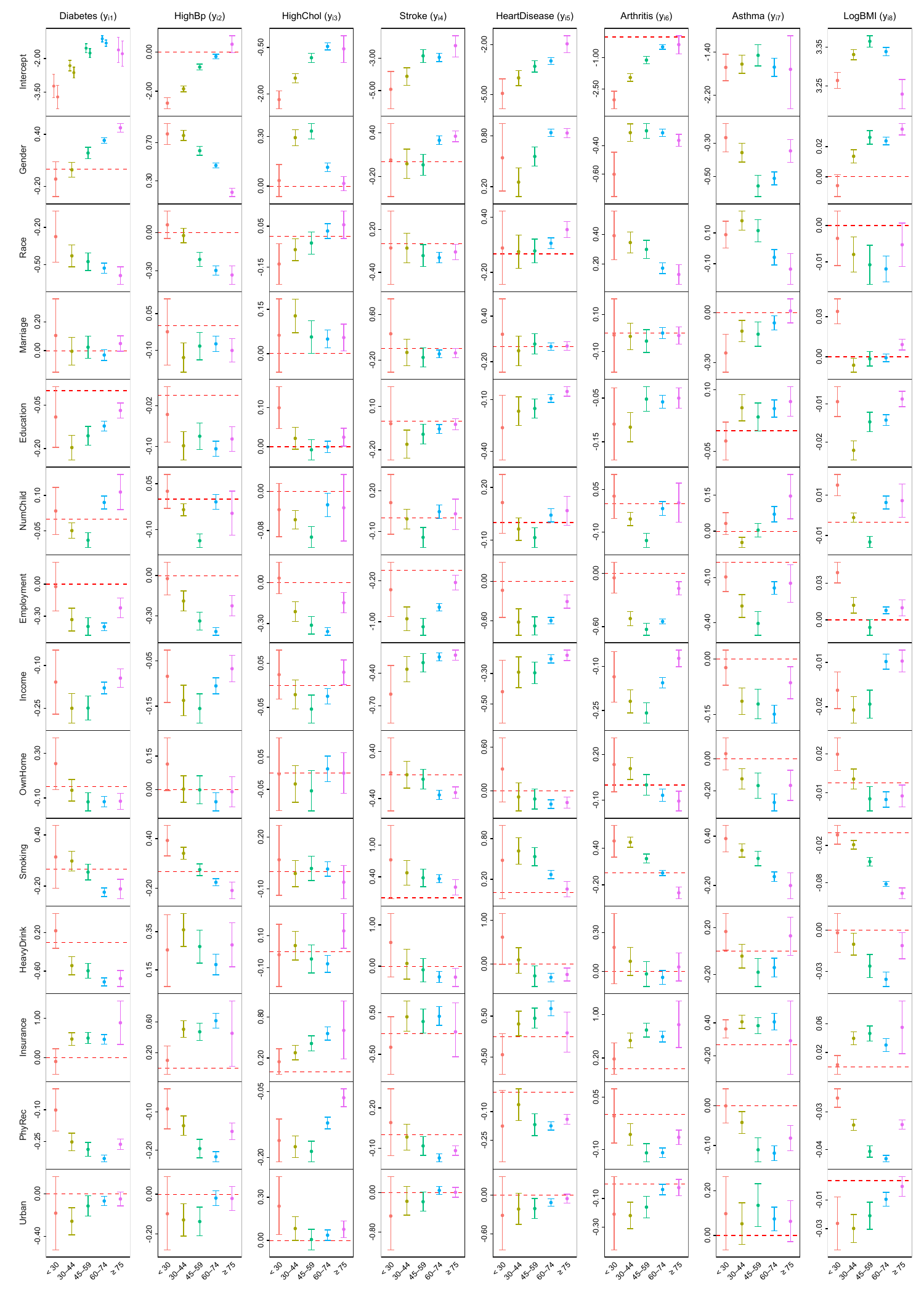}
    \caption{Parameter estimates and 95\% confidence intervals from the analyses of the BRFSS 2023 data. Each column presents the results of a separate analysis relating the corresponding response variable to the explanatory variables. For the diabetes variable, the two intercepts represent the cut points, $-\text{logit}\Pr\{y_{i1}\le b\mid t_i\}$, $b=1,2$, when $\mathbf x_i$ is zero.}
    \label{fig:motiv_marginal}
\end{figure}
%---------------------------------------

%------------------------- corr plot
\begin{figure}[t!]
    \centering
    \includegraphics[width=\textwidth]{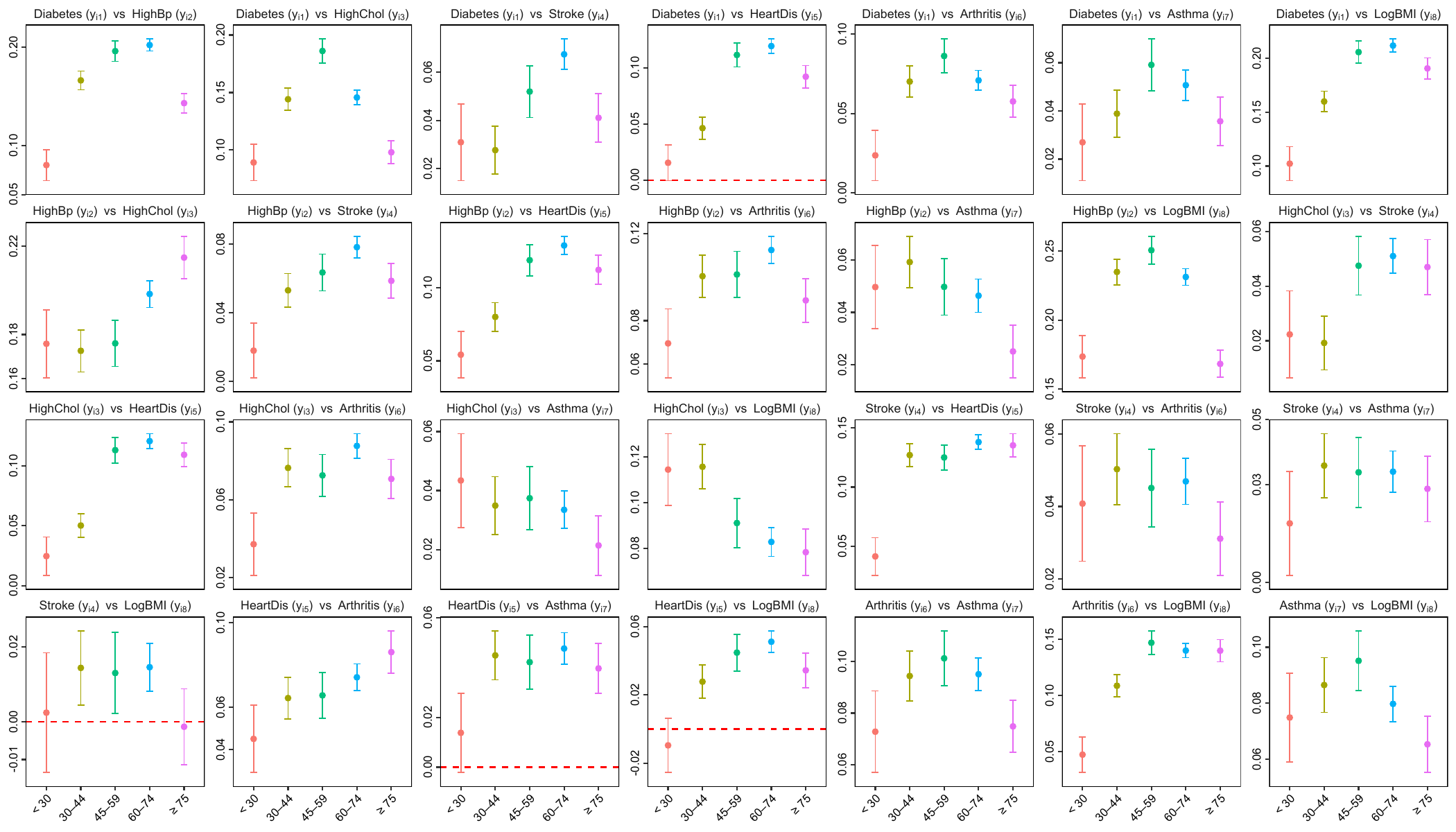}
    \caption{Estimates and 95\% confidence intervals of pairwise correlations of standardized residuals $(y_{ik}-\widehat{\mathbb E}( y_{ik}))/\widehat{sd}(y_{ik})$, $k=1,\dots,m$, where $\widehat{\mathbb E}( y_{ik})$ and $\widehat{sd}(y_{ik})$ denote the estimated mean and standard deviation of $y_{ik}$, respectively.
    }
    \label{fig:motiv_corr}
\end{figure}
%---------------------------------------

While the preliminary analysis is informative, it has several limitations.
First, discretizing age into a small number of groups may obscure smooth age-dependent variation and lead to a loss of information.
Moreover, the marginal models are fitted separately for each response variable. The dependence among the responses is therefore assessed only through correlations of residuals rather than being modeled directly within a joint framework.
Lastly, the dependence structure itself may vary with age, yet the preliminary analysis does not provide a coherent way to model such covariate-dependent dependence patterns.
These limitations motivate the development of a more flexible copula regression model that allows age-dependent effects in both the marginal and dependence structures.

\subsection{Contributions}
\label{sec:cont}

Motivated by the analysis of the BRFSS 2023 data in Section~\ref{sec:mot}, we consider a covariate-dependent copula regression model.
Specifically, the parametric regression parameters are extended to varying coefficients, $t\mapsto{\boldsymbol \beta}_k(t)=( \beta_{1k}(t),\dots,  \beta_{pk}(t))^\top$, where each $\beta_{jk}:\mathcal T\rightarrow\mathbb R$ is a univariate function \citep{hastie1993varying}. In addition, the copula function is allowed to depend on the covariate through the mapping $t\mapsto C(\cdot\,;t)$. 
The varying-coefficient structure can be readily implemented using flexible nonparametric approaches, such as splines. In contrast, the covariate-dependent copula modeling, often referred to as conditional copula modeling \citep{patton2006modelling}, involves substantially more complexity \citep{gijbels2011conditional,acar2011dependence,abegaz2012semiparametric,craiu2012mixed,klein2016simultaneous}.
Because a wide range of flexible bivariate copulas, such as Archimedean copulas, is readily available, much of the literature on conditional copulas has focused on bivariate response variables, with multivariate extensions typically pursued through pair-copula constructions using vine copulas \citep{aas2009pair,acar2012beyond,stoeber2013simplified,zhang2018vine}.
However, such constructions become increasingly complex and difficult to interpret as dimensionality grows. In addition, the pair-copula decomposition is not unique, so the resulting inference depends heavily on the user-specified decomposition setup \citep{czado2022vine}.
Another approach to multivariate modeling is to directly employ suitable multivariate copulas, such as Gaussian or $t$ copulas \citep{demarta2005t}, allowing the dependence structure to vary with covariate $t_i$. This approach also presents challenges because the covariate-dependent copula parameters must satisfy constraints such as the positive definiteness of scale matrices. For example, when a Gaussian copula has a correlation matrix that varies with $t_i$ (see Section~\ref{sec:copula} for more details), positive definiteness must hold over the entire domain $\mathcal T$, which becomes increasingly difficult to enforce as $m$ grows.
A few decomposition-based strategies, such as using the Cholesky decomposition or factor models, have been proposed to address this issue \citep{wilson2011generalised,fox2015bayesian,lan2019flexible}. However, such decompositions impose global structural constraints on the dependence pattern because pairwise correlations are determined through shared latent factors or triangular matrix representations, which may be restrictive when dependence relationships vary heterogeneously across response pairs.
Moreover, the desirable properties of the Gaussian copula, such as its connection to Kendall's and Spearman's rank correlations, no longer hold when some marginals are discrete \citep{geenens2020copula}. Because our primary interest lies in incorporating heterogeneous marginals, the appeal of such parametric copula structures is diminished, motivating the development of a more flexible and computationally efficient nonparametric framework.

In this study, we propose a Bayesian nonparametric approach to modeling conditional copula regression for multivariate responses.
Our first goal is to incorporate varying coefficients, $t\mapsto{\boldsymbol \beta}_k(t)$, to capture how the effects of the explanatory variables $\mathbf x_i$ vary with an additional covariate $t_i$ in the marginal distributions, leading to the expression $y_{ik}\sim F_k(\cdot\,;\mathbf x_i,t_i)$, $k=1,\dots,m$. To achieve this, we employ Bayesian knot selection with suitably designed natural cubic spline (NCS) basis functions \citep{smith1996basis,kang2024gam}, which allows local adaptation to varying smoothness.
The second goal is to develop a covariate-dependent copula model $t\mapsto C(\cdot\,;t)$ that accommodates a dependence structure varying with $t_i$. To this end, we construct an infinite mixture of Gaussian copulas with covariate-dependent mixture weights defined through the probit stick-breaking process (PSBP) \citep{chung2009nonparametric,rodriguez2011nonparametric}. As established theoretically in Section~\ref{sec:app} and empirically in Section~\ref{sec:simmis}, this PSBP framework provides greater flexibility than a single parametric family such as the Gaussian copula with varying correlations and enables multivariate conditional copula modeling without relying on pair-copula constructions, as it avoids explicit constraints on copula parameters.
Although several approaches to finite and infinite mixtures of Gaussian copulas have been proposed \citep{wu2015bayesian,tewari2023estimation}, to the best of our knowledge, this is the first work to incorporate covariate dependence using the PSBP.
Similar to the varying coefficients for the marginals, to ensure flexible modeling of the varying weights as a function of $t_i$, we employ knot selection with NCS basis functions. The resulting joint model specification is written as 
$\mathbf y_i\sim C(F_1(y_{i1};\mathbf x_i,t_i),\dots,F_m(y_{im};\mathbf x_i,t_i);t_i)$, $i=1,\dots,n$. The proposed method is implemented in the R package \texttt{bnpmccr}, available at \url{https://github.com/YujinJeong0202/bnpmccr}.

The remainder of this paper is organized as follows. Section~\ref{sec2} presents the proposed modeling framework, including the marginal specification with varying coefficients and the covariate-dependent copula construction. Section~\ref{sec3} describes the prior specification for both the marginal and dependence components. Section~\ref{sec4} develops the posterior sampling procedure based on Markov chain Monte Carlo (MCMC). Section~\ref{sec5} evaluates the proposed method through numerical studies, including an illustrative example and a copula misspecification study. Section~\ref{sec6} applies the proposed model to the BRFSS 2023 data. Section~\ref{sec7} concludes with a discussion. Additional details on posterior computation, likelihood derivations, and theoretical proofs are provided in the supplementary material.

\section{Modeling Framework
%: Conditional Copula Regression with Varying Coefficients
}\label{sec2}

In this section, we present the modeling framework for the proposed method. The model specification consists of two key components: modeling the marginal distributions as functions of the explanatory variables $\mathbf x_i$ and the covariate $t_i$, and defining a copula model indexed by $t_i$ to capture the covariate-dependence structure.

\subsection{Marginal Modeling with Varying Coefficients}

First, we assume that the marginal distribution of $y_{ik}$ is specified by the varying linear predictor $\mathbf x_i^\top {\boldsymbol \beta}_k(t_i)$, where ${\boldsymbol \beta}_k:\mathcal T\rightarrow \mathbb R^p$, together with additional Euclidean parameters ${\boldsymbol\xi}_k$.
This parameterization provides a flexible interpretation of how the effects of the explanatory variables $\mathbf x_i$ change with the covariate $t_i$. The additional parameters ${\boldsymbol\xi}_k$ remain constant with respect to $\mathbf x_i$ and $t_i$. We write the resulting expression of the law of $y_{ik}$ as
$$  
{y}_{ik} \sim F_k(\cdot \,;\mathbf x_i,t_i)=F_k(\cdot \,;{\boldsymbol \beta}_k,{\boldsymbol\xi}_k,\mathbf x_i,t_i).
$$
There are numerous ways to relate the varying linear predictor $\mathbf x_i^\top {\boldsymbol \beta}_k(t_i)$ to the marginal distribution. 
It is often directly linked to the primary quantity of interest, such as the expected value for continuous $y_{ik}$ or the probabilities for discrete $y_{ik}$. 

A convenient yet versatile approach is the generalized linear model (GLM) framework, which defines $g_k(\mathbb E(y_{ik})) = \mathbf x_i^\top {\boldsymbol \beta}_k(t_i)$ with a known link function $g_k$, making our joint model correspond to the vector GLM \citep{song2009joint}.
For example, in the case of a binary response variable $y_{ik}\in\{0,1\}$, we specify $\mathbb E(y_{ik}) = \Phi(\mathbf x_i^\top {\boldsymbol \beta}_k(t_i))$ or  $\mathbb E(y_{ik}) = \text{expit}(\mathbf x_i^\top {\boldsymbol \beta}_k(t_i))$ using a probit or logistic model, where $\Phi$ denotes the distribution function of the standard normal distribution, and $\text{expit}(x)=e^x/(1+e^x)$ denotes the standard logistic function. 
We extend the logit model to handle binomial data with multiple trials, while using the probit model only for binary responses owing to its convenient data augmentation for inference (see Section~\ref{sec3} for details).
A count response $y_{ik}\in\{0,1,\dots\}$ with overdispersion can be modeled with a negative binomial distribution.

Beyond the single-parameter GLM framework, our marginal specification offers further flexibility. For example, an ordinal response variable $y_{ik}\in\{1,\dots,B_k\}$ can be incorporated via an ordinal probit model as $F_k(b ;{\boldsymbol \beta}_k,{\boldsymbol\xi}_k,\mathbf x_i,t_i)=\Phi(\xi_{kb}-\mathbf x_i^\top {\boldsymbol \beta}_k(t_i))$, $b=1,\dots,B_k-1$, where ${\boldsymbol\xi}_k=(\xi_{k1},\dots,\xi_{k,B_k-1})^\top$ represents the threshold parameters. By including a vector of ones in $\mathbf x_i$ to incorporate a varying intercept, we fix $\xi_{k1}=0$ to ensure identifiability. 
Table~\ref{tab:dist} summarizes the marginal specifications considered in this study. To fully specify the marginal distribution, it remains necessary to model the nonparametric functions ${\boldsymbol \beta}_k$ and to assign suitable priors to both ${\boldsymbol \beta}_k$ and ${\boldsymbol \xi}_k$. Details of these modeling components are provided in Section~\ref{sec:prior1}.

\begin{table}[t!]
\caption{Summary of the marginal parameterizations used in this study. The parameter $\xi_k$ represents the variance in a Gaussian marginal, the shape parameter in a gamma marginal, and the size parameter in a negative binomial marginal, while $\xi_{kb}$ denotes the $b$th threshold parameter in an ordinal probit model. For a negative binomial marginal, to recover the common parameterization in which $y_{ik}$ denotes the number of failures given $\xi_k$ successes, the success probability is expressed as $\xi_k/(\xi_k+\mathbb E(y_{ik}))$.}
\centering
\resizebox{1\textwidth}{!}{%
\begin{tabular}{lll}
\toprule
Response Type              & Model             & Parameterization \\
\midrule
Continuous on $\mathbb R$  & Gaussian          & $\mathbb E(y_{ik}) = \mathbf x_i^\top {\boldsymbol \beta}_k(t_i)$ , $\text{Var}(y_{ik})=\xi_k$\\
Continuous on $(0,\infty)$ & Gamma             &    $\mathbb E(y_{ik}) =\exp(\mathbf x_i^\top {\boldsymbol \beta}_k(t_i))$, $\text{Var}(y_{ik}) = \mathbb E(y_{ik})^2/\xi_k$ \\
Binary                     & Probit            &    $\mathbb E(y_{ik}) = \Phi(\mathbf x_i^\top {\boldsymbol \beta}_k(t_i))$ \\
Binomial with $B_k$ trials                  & Logit                 & $\mathbb E(y_{ik}) = B_k \text{expit}(\mathbf x_i^\top {\boldsymbol \beta}_k(t_i))$  \\
Count                      & Negative binomial &     $\mathbb E(y_{ik}) = \exp(\mathbf x_i^\top {\boldsymbol \beta}_k(t_i))$, $\text{Var}(y_{ik}) = \mathbb E(y_{ik}) + \mathbb E(y_{ik})^2/\xi_k$    \\
Ordinal with $B_k$ categories                   & Ordinal probit    &  $F_k(b ;{\boldsymbol \beta}_k,{\boldsymbol\xi}_k,\mathbf x_i,t_i)=\Phi(\xi_{kb}-\mathbf x_i^\top {\boldsymbol \beta}_k(t_i))$, $b=1,\dots,B_k-1$\\
\bottomrule
\end{tabular}
}
\label{tab:dist}
\end{table}

\subsection{Dependence Modeling with Conditional Copula}
\label{sec:copula}

The second key modeling component concerns the dependence structure of the response vector $\mathbf y_i$. We adopt a Bayesian nonparametric framework based on a covariate-dependent infinite mixture of Gaussian copulas. Let $\mathcal Q^m(\lambda_0)$ denote the space of $m\times m$ correlation matrices whose minimum eigenvalue is bounded below by $\lambda_0>0$. For $\mathbf R_0\in \mathcal Q^m(\lambda_0)$, define the Gaussian copula
$C_G(u_1,\dots,u_m;\mathbf R_0)=\Phi_{\mathbf R_0}[\Phi^{-1}(u_1),\dots,\Phi^{-1}(u_m)]$, where $\Phi_{\mathbf V}:\mathbb R^m\rightarrow [0,1]$ denotes the distribution function of an $m$-dimensional zero-mean normal distribution with covariance matrix $\mathbf V$. We model the dependence structure by an infinite-mixture copula model
$t\mapsto C(\cdot\,;t)$ given by
\begin{align}
\begin{split}
    C(\cdot\,;t)&=C(\cdot\,; \{\pi_h,\mathbf R_h\}_{h=1}^{\infty},t) 
        = \sum_{h=1}^\infty \pi_h(t) C_G(\cdot\,;\mathbf R_h),\quad t\in\mathcal T,
\end{split}
    \label{eqn:mixcop}
\end{align}
where $\pi_h : \mathcal T\rightarrow [0,1]$ are weight functions satisfying $\sum_{h=1}^\infty \pi_h(\cdot)=1$, and $\mathbf R_h\in \mathcal Q^m(\lambda_0)$ are positive definite correlation matrices.
The construction in \eqref{eqn:mixcop} reduces to the standard Gaussian copula if $\pi_1(\cdot)=1$.
To the best of our knowledge, no previous work has incorporated covariate dependence in mixtures of Gaussian copulas; our study provides the first such contribution.

As an alternative approach to covariate-dependent copula modeling, one may consider employing a Gaussian copula mapping $t\mapsto C_G(\cdot;\tilde{\mathbf R}(t))$ with a functional correlation matrix $\tilde{\mathbf R}:\mathcal T\to \mathcal Q^m(\lambda_0)$, as discussed in Section~\ref{sec:cont}. This formulation poses challenges for modeling and inference. 
In contrast, our construction in \eqref{eqn:mixcop} enables straightforward inference through mixture modeling. Specifically, for the prior on infinite mixtures with varying weights, we adopt the PSBP \citep{rodriguez2011nonparametric}, as detailed in Section~\ref{sec:prior2}. 
Additionally, the construction in \eqref{eqn:mixcop} offers greater flexibility through its nonparametric structure. In Theorem~\ref{thm:appGC}, we show that the mapping $t\mapsto C_G(\cdot;\tilde{\mathbf R}(t))$ can be well approximated by our copula construction in \eqref{eqn:mixcop}, whereas the reverse does not hold. Moreover, beyond this restricted family, Theorem~\ref{thm:appGCM} shows that \eqref{eqn:mixcop} can also approximate kernel mixtures of Gaussian copulas well under mild conditions. See Section~\ref{sec:app} for details.

The resulting joint distribution $H(\cdot\, ;\mathbf x_i,t_i)$ of $\mathbf y_i$ is expressed as
$H(y_{i1},\dots,y_{im};\mathbf x_i,t_i)=C(F_1(y_{i1};\mathbf x_i,t_i),\dots,F_m(y_{im};\mathbf x_i,t_i);t_i)$.
An equivalent expression for $\mathbf y_i \sim H(\cdot\, ;\mathbf x_i,t_i)$ is given by
\begin{align} 
        {y}_{ik} = h^{-1}_{ik}(\tilde z_{ik}), \quad \tilde{\mathbf{z}}_i \sim \sum_{h=1}^\infty \pi_h(t_i)\mathrm{N}_m(\mathbf{0},\textbf{R}_h),
        \label{eqn:model}
\end{align}
where $\tilde{\mathbf{z}}_i=(\tilde z_{i1},\ldots,\tilde z_{im})^\top$ and $h_{ik}(\cdot) = \Phi^{-1}(F_k(\cdot \,; {\boldsymbol \beta}_k,{\boldsymbol\xi}_k,\mathbf x_i,t_i))$.
As noted by \citet{pitt2006gcr}, $h_{ik}^{-1}$ is one-to-one for a continuous $y_{ik}$, but not for a discrete $y_{ik}$. This distinction between continuous and discrete cases implies that $\tilde z_{ik}$ is an unobserved latent variable for a discrete $y_{ik}$ and should be generated within the algorithm.

\subsubsection{Approximation via Covariate-Dependent Mixtures of Gaussian Copulas}
\label{sec:app}

We now investigate the approximation properties of the infinite Gaussian copula mixture model in \eqref{eqn:mixcop}. As discussed above, one may alternatively consider the mapping $t \mapsto C_G(\cdot;\tilde{\mathbf R}(t))$ as a model for covariate-dependent copulas. Although our mixture modeling framework does not strictly contain the class $t \mapsto C_G(\cdot;\tilde{\mathbf R}(t))$, it can approximate any such Gaussian copula model arbitrarily well to any prescribed level of accuracy. The following theorem formalizes this result, where we denote the density of $C_G$ by $c_G$.

\begin{theorem} 
\label{thm:appGC}
Suppose $\mathcal T\subset\mathbb R$ is compact. For any continuous map $\tilde{\mathbf R}:\mathcal T\rightarrow \mathcal Q^m(\lambda_0)$ and any $\epsilon>0$, there exist correlation matrices $\{\mathbf R_h\}_{h=1}^H \subset \mathcal Q^m(\lambda_0)$ and continuous weights $\{\pi_h\}_{h=1}^H $ with finite $H$ such that
\begin{align*}
    \sup_{t\in\mathcal T}\bigg\lVert c_G(\cdot \,;\tilde{\mathbf R}(t))-\sum_{h=1}^H \pi_h(t)c_G(\cdot\,;\mathbf R_h)\bigg\rVert_{L^1[0,1]^m}\le \epsilon.
\end{align*}
\end{theorem}

Theorem~\ref{thm:appGC} supports our mixture framework in \eqref{eqn:mixcop} for conditional copula modeling. The reverse statement does not hold: the mapping $t\mapsto C_G(\cdot;\tilde{\mathbf R}(t))$ cannot approximate \eqref{eqn:mixcop} since the latter is not elliptical for any $t$. 

While Theorem~\ref{thm:appGC} illustrates the flexibility of our framework, it is easy to see that the class in \eqref{eqn:mixcop} approximates a much broader family. The next theorem shows that \eqref{eqn:mixcop} can approximate any kernel mixture of Gaussian copulas under mild conditions.
To proceed, let $\mathscr B$ be a $\sigma$-field on $\mathcal Q^m(\lambda_0)$. A Markov kernel $K:\mathcal T\times \mathscr B\rightarrow [0,1]$ is defined so that for any $t\in\mathcal T$, the map $\mathcal A \mapsto K(t,\mathcal A)$ is a probability measure on $\mathscr B$, and for any $\mathcal A\in\mathscr B$, the map $t\mapsto K(t,\mathcal A)$ is measurable. A kernel $K$ is called weak-Feller if and only if $t\mapsto \int_{\mathcal Q^m(\lambda_0)} f(\mathbf R) K(t,d\mathbf R)$ is continuous for every bounded continuous function $f$ \citep[][Chapter 7]{hernandez2003markov}. 
The following theorem
demonstrates the approximation power of our construction in \eqref{eqn:mixcop}.

\begin{theorem}
\label{thm:appGCM}
For a weak-Feller Markov kernel $K:\mathcal T\times \mathscr B\rightarrow [0,1]$,
define the copula density
$$
c_K(\cdot\,;t) = \int_{\mathcal Q^m(\lambda_0)} c_G(\cdot\,;\mathbf R)K(t,d\mathbf R),\quad t\in\mathcal T.
$$
For any $\epsilon>0$, there exist correlation matrices $\{\mathbf R_h\}_{h=1}^H \subset \mathcal Q^m(\lambda_0)$ and continuous weights $\{\pi_h\}_{h=1}^H $ with finite $H$ such that
\begin{align*}
    \sup_{t\in\mathcal T}\bigg\lVert c_K(\cdot \,;t)-\sum_{h=1}^H \pi_h(t)c_G(\cdot\,;\mathbf R_h)\bigg\rVert_{L^1[0,1]^m}\le \epsilon.
\end{align*}
\end{theorem}

The weak-Feller property is a mild continuity condition on the kernel 
$K$. It requires $K(t,\cdot)$ to vary continuously with $t$ in the weak sense and is satisfied by many standard constructions.
In particular, if $K(t,\cdot)=\delta_{\tilde{\mathbf R}(t)}$ for some $\tilde{\mathbf R}:\mathcal T\rightarrow \mathcal Q^m(\lambda_0)$, where $\delta_{\mathbf R_0}$ is the Dirac measure at $\mathbf R_0\in\mathcal Q^m(\lambda_0)$, then the copula density reduces to $c_K(\cdot;t)=c_G(\cdot;\tilde{\mathbf R}(t))$. If $\tilde{\mathbf R}$ is continuous, this kernel $K$ is weak-Feller, and the result of Theorem~\ref{thm:appGC} follows from Theorem~\ref{thm:appGCM}. The proofs of Theorems~\ref{thm:appGC} and~\ref{thm:appGCM} are provided in Section~\ref{app:sec2} of the supplementary material.

\section{Prior Specification}\label{sec3}

\subsection{Prior for the Marginal Distributions}
\label{sec:prior1}

The marginal distribution $F_k$ of the response vector $\mathbf{y}_i \in \mathbb{R}^m$ is fully characterized by the covariate-varying linear predictor $\mathbf{x}_i^\top \boldsymbol{\beta}_k(t_i)$ and the additional parameters $\boldsymbol{\xi}_k$.
We first discuss the prior specification for the varying coefficients  $\boldsymbol\beta_k$.
There are numerous Bayesian approaches to modeling univariate nonparametric functions, such as Gaussian process regression and Bayesian kernel regression. Among these possibilities, we use spline basis expansion with adaptive knot selection to avoid bias caused by relying on a specific knot configuration and to achieve spatial adaptation to local features \citep{smith1996basis,denison1998automatic,dimatteo2001bayesian}.
Specifically, we employ the NCS basis functions proposed by \citet{kang2024gam}, which provide a convenient framework in which knot selection is equivalent to basis selection, thereby enabling fast computation. Let $\{\tau^L,\tau^U\}$ denote the boundary knots, and let $\{\tau_1,\dots,\tau_L\}$ be the set of $L$ interior knots such that $\tau^L<\tau_1<\cdots<\tau_L<\tau^U$.
We set $\tau^L=\min_{1\le i\le n}t_i$ and $\tau^U=\max_{1\le i\le n}t_i$. The interior knots are chosen as the quantiles of the unique design points $t_i$.
The NCS basis functions defined by \citet{kang2024gam} are given by
\begin{align}
\begin{split}
N(u;\tau^L,\tau^U,\tau_{\ell})
    & \equiv \frac{(u-\tau_{\ell})^3_+-(u-\tau^U)^3_+}{\tau^U-\tau_{\ell}}- \frac{(u-\tau^L)^3_+-(u-\tau^U)^3_+}{\tau^U-\tau^L},\quad \ell=1,\ldots,L.
\end{split}
\label{eq:basisdef}
\end{align}
Although they appear different from conventional NCS basis functions, the basis functions in \eqref{eq:basisdef}, combined with the constant and linear terms, span the same NCS space while being more efficient for knot selection \citep{kang2024gam}. This is because each basis term depends only on the corresponding knot $\tau_\ell$.
To enhance stability, we define the centered basis functions as
\begin{align}
    b_{\text{lin}}(u)&=u-\frac{1}{n}\sum_{i=1}^n t_i,\\
    b_{\ell}(u) &= N(u,\tau^L,\tau^U,\tau_{\ell})-\frac{1}{n}\sum_{i=1}^n N(t_i,\tau^L,\tau^U,\tau_{\ell}),\quad \ell=1,\ldots,L.
    \label{eqn:basisset}
\end{align} 
Each varying coefficient can be expressed using this set of basis terms along with an intercept.

As discussed, we employ adaptive knot selection to prevent overfitting that may arise from using all basis terms. For the varying coefficient $\beta_{jk}$, we introduce binary latent variables
$\tilde{\boldsymbol{\gamma}}_{jk}= (\tilde \gamma_{jk1},\ldots,\tilde \gamma_{jkL})^\top$, where
$\tilde{\gamma}_{jk\ell}=1$ indicates that the basis term $b_{\ell}$ is included for $\beta_{jk}$, and $\tilde{\gamma}_{jk\ell}=0$ indicates that it is excluded. This approach is equivalent to selecting important knots $\tau_\ell$ due to the structure of the basis term in \eqref{eq:basisdef}.
The varying coefficient $\beta_{jk}$ is expressed as
\begin{align*}
     \beta_{jk}(t)  =\tilde\alpha_{jk0} +  \tilde{\boldsymbol \alpha}_{jk \backslash 0,\tilde{\boldsymbol{\gamma}}_{jk}}^\top \mathbf b_{jk,\tilde{\boldsymbol{\gamma}}_{jk}}(t),
\end{align*}
where $\mathbf b_{jk,\tilde{\boldsymbol{\gamma}}_{jk}}(t) = (b_{\text{lin}}(t),\{ b_\ell(t)\}_{\ell :\tilde{\gamma}_{jk\ell}=1})^\top\in\mathbb R^{1+|\tilde{\boldsymbol{\gamma}}_{jk}|}$ is the vector of basis terms chosen by $\tilde{\boldsymbol{\gamma}}_{jk}$, $\tilde\alpha_{jk0}\in\mathbb R$ and $ \tilde{\boldsymbol \alpha}_{jk \backslash 0,\tilde{\boldsymbol{\gamma}}_{jk}}\in\mathbb R^{1+|\tilde{\boldsymbol{\gamma}}_{jk}|}$ are the coefficients, and $|\tilde{\boldsymbol{\gamma}}_{jk}| = \sum_{\ell=1}^L \tilde{\gamma}_{jk\ell}$ denotes the number of chosen knots.
To assign a reasonable prior on the coefficients, define the matrix 
$\tilde{\mathbf{W}}_{jk,\tilde{\boldsymbol{\gamma}}_{jk}}  = [ 
    x_{1j} \mathbf b_{jk,\tilde{\boldsymbol{\gamma}}_{jk}}(t_1),\dots, x_{nj} \mathbf b_{jk,\tilde{\boldsymbol{\gamma}}_{jk}}(t_n)
]^\top\in\mathbb R^{n\times(1+|\tilde{\boldsymbol{\gamma}}_{jk}|)}$, which represent the design matrix for the varying coefficient $\beta_{jk}$ excluding the intercept term.
For each $j$ and $k$, we then assign the following prior distributions
motivated by Zellner's $g$-prior:
\begin{align}
\begin{split}
 \tilde \alpha_{jk0} & \sim {\text N}(0, \nu^2),\\
    \tilde{\boldsymbol \alpha}_{jk \backslash 0,\tilde{\boldsymbol{\gamma}}_{jk}}\mid \tilde{\boldsymbol{\gamma}}_{jk},\tilde g_{jk} & \sim {\text N}_{1+|\tilde{\boldsymbol{\gamma}}_{jk}|}\!\left(\mathbf{0}, \tilde g_{jk}(\tilde{\mathbf{W}}_{jk,\tilde{\boldsymbol{\gamma}}_{jk}}^{\top} \tilde{\mathbf{W}}_{jk,{\boldsymbol{\tilde\gamma}_{jk}}})^{-1}\right), \\
    \tilde g_{jk} &\sim \text{IG}({1/2}, {n/2}),
\end{split}
    \label{eqn:prior1}
\end{align}
where $\nu>0$ is a predetermined constant. We take $\nu$ to be sufficiently large so that the intercept is treated separately and has minimal influence on knot selection.
Similar to the $g$-prior, the prior specification in \eqref{eqn:prior1} has desirable invariance properties.
If $\nu\rightarrow\infty$, the induced prior for the linear predictor
$\mathbf x_i^\top \boldsymbol{\beta}_k(t_i)$ is invariant to the scale of $x_{ij}$
and to  nonsingular basis transformations of the form
$(1,\mathbf b_{jk,\tilde{\boldsymbol{\gamma}}_{jk}}(t))\mapsto
(a,\mathbf c+\mathbf M \mathbf b_{jk,\tilde{\boldsymbol{\gamma}}_{jk}}(t))$
for $a\ne0$ and $\mathbf M$ nonsingular, so that the choice of basis is irrelevant
provided that the constant term is treated separately.
The inverse gamma prior for $\tilde g_{jk}$ is motivated by the Zellner-Siow Cauchy prior \citep{zellner1980posterior}.
For the binary latent variables $\tilde{\boldsymbol \gamma}_{jk}$, we first assign a truncated geometric prior on its effective dimension $|\tilde{\boldsymbol \gamma}_{jk}|$ and then a uniform prior on $\tilde{\boldsymbol \gamma}_{jk}$ conditional on $|\tilde{\boldsymbol \gamma}_{jk}|$. Specifically, for each $j$ and $k$,
\begin{align*}
p(\tilde{\boldsymbol{\gamma}}_{jk}) = q_{\varpi_{jk}}(|\tilde{\boldsymbol{\gamma}}_{jk}|) {L \choose |\tilde{\boldsymbol{\gamma}}_{jk}| }^{-1}, \quad |\tilde{\boldsymbol{\gamma}}_{jk}|=0,1,\dots, L,
\end{align*}
where $\varpi_{jk}\in(0,1)$ denotes the success probability associated with the corresponding $j$ and $k$, and $q_a$ denotes the geometric density with success probability $a\in(0,1)$, truncated to $\{0,1,\dots,L\}$.

The prior specification for the additional parameters ${\boldsymbol{\xi}}_k$ depends on the marginal distribution. Specifically, $\xi_k$ denotes the variance for a Gaussian marginal, for which we assign the conventional $\textrm{IG}(\varepsilon,\varepsilon)$ prior with small $\varepsilon>0$. For negative binomial and gamma marginals, $\xi_k$ corresponds to the size parameter (or the number of successes) and the shape parameter, respectively. Given that these parameters have positive support, we assign normal priors with zero mean and large variance to $\log\xi_k$. For the threshold parameters ${\boldsymbol\xi}_k=(\xi_{k1},\dots,\xi_{k,B_k-1})^\top$ in an ordinal probit model with $B_k$ categories, we define $\xi^\ast_{kb}=\log(\xi_{k,b+1}-\xi_{kb})$ for $b=1,\dots,B_k-2$, while fixing $\xi_{k1}=0$, and assign normal priors with zero mean and large variance to $\xi^\ast_{kb}$.

\subsection{Prior for the Dependence Structure}
\label{sec:prior2}

To represent the underlying dependence structure, we specify the prior distribution for the copula mixture with varying weights in \eqref{eqn:mixcop}. We employ the PSBP proposed by \citep{rodriguez2011nonparametric} with a slight adaptation incorporating the knot selection procedure described in Section~\ref{sec:prior1}.
In this prior construction, the weights of the copula mixtures are parameterized through the probit transformation of Gaussian prior functions. Specifically, the covariate-dependent weights $\pi_h$ in \eqref{eqn:mixcop} are constructed as
\begin{align} \label{eq:probitdep}
    \pi_h(t) = \Phi( f_h(t))\prod_{\eta<h}[1-\Phi( f_{\eta}(t))],\quad h=1,\dots,\infty,
\end{align}
where $f_h:\mathbb R\rightarrow\mathbb R$ are univariate functions following Gaussian processes. This stochastic partitioning ensures that $\sum_{h=1}^\infty \pi_h(t)=1$ for any $t\in\mathbb R$. 
For practical reasons the infinite mixture model is often approximated by a finite truncation with $H$ mixture components set to a moderately large value such as $H=30$. To ensure the sum is unity, the last mixture index must satisfy $\Phi(f_H(t))=1$, and $f_h$ should be specified for the remaining $h=1,\dots,H-1$. 
 The strength of the PSBP lies in its latent variable representation for \eqref{eq:probitdep}, which enables data augmentation for an efficient updating scheme. Let $s_i$ denote the copula mixture component assigned to the $i$th observation. Then, for latent variables $z_{i\eta}^\ast \sim \mathrm{N}(f_\eta(t_i),1)$, $\eta=1,\dots,H-1$, the construction in \eqref{eq:probitdep} implies that the assignment $s_i=h$ for $h< H$ is stochastically equivalent to $z_{i h}^\ast \ge 0$ and $z_{i\eta}^\ast < 0$ for all $\eta<h$. For $s_i=H$, we only need $z_{i\eta}^\ast < 0$ for all $\eta<H$ because $\Phi(f_H(t))=1$ by construction.

Using the basis functions in \eqref{eqn:basisset} and the latent variables
$\boldsymbol{\gamma}_h^\ast = (\gamma_{h1}^\ast,\dots,\gamma_{hL}^\ast)^\top$,
where each $\gamma_{h\ell}^\ast$ is set to 1 if the corresponding knot is included and 0 otherwise,
we parameterize the functions $f_h$ as
\begin{align*}
     f_h(t)  =\alpha_{h0}^\ast +  \boldsymbol{\alpha}_{h\backslash 0,\boldsymbol{\gamma}_h^*}^{\ast\top} \mathbf b_{h,{\boldsymbol{\gamma}}_{h}^\ast}(t),
\end{align*}
where $\mathbf b_{h,{\boldsymbol{\gamma}}_{h}^\ast}(t)=(b_{\text{lin}}(t),\{ b_\ell(t)\}_{\ell :{\gamma}_{h\ell}^\ast=1})^\top\in\mathbb R^{1+|{\boldsymbol{\gamma}}_{h}^\ast|}$  is the vector of basis terms chosen by ${\boldsymbol{\gamma}}_{h}^\ast$,  $\alpha_{h0}^\ast\in\mathbb R$ and $ \boldsymbol{\alpha}_{h\backslash 0,\boldsymbol{\gamma}_h^*}^{\ast\top}\in\mathbb R^{1+|{\boldsymbol{\gamma}}_{h}^\ast|}$ are the coefficients, and $|{\boldsymbol{\gamma}}_{h}^\ast| = \sum_{\ell=1}^L {\gamma}_{h\ell}^\ast$ denotes the number of chosen knots.
Similar to Section~\ref{sec:prior1}, we define the matrix
$\mathbf{W}_{h,{\boldsymbol{\gamma}}_{h}^\ast}^{\ast} = [ 
    \mathbf b_{h,{\boldsymbol{\gamma}}_{h}^\ast}(t_1),\dots, \mathbf b_{h,{\boldsymbol{\gamma}}_{h}^\ast}(t_n)
]^\top\in\mathbb R^{n\times (1+|{\boldsymbol{\gamma}}_{h}^\ast|)}$, and elicit the prior distributions as
\begin{align}
\begin{split}
    \alpha^*_{h0} &\sim \text{N}(0,1), \\ 
    \boldsymbol{\alpha}^*_{h\backslash 0,\boldsymbol{\gamma}_h^*}\mid \boldsymbol{\gamma}^*_h, g_h^* &\sim \text{N}_{1+|\boldsymbol{\gamma}^*_h|}(\mathbf{0},g^*_h(\mathbf{W}^{*\top}_{h,\boldsymbol{\gamma}^*_h}\mathbf{W}^*_{h,\boldsymbol{\gamma}^*_h})^{-1}), \\
    g^*_h &\sim \text{IG}(1/2, n/ 2), \\
     p(\boldsymbol{\gamma}^*_{h}) &\propto q_{\varpi_h}(|\boldsymbol{\gamma}^*_h|) {L \choose |\boldsymbol{\gamma}^*_h| }^{-1},\quad |\boldsymbol{\gamma}^*_h|=0,1,\dots,L,
\end{split}
\label{eqn:prior2}
\end{align}
where $\varpi_h\in(0,1)$.
The Gaussian priors in \eqref{eqn:prior2} are conjugate to the latent variables $z_{i\eta}^\ast$, enabling efficient posterior computation.
The key difference from \eqref{eqn:prior1} is that the prior for $\alpha_{h0}^\ast$ is set to the standard normal distribution with unit variance. This choice ensures that the baseline case with no predictors aligns with the standard Dirichlet process prior with unit concentration \citep{chung2009nonparametric}.

It remains to specify the prior distributions for the correlation matrices $\mathbf R_h$, $h=1,\dots, H$. Following \citet{alexopoulos2021bvsgcr}, we adopt the data augmentation strategy proposed by \citet{talhouk2012corr}, which facilitates efficient posterior sampling as discussed in Section~\ref{sec3}. To this end, we parameterize 
$\boldsymbol{\Sigma}_h =\mathbf{D}_h\mathbf{R}_h\mathbf{D}_h$, where $\mathbf{D}_h=\text{diag}(d_{h1},\dots,d_{hm})$ is an expansion parameter with $d_{hk}>0$. We then place an inverse Wishart prior on $\boldsymbol{\Sigma}_h$, which provides a conditionally conjugate distribution; that is, $\boldsymbol{\Sigma}_h\sim \text{IW}_\kappa(\mathbf I_m)$, where $\text{IW}_\kappa(\boldsymbol\Psi)$ denotes an inverse wishart distribution with $\kappa\ge m$ degrees of freedom and a $m\times m$ scale matrix $\boldsymbol\Psi$.
We assign a prior on $\mathbf R_h$ as
$$
p(\mathbf R_h)\propto |\mathbf R_h|^{-(\kappa+m+1)/2}\!\left(\prod_{k=1}^m (\mathbf R_h^{-1})_{k,k}\right)^{-\kappa/2},
$$ where $(\mathbf R_h^{-1})_{k,k}$ is the $k$th diagonal entry of $\mathbf R_h^{-1}$. This prior is motivated by the fact that it leads to the marginally uniform prior when $\kappa=m+1$ \citep{barnard2000corr}.
The corresponding induced prior for $\mathbf{D}_h$ is given by
$$
d_{hk}^2 \mid \mathbf R_h \sim \text{IG}(\kappa/2,(\mathbf R_h^{-1})_{k,k}/2),\quad k=1,\dots, m.
$$
 
\section{Posterior Sampling via Markov Chain Monte Carlo}\label{sec4}

In this section, we describe the MCMC procedure for sampling from the joint posterior distribution.
For the $i$th response vector $\mathbf y_i\in(y_{i1},\ldots,y_{im})^\top\in\mathbb{R}^m$, we let $\mathbf{y}=\{\mathbf{y}_i\}_{i=1}^n$ denote the collection of all observations. 
For each $k=1,\dots,m$, we define $\tilde{\boldsymbol{\gamma}}_k=(\tilde{\boldsymbol{\gamma}}_{1k},\ldots,\tilde{\boldsymbol{\gamma}}_{pk})^\top \in \mathbb R^{pL}$, $\tilde{\boldsymbol{\alpha}}_{k,\tilde{\boldsymbol{\gamma}}_k} = (\tilde \alpha_{1k0},\tilde{\boldsymbol{\alpha}}_{1k\backslash 0,\tilde{\boldsymbol{\gamma}}_{1k}}^\top,\ldots, \tilde{\alpha}_{pk0},\tilde{\boldsymbol{\alpha}}^\top_{pk\backslash0,\tilde{\boldsymbol{\gamma}}_{pk}})^\top \in \mathbb R^{2p + \sum_j |\tilde{\boldsymbol{\gamma}}_{jk}|}$, and  $\tilde{\mathbf{g}}_k = (\tilde g_{1k},\ldots, \tilde g_{pk})^\top \in \mathbb R^p$, and their collections are denoted by $\tilde{\boldsymbol{\gamma}}=\{\tilde{\boldsymbol{\gamma}}_k\}_{k=1}^m$, $\tilde{\boldsymbol{\alpha}}_{\tilde{\boldsymbol{\gamma}}}=\{\tilde{\boldsymbol{\alpha}}_{k,\tilde{\boldsymbol{\gamma}}_k}\}_{k=1}^{m}$, and $\tilde{\mathbf{g}} = \{\tilde{\mathbf{g}}_k\}_{k=1}^m$, respectively.
The collection of additional parameters is $\boldsymbol{\xi}=\{\boldsymbol{\xi}_k\}_{k=1}^m$, and the collection of Gaussian copula latent variables is $\tilde{\mathbf{z}} = \{\tilde{\mathbf{z}}_i \}_{i=1}^n$. For the probit stick-breaking process, we define
$\boldsymbol{\gamma}^\ast=\{\boldsymbol{\gamma}_h^\ast\}_{h=1}^{H-1}$,
$\boldsymbol{\alpha}^*_{\boldsymbol{\gamma}^*}=\{\boldsymbol{\alpha}_{h,\boldsymbol{\gamma}_h^\ast}^*\}_{h=1}^{H-1}$ with $\boldsymbol\alpha_{h,\boldsymbol\gamma_h^\ast}^\ast = (\alpha_{h0}^\ast,\boldsymbol\alpha_{h\backslash0,\boldsymbol\gamma_h^\ast}^{\ast\top})^\top\in\mathbb R^{2+|\boldsymbol\gamma_h^\ast|}$, $\mathbf{g}^*=(g_1^*,\ldots,g_{H-1}^*)^\top\in \mathbb R^{H-1}$, and $\mathbf{z}^* = \{\mathbf{z}_i^*\}^n_{i=1}$ with $\mathbf{z}_i^*=(z_{i1}^*,\ldots,z_{is_i}^*)^\top\in \mathbb R^{s_i}$. 
In addition, \(\mathbf{R} = \{\mathbf{R}_h\}_{h=1}^H\), \(\mathbf{D} = \{\mathbf{D}_h\}_{h=1}^H\), and $\mathbf{s}=(s_1,\ldots,s_n)^\top$. 
The target posterior distribution is $p(\tilde{\boldsymbol{\alpha}}_{\tilde{\boldsymbol{\gamma}}},\tilde{\boldsymbol{\gamma}},\tilde{\mathbf{g}}, \boldsymbol{{\xi}}, \tilde{\mathbf{z}}, \boldsymbol{\alpha}^*_{\boldsymbol{\gamma}^*},\boldsymbol{\gamma}^*,\mathbf{g}^*,\mathbf{s},\mathbf{z}^*, \mathbf{D}, \mathbf{R}\mid \mathbf{y})$. Each step of the sampling procedure is detailed below.
Throughout, the expression $A\mid\mathrm{rest}$ denotes conditioning on all model components except for the argument represented by $A$.

\paragraph{Step 1: Update of $(\tilde{\boldsymbol{\alpha}}_{\tilde{\boldsymbol{\gamma}}},\tilde{\boldsymbol{\gamma}})$.} We employ a strategy that depends on the type of response variable. For Gaussian, binary, and ordinal responses, we use a Gaussian latent variable representation, which enables the construction of an efficient sampler. For response types that do not admit such a representation, we adopt a Metropolis–Hastings algorithm with an appropriately designed proposal distribution.
In both cases, we introduce a latent variable ${z}_{ik}$ that reparameterizes the Gaussian copula latent variable $\tilde z_{ik}$ as
\begin{align}
\label{eq:reparamet}
    z_{ik} = \delta_{ik}\tilde z_{ik}+ \mathbf{x}_i^\top {\boldsymbol{\beta}}_k (t_i),\quad i=1,\dots,n,\quad k=1,\dots,m,
\end{align}
where $\delta_{ik}$ is a predetermined scaling factor.
\begin{itemize}
    \item {\it Gaussian, binary, and ordinal responses.} For a Gaussian response variable, $z_{ik}$ equals the observation $y_{ik}$ with $\delta_{ik}=\sqrt{\xi_k}$. For binary and ordinal cases, setting $\delta_{ik}=1$ makes $z_{ik}$ a Gaussian latent variable for data augmentation \citep{albert1993binary}. In these cases, we exploit the Gaussian complete-data likelihood, which provides conditional conjugacy for $\tilde{\boldsymbol{\alpha}}_{\tilde{\boldsymbol{\gamma}}}$. The target conditional posterior density of $(\tilde{\boldsymbol{\alpha}}_{k,\tilde{\boldsymbol{\gamma}}_k},\tilde{\boldsymbol{\gamma}}_k)$ is 
\begin{align}
\begin{split}
&p(\tilde{\boldsymbol{\alpha}}_{k,\tilde{\boldsymbol{\gamma}}_k},\tilde{\boldsymbol{\gamma}}_k \mid   \mathrm{rest} )  \propto p(\mathbf{z}_{\cdot k}\mid \tilde{\boldsymbol{\alpha}}_{k,\tilde{\boldsymbol{\gamma}}_k},\tilde{\boldsymbol{\gamma}}_k, \boldsymbol{\xi}_k, \tilde{\mathbf{z}}_{\cdot \backslash k},\mathbf s, \mathbf R )p(\tilde{\boldsymbol{\alpha}}_{k,\tilde{\boldsymbol{\gamma}}_k}\mid \tilde{\boldsymbol{\gamma}}_k,\tilde{\mathbf{g}}_k)p(\tilde{\boldsymbol{\gamma}}_k).
\end{split} \label{eq:gibbs_samp}
\end{align}
This shows that the conditional posterior of $\tilde{\boldsymbol{\alpha}}_{k,\tilde{\boldsymbol{\gamma}}_k}$ given  $\tilde{\boldsymbol{\gamma}}_k$ is a Gaussian distribution. Accordingly, for each $k=1,\dots,m$, we sample each element of $\tilde{\boldsymbol{\gamma}}_k$ via collapsed Gibbs sampling after marginalizing out $\tilde{\boldsymbol{\alpha}}_{k,\tilde{\boldsymbol{\gamma}}_k}$, and then draw $\tilde{\boldsymbol{\alpha}}_{k,\tilde{\boldsymbol{\gamma}}_k}$ conditional on the updated $\tilde{\boldsymbol{\gamma}}_k$. 
Specifically, $\tilde{\boldsymbol\gamma}_{k}$ is updated by iteratively sampling $\tilde{\gamma}_{k(\ell)}$, the $\ell$th entry of $\tilde{\boldsymbol{\gamma}}_k$, with probability
\begin{align*}
&\Pr\!\left(\tilde{\gamma}_{k(\ell)}=1\mid \mathrm{rest} \setminus \tilde{\boldsymbol{\alpha}}_{k,\tilde{\boldsymbol{\gamma}}_k}  \right) \\
    &\quad=  \!\left(1+\frac{p(\mathbf{z}_{\cdot k}\mid  {\tilde{\gamma}_{k(\ell)}}=0, \tilde{\boldsymbol{ \gamma}}_{k\backslash (\ell)},\tilde{\mathbf{g}}_k, \boldsymbol{{\xi}}_k,\tilde{\mathbf{z}}_{\cdot \backslash k}, \mathbf s,\mathbf R)p(\tilde{\gamma}_{k(\ell)}=0\mid \tilde{\boldsymbol{\gamma}}_{k\backslash (\ell )})}{p(\mathbf{z}_{\cdot k}\mid {\tilde{\gamma}_{k(\ell)}}=1, \tilde{\boldsymbol{ \gamma}}_{k\backslash (\ell)},\tilde{\mathbf{g}}_k, \boldsymbol{{\xi}}_k,\tilde{\mathbf{z}}_{\cdot \backslash k}, \mathbf s,\mathbf R)p(\tilde{\gamma}_{k(\ell)}=1\mid\tilde{\boldsymbol{\gamma}}_{k\backslash (\ell )}) }\right)^{-1},
\end{align*}
where 
$
p(\mathbf{z}_{\cdot k}\mid  \tilde{\boldsymbol\gamma}_{k},\tilde{\mathbf{g}}_k, \boldsymbol{{\xi}}_k,\tilde{\mathbf{z}}_{\cdot \backslash k}, \mathbf s,\mathbf R)=\int p(\mathbf{z}_{\cdot k}\mid\tilde{\boldsymbol{\alpha}}_{k,\tilde{\boldsymbol{\gamma}}_k},\tilde{\boldsymbol{\gamma}}_k, \boldsymbol{\xi}_k, \tilde{\mathbf{z}}_{\cdot \backslash k},\mathbf s, \mathbf R )p(\tilde{\boldsymbol{\alpha}}_{k,\tilde{\boldsymbol{\gamma}}_k}\mid\tilde{\boldsymbol{\gamma}}_k,\tilde{\mathbf{g}}_k)d \tilde{\boldsymbol{\alpha}}_{k,\tilde{\boldsymbol{\gamma}}_k}
$,
whose explicit form is provided in \eqref{eq:lhd_gamma} of the supplementary material. Conditional on the updated $\tilde{\boldsymbol\gamma}_{k}$, we draw $\tilde{\boldsymbol{\alpha}}_{k,\tilde{\boldsymbol{\gamma}}_k}$ from the Gaussian distribution:
\begin{align}
\tilde{\boldsymbol{\alpha}}_{k,\tilde{\boldsymbol{\gamma}}_k}\mid\mathrm{rest}\sim\text{N}_{2p + \sum_j |\tilde{\boldsymbol{\gamma}}_{jk}|}(\boldsymbol{\mu}_{\tilde{\boldsymbol{\alpha}}_k},\mathbf{V}_{\tilde{\boldsymbol{\alpha}}_k}),
    \label{eqn:gaussalpha}
\end{align}
with parameters $
\mathbf{V}_{\tilde{\boldsymbol{\alpha}}_k}= (\overline{\mathbf{W}}_{k,\tilde{\boldsymbol{\gamma}}_k}^\top \boldsymbol{\Psi}_k\overline{\mathbf{W}}_{k,\tilde{\boldsymbol{\gamma}}_k}  + \mathbf{V}_{k,0}^{-1})^{-1}$ and $
\boldsymbol{\mu}_{\tilde{\boldsymbol{\alpha}}_k}= \mathbf{V}_{\tilde{\boldsymbol{\alpha}}_k}\overline{\mathbf{W}}_{k,\tilde{\boldsymbol{\gamma}}_k}^\top({\boldsymbol{\zeta}}_k+ \boldsymbol{\Psi}_k \mathbf{z}_{\cdot k})$,
where 
\begin{align*}
\overline{\mathbf{W}}_{k,\tilde{\boldsymbol{\gamma}}_k} &= [\mathbf x_{\cdot 1},\tilde{\mathbf{W}}_{1k,\tilde{\boldsymbol{\gamma}}_{1k}},\ldots,\mathbf x_{\cdot p},\tilde{\mathbf{W}}_{pk,\tilde{\boldsymbol{\gamma}}_{pk}} ]\in\mathbb R^{n\times (2p + \sum_j |\tilde{\boldsymbol{\gamma}}_{jk}|)},\\
\boldsymbol{\Psi}_k &= \text{diag}(\delta_{1k}^{-2}(\mathbf{R}^{-1}_{s_1})_{k,k}, \ldots, \delta_{nk}^{-2}(\mathbf{R}^{-1}_{s_n})_{k,k})\in\mathbb R^{n\times n},\\
{\boldsymbol{\zeta}}_k &= (\delta_{1k}^{-1}(\mathbf{R}_{s_1}^{-1})_{k,\backslash k} \tilde{\mathbf z}_{1,\backslash  k}, \ldots, \delta_{nk}^{-1}(\mathbf{R}_{s_n}^{-1})_{k,\backslash k} \tilde{\mathbf z}_{n,\backslash k})^\top\in\mathbb R^{n},
\end{align*} 
$\mathbf x_{\cdot j}=(x_{1j},\ldots,x_{nj})^\top\in\mathbb R^n$,
 $(\mathbf R_{s_i}^{-1})_{k,\backslash k}$ denotes the $k$th row of $\mathbf R_{s_i}^{-1}$ with its $k$th entry removed,
$\tilde{\mathbf z}_{i,\backslash k}$ denotes $\tilde{\mathbf z}_i$ without its $k$th component,
and $\mathbf V_{k,0}$ denotes the prior covariance of $\tilde{\boldsymbol{\alpha}}_{k,\tilde{\boldsymbol{\gamma}}_k}$ induced from \eqref{eqn:prior1}.
The derivation is provided in Section~\ref{app:sec1} of the supplementary material.

\item {\it Binomial, negative binomial, and gamma responses.} As conditional conjugacy does not hold and marginalization is intractable, we employ a Metropolis-Hastings step to jointly update $(\tilde{\boldsymbol{\alpha}}_{k,\boldsymbol{\tilde{\gamma}}_k}, \tilde{\boldsymbol{\gamma}}_k, \tilde{\mathbf{z}}_{\cdot k})$. (Although the primary goal of this step is to update $(\tilde{\boldsymbol{\alpha}}_{k,\boldsymbol{\tilde{\gamma}}_k}, \tilde{\boldsymbol{\gamma}}_k)$, we update $\tilde{\mathbf{z}}_{\cdot k}$ jointly because our proposal distribution is constructed based on $\tilde{\mathbf{z}}_{\cdot k}$, which in turn depends on $\tilde{\boldsymbol{\alpha}}_{k,\boldsymbol{\tilde{\gamma}}_k}$.) The target conditional posterior density is
\begin{align} 
\begin{split}
p(\tilde{\boldsymbol{\alpha}}_{k,\tilde{\boldsymbol{\gamma}}_k},\tilde{\boldsymbol{\gamma}}_k, \tilde{\mathbf{z}}_{\cdot k}\mid \mathrm{rest}) 
& \propto  p(\mathbf{y}_{\cdot k}\mid \tilde{\boldsymbol{\alpha}}_{k,\tilde{\boldsymbol{\gamma}}_k},\tilde{\boldsymbol{\gamma}}_k, \boldsymbol{\xi}_k, \tilde{\mathbf{z}}_{\cdot k})p(\tilde{\mathbf{z}}_{\cdot k} \mid\tilde{\mathbf{z}}_{\cdot \backslash k},\mathbf s, \mathbf R )p(\tilde{\boldsymbol{\alpha}}_{k,\tilde{\boldsymbol{\gamma}}_k}\mid\tilde{\boldsymbol{\gamma}}_k,\tilde{\mathbf{g}}_k)p(\tilde{\boldsymbol{\gamma}}_k),
\label{eq:MH_target}
\end{split}
\end{align} 
 where $p(\mathbf{y}_{\cdot k}\mid \tilde{\boldsymbol{\alpha}}_{k,\tilde{\boldsymbol{\gamma}}_k},\tilde{\boldsymbol{\gamma}}_k, \boldsymbol{\xi}_k, \tilde{\mathbf{z}}_{\cdot k})$ is omitted for continuous responses because of the one-to-one correspondence  between  $\mathbf{y}_{\cdot k}$ and $\tilde{\mathbf{z}}_{\cdot k}$. 
A proposal $\tilde{\boldsymbol{\gamma}}_k^\star$ of $\tilde{\boldsymbol{\gamma}}_k$ is generated by randomly selecting one of its entries and proposing a binary value with equal probability, which yields a symmetric proposal density $q(\tilde{\boldsymbol{\gamma}}_k^\star\mid\tilde{\boldsymbol{\gamma}}_k)$.
Given the proposed $\tilde{\boldsymbol{\gamma}}_k^\star$, the proposal distribution of the 
coefficients is constructed as a normal distribution by exploiting the latent variable 
representation in \eqref{eq:reparamet},  namely,
\begin{align}
q(\tilde{\boldsymbol{\alpha}}_{k,\tilde{\boldsymbol{\gamma}}^\star}^\star\mid\tilde{\boldsymbol{\gamma}}_k^\star,\boldsymbol{\xi}_k,\tilde{\mathbf{g}}_k, \mathbf{z}_{\cdot k}, \tilde{\mathbf{z}}_{\cdot \backslash k},\mathbf s,\mathbf R) & \propto p(\mathbf{z}_{\cdot k}\mid \tilde{\boldsymbol{\alpha}}_{k,\tilde{\boldsymbol{\gamma}}_k^\star}^\star, \tilde{\boldsymbol{\gamma}}_k^\star,\boldsymbol{\xi}_k,\tilde{\mathbf{z}}_{\cdot \backslash k},\mathbf s,\mathbf R)p(\tilde{\boldsymbol{\alpha}}_{k,\tilde{\boldsymbol{\gamma}}_k^\star}^\star\mid\tilde{\boldsymbol{\gamma}}_k^\star,\tilde{\mathbf{g}}_k).
\label{eq:marginaldens_z}
\end{align}
This corresponds to the density of the Gaussian distribution in \eqref{eqn:gaussalpha}.  The proposal variance parameter $\delta_{ik}^2$ is adjusted during the burn-in period to achieve a reasonable acceptance rate and is fixed thereafter. Note that $\mathbf{z}_{\cdot k}$ depends on $(\tilde{\boldsymbol{\alpha}}_{k,\boldsymbol{\tilde{\gamma}}_k}, \tilde{\boldsymbol{\gamma}}_k, \tilde{\mathbf{z}}_{\cdot k})$, and therefore must be taken into account when computing the acceptance probability. Finally, we propose $\tilde{\mathbf z}_{\cdot k}^\star$ from its target conditional distribution, as described in Step~\ref{par:tildez-update}. The proposed triple $(\tilde{\boldsymbol{\alpha}}_{k,\boldsymbol{\tilde{\gamma}}_k}^\star, \tilde{\boldsymbol{\gamma}}_k^\star, \tilde{\mathbf{z}}_{\cdot k}^\star)$ is accepted with probability
\begin{align}
\label{eq:acceptratio}
 1\land \frac{p(\mathbf{y}_{\cdot k} \mid\tilde{\boldsymbol{\alpha}}_{k,\tilde{\boldsymbol{\gamma}}_k^\star}^\star,\tilde{\boldsymbol{\gamma}}_k^\star,  \boldsymbol{{\xi}}_k,\tilde{\mathbf{z}}_{\cdot \backslash k},\mathbf s,\mathbf R)p(\tilde{\boldsymbol{\alpha}}_{k,\tilde{\boldsymbol{\gamma}}_k^\star}^\star\mid\tilde{\boldsymbol{\gamma}}_k^\star,\tilde{\mathbf{g}}_k)
p(\tilde{\boldsymbol{\gamma}}_k^\star) q(\tilde{\boldsymbol{\alpha}}_{k,\tilde{\boldsymbol{\gamma}}_k}\mid\tilde{\boldsymbol{\gamma}}_k,\tilde{\mathbf{g}}_k, \boldsymbol{\xi}_k,  \mathbf{z}_{\cdot k}^\star,\tilde{\mathbf{z}}_{\cdot \backslash k},\mathbf s,\mathbf R)}{p(\mathbf{y}_{\cdot k} \mid\tilde{\boldsymbol{\alpha}}_{k,\tilde{\boldsymbol{\gamma}}_k},\tilde{\boldsymbol{\gamma}}_k, \boldsymbol{{\xi}}_k,\tilde{\mathbf{z}}_{\cdot \backslash k},\mathbf s,\mathbf R)p(\tilde{\boldsymbol{\alpha}}_{k,\tilde{\boldsymbol{\gamma}}_k}\mid\tilde{\boldsymbol{\gamma}}_k,\tilde{\mathbf{g}}_k)
p(\tilde{\boldsymbol{\gamma}}_k) q(\tilde{\boldsymbol{\alpha}}^\star_{k,\tilde{\boldsymbol{\gamma}}_k^\star}\mid\tilde{\boldsymbol{\gamma}}_k^\star,\tilde{\mathbf{g}}_k, \boldsymbol{\xi}_k,  \mathbf{z}_{\cdot k},\tilde{\mathbf{z}}_{\cdot \backslash k},\mathbf s,\mathbf R)},
\end{align}
where $\mathbf{z}_{\cdot k}^\star = (z^\star_{1k},\ldots, z^\star_{nk})^\top \in \mathbb R^n$ and ${z}^\star_{ik} = \delta_{ik}\tilde{z}^\star_{ik} + \overline{\mathbf{W}}_{ik, \boldsymbol{\tilde \gamma}^\star_k}\tilde{\boldsymbol{\alpha}}^\star_{k,\tilde{\boldsymbol{\gamma}}_{k}^\star}$. Note that the term
$p(\mathbf{y}_{\cdot k} \mid\tilde{\boldsymbol{\alpha}}_{k,\tilde{\boldsymbol{\gamma}}_k},\tilde{\boldsymbol{\gamma}}_k, \boldsymbol{\xi}_k, \tilde{\mathbf{z}}_{\cdot \backslash k},\mathbf s, \mathbf R )$ denotes the conditional likelihood obtained by marginalizing out  $\tilde{\mathbf{z}}_{\cdot k}$ for discrete responses \citep{pitt2006gcr}, whereas for continuous responses it coincides with the usual conditional likelihood (see Section~\ref{app:sec1} of the supplementary material).
\end{itemize}

\paragraph{Step 2: Update of $\tilde {\mathbf g}$.} 
The conditional posterior distribution of $\tilde{{g}}_{jk}$ is inverse gamma by semi-conjugacy. For $j=1,\ldots,p$ and $k=1,\ldots,m$,
\begin{align*}
    \tilde{{g}}_{jk}\mid  \mathrm{rest} \sim \text{IG}\!\left({1\over 2}\!\left({|\boldsymbol{\tilde\gamma}_{jk}|+2}\right), {1\over 2}\!\left({\tilde{\boldsymbol{\alpha}}_{jk\backslash 0, \tilde{\boldsymbol{\gamma}}_{jk}}^\top\tilde{\mathbf{W}}^\top_{jk,\boldsymbol{\tilde \gamma}_{jk}}\tilde{\mathbf{W}}_{jk,\boldsymbol{\tilde \gamma}_{jk}}\tilde{\boldsymbol{\alpha}}_{jk\backslash 0, \tilde{\boldsymbol{\gamma}}_{jk}}+n}\right)\right).
\end{align*}

\paragraph{Step 3: Update of $ {\boldsymbol \xi}$.} 
The target conditional posterior density of $\boldsymbol{\xi}_k$ is 
\begin{align*}
    p(\boldsymbol{\xi}_k\mid \mathrm{rest}) \propto p(\mathbf{y}_{\cdot k}\mid \tilde{\boldsymbol{\alpha}}_{k,\tilde{\boldsymbol{\gamma}}_k}, \tilde{\boldsymbol{\gamma}}_k, \boldsymbol{\xi}_k, \tilde{\mathbf{z}}_{\cdot \backslash k},\mathbf{s}, \mathbf{R})p(\boldsymbol{\xi}_k),
\end{align*} for $k=1,\ldots,m$.
Because $\boldsymbol{\xi}_k$ depends on the response type, its updating rule is response-specific. 
For a Gaussian response with variance
$\xi_k$, the conditional posterior of $\xi_k$ is not inverse gamma under the copula dependence. We therefore employ an independent Metropolis–Hastings algorithm, 
using an inverse gamma proposal derived from the copula-free model.
For the other cases we employ random walk Metropolis–Hastings updates with appropriate proposal distributions. Further details are provided in Section~\ref{app:sec1} of the supplementary material.

\paragraph{Step 4: Update of $\tilde {\mathbf z}$.}\label{par:tildez-update} 

The updating scheme for the Gaussian latent variables $\tilde{\mathbf{z}}_{\cdot k}$
depends on whether $\mathbf y_{\cdot k}$ is continuous or discrete.
For a continuous response, the update is deterministic with $\tilde z_{ik}=\Phi^{-1}(F_k(y_{ik} ; {\boldsymbol{\beta}}_{k}, {\boldsymbol{\xi}}_k,\mathbf{x}_i,t_i))$ for $i=1,\ldots,n$. If the response is discrete, we need to update it by sampling from its conditional posterior distribution, which is given by  
\begin{align}\label{eq:updatetildez}
  \tilde z_{ik} \mid\mathrm{rest}
   \sim  \text{TN}_{(l_{ik},u_{ik}]}(\tilde{\mu}_{i,s_i(k,\backslash  k)}, \tilde{\sigma}^2_{s_i(k, \backslash  k)}),
\end{align}
where $\mathrm{TN}_A$ denotes a normal distribution truncated to $A$, with truncation bounds 
$l_{ik} = \Phi^{-1}\{F_k(y_{ik}-1;{\boldsymbol{\beta}}_{k}, {\boldsymbol{\xi}}_k,\mathbf{x}_i,t_i )\}$ and $u_{ik} = \Phi^{-1}\{F_k(y_{ik};{\boldsymbol{\beta}}_{k}, {\boldsymbol{\xi}}_k,\mathbf{x}_i,t_i )\}$. The parameters are given by
$\tilde{\mu}_{i,s_i(k,\backslash k)} = (\mathbf{R}_{s_i})_{k, \backslash  k}(\mathbf{R}_{s_i})^{-1}_{\backslash  k, \backslash  k}\tilde{\mathbf{z}}_{i,\backslash k}$ and $\tilde{\sigma}_{s_i(k, \backslash  k)}^2 = 1-(\mathbf{R}_{s_i})_{k, \backslash  k}(\mathbf{R}_{s_i})^{-1}_{\backslash  k,\backslash  k}(\mathbf{R}_{s_i})_{\backslash  k,k}$, where $(\mathbf{R}_{s_i})_{k,\backslash k}$ denotes the $k$th row of matrix $\mathbf{R}_{s_i}$ excluding its $k$th element, and 
$(\mathbf{R}_{s_i})^{-1}_{\backslash  k,\backslash  k}$ denotes the inverse of the submatrix of $\mathbf{R}_{s_i}$ obtained by removing its $k$-th row and $k$-th column.

\paragraph{Step 5: Update of ${\mathbf s}$.} 
The clustering membership variable $s_i$ is updated from the discrete distribution
\begin{align*}
    \Pr({s}_i= h\mid \mathrm{rest}) = {|\mathbf{R}_h|^{-1/2} \exp\!\left\{-{1\over 2}\tilde {\mathbf{z}}_i^\top\mathbf{R}_h^{-1}\tilde {\mathbf{z}}_i\right\}\pi_h(t_i)\over \sum_{\eta=1}^H |\mathbf{R}_{\eta}|^{-1/2} \exp\!\left\{-{1\over 2}\tilde {\mathbf{z}}_i^\top\mathbf{R}_{\eta}^{-1}\tilde {\mathbf{z}}_i\right\}\pi_{\eta}(t_i)},\quad h=1,\dots,H.
\end{align*}

\paragraph{Step 6: Update of $\mathbf{z}^*$.} 
The updating of $\mathbf{z}^*$ is based on a truncated normal distribution. For $i=1,\ldots,n$ and $\eta=1,\ldots,s_i$, 
\begin{gather} \label{eq:truncnorm}
    z_{i\eta}^* \mid \mathrm{rest}\sim  \begin{cases}
\text{TN}_{(0,\infty)}(f_\eta(t_i)
,1), & \mbox{if }\eta=s_i \\
\text{TN}_{(-\infty,0]}(f_\eta(t_i) ,1), & \mbox{if }\eta<s_i
\end{cases}
\end{gather}
where $f_\eta(t_i)  =\alpha_{\eta0}^\ast +  \boldsymbol{\alpha}_{\eta\backslash 0,\boldsymbol{\gamma}_\eta^\ast}^{\ast\top} \mathbf b_{\eta,{\boldsymbol{\gamma}}_{\eta}^\ast}(t_i)$.

\paragraph{Step 7: Update of $(\boldsymbol{\alpha}^*_{\boldsymbol{\gamma}^*},\boldsymbol{\gamma}^*)$.} 
The introduction of the latent variables $\mathbf{z}^*$ leads to a tractable Gaussian posterior of $\boldsymbol{\alpha}^*_
{h,\boldsymbol{\gamma}_h^*}$ given $\mathbf{z}^*$. Let $\textbf{z}_{\cdot h}^* = \{z_{ih}^*\}_{i : s_i \ge h}$.
After integrating out $\boldsymbol{\alpha}^*_{h,\boldsymbol{\gamma}_h^*}$, the Gibbs update for $\boldsymbol{\gamma}_h^*$ is given as follows. For $\ell=1,\ldots, L$ and $h=1,\ldots,H-1$, 
\begin{align*}
    \Pr\!\left(\gamma_{h\ell}^* =1 \mid \mathrm{rest} \setminus \boldsymbol{\alpha}_{h,\boldsymbol{\gamma}_h^*}^*\right)= \!\left(1 + \frac{p(\gamma^*_{h\ell}=0\mid \boldsymbol{\gamma}^*_{h\backslash \ell})\,
    p(\mathbf{z}^*_{\cdot h}\mid\gamma^*_{h\ell}=0, \boldsymbol{\gamma}^*_{h\backslash \ell},g_h^*)}
    {p(\gamma^*_{h\ell}=1\mid \boldsymbol{\gamma}^*_{h\backslash \ell})\,
    p(\mathbf{z}^*_{\cdot h}\mid\gamma^*_{h\ell}=1, \boldsymbol{\gamma}^*_{h\backslash \ell},g_h^*)}
    \right)^{-1},
\end{align*}
where
$
p(\mathbf{z}^*_{\cdot h}\mid\boldsymbol{\gamma}_h^*,g_h^*) = \int p(\mathbf{z}^*_{\cdot h}\mid \boldsymbol{\alpha}^*_{h,\boldsymbol{\gamma}_h^*},\boldsymbol{\gamma}_h^*)p(\boldsymbol{\alpha}^*_{h,\boldsymbol{\gamma}^*_h}\mid\boldsymbol{\gamma}^*_h,g^*_h) d \boldsymbol{\alpha}^*_{h,\boldsymbol{\gamma}^*_h}
$, whose closed-form expression is provided in {\eqref{app:zstar-lhd} of the supplementary material}.
Conditional on $\boldsymbol{\gamma}_h^*$, the posterior for $\boldsymbol{\alpha}^*_{h,\boldsymbol{\gamma}_h^*}$ is easily derived as a Gaussian distribution. Let $\widehat{\mathbf{W}}_{h, \boldsymbol{\gamma}_h^*}^{*}\in\mathbb R^{(\sum_{\eta \ge h} n_\eta) \times (1+|\boldsymbol{\gamma}_h^*|)}$ denote the sub-matrix of $\mathbf{W}_{h, \boldsymbol{\gamma}_h^*}^*$ whose rows are the $i$th row of $\mathbf{W}_{h, \boldsymbol{\gamma}_h^*}^*$ for those indices $i$ satisfying $s_i \ge h$, where $n_\eta=\sum_{i=1}^n \mathbb{I}(s_i=\eta)$. Then, for $h=1,\ldots,H-1$,
\begin{gather*}
\boldsymbol{\alpha}_{h,\boldsymbol{\gamma}_h^*}^*\mid\mathrm{rest}
\sim  \text{N}_{2+|\boldsymbol{\gamma}_h^*|}(\boldsymbol{\mu}_{\boldsymbol{\alpha}_h^*},\mathbf{V}_{\boldsymbol{\alpha}_h^*} ),
\end{gather*}
where
\begin{align*}
\mathbf{V}_{\boldsymbol{\alpha}^*_h}&=\begin{pmatrix}
             \sum_{\eta\ge h}n_\eta + 1 & \mathbf{1}^\top\widehat{\mathbf{W}}^*_{h,\boldsymbol{\gamma}_h^*}\\
             \widehat{\mathbf{W}}^{*\top}_{h,\boldsymbol{\gamma}_h^{*}}\mathbf{1} & \widehat{\mathbf{W}}^{*\top}_{h,\boldsymbol{\gamma}_h^{*}} \widehat{\mathbf{W}}^*_{h,\boldsymbol{\gamma}_h^*}
             + {g_h^*}^{-1}\mathbf{W}^{*\top}_{h,\boldsymbol{\gamma}_h^*}\mathbf{W}^{*}_{h,\boldsymbol{\gamma}_h^*}
         \end{pmatrix}^{-1}, \\
    \boldsymbol{\mu}_{\boldsymbol{\alpha}^*_h} &= \mathbf{V}_{\boldsymbol{\alpha}_h^*}\begin{pmatrix}
             \mathbf{1}^\top
             \\ \widehat{\mathbf{W}}^{*\top}_{h,\boldsymbol{\gamma}_h^*}
         \end{pmatrix}\mathbf{z}_{\cdot h}^*.
\end{align*}

\paragraph{Step 8: Update of $\mathbf{g}^*$.} 
The conditional posterior distribution of $g_h^\ast$ is inverse gamma: for $h=1,\ldots,H-1$,
\begin{align*}
   g_h^*\mid \mathrm{rest} \sim \text{IG}\!\left( {1\over 2}\!\left({|\boldsymbol{\gamma}_h^*|+2}\right),{1\over 2}\!\left({\boldsymbol{\alpha}_{h\backslash 0,\boldsymbol{\gamma}_h^*}^{*\top}\mathbf{W}_{h,\boldsymbol{\gamma}_h^*}^{*\top}\mathbf{W}_{h,\boldsymbol{\gamma}_h^*}^{*}\boldsymbol{\alpha}_{h\backslash 0,\boldsymbol{\gamma}_h^*}^{*} + n }\right)\right).
\end{align*}

%--------------------------------

\paragraph{Step 9: Update of $(\mathbf D,\mathbf R)$.} 
As discussed in Section~\ref{sec:prior2},
we employ a parameter expansion strategy for the correlation matrix $\mathbf R_h$. 
The conditional posterior density of $(\mathbf{D}_h,\mathbf{R}_h)$ is
\begin{align*}
p(\mathbf D_h,\mathbf R_h \mid \mathrm{rest})\propto 
p(\mathbf{D}_h \mid \mathbf{R}_h)p(\mathbf{R}_h)\prod_{i:s_i=h}p(\tilde{\mathbf{z}}_i\mid s_i,\mathbf{R}_{s_i}),
\end{align*} for $h=1,\ldots,H$.
Expanding $\tilde{\mathbf{z}}_i \sim N(0,\mathbf{R}_{s_i})$ to  $ \tilde{\mathbf{z}}^\top_i\mathbf{D}_{s_i}$  yields the semi-conjugate posterior of $\boldsymbol{\Sigma}_h$, enabling an efficient and concise update of $\mathbf{R}_h$. The sampling steps are as follows. 
\begin{enumerate}[label=\roman*)]
    \item Draw $d_{hk}^2\mid\mathbf{R}_h \overset{\text{iid}}{\sim} \text{IG}((m+1)/2,(\mathbf{R}_h^{-1})_{k,k}/2)$  for $k=1,\ldots,m$ and set $\mathbf{D}_{h} = \text{diag}(d_{h1},\ldots,d_{hm})$.
    \item Set $\tilde{\mathbf{z}}^\star_h = \{ \tilde{\mathbf{z}}_i^\top\mathbf{D}_{s_i}\}_{i:s_i=h}$ and then sample $\boldsymbol{\Sigma}_{h}\sim \text{IW}_{n_h+m+1}\!\left((\mathbf{I}_m + \tilde{\mathbf{z}}_h^{\star\top} \tilde{\mathbf{z}}_h^\star)^{-1} \right)$.
    \item Set $\mathbf{D}_h = \sqrt{\text{diag}(\boldsymbol{\Sigma}_h)}$ and calculate $\mathbf{R}_h = \mathbf{D}^{-1}_h \boldsymbol{\Sigma}_h\mathbf{D}^{-1}_h$.
    \item For each \(i\), update
$\tilde{\mathbf{z}}_{i}
=
\tilde{\mathbf{z}}^\star_{is_i}\mathbf{D}_{s_i}^{-1}$.
\end{enumerate}
Technically, this sampling scheme corresponds to marginal data augmentation \citep{meng1999seeking,van2001art}.

\section{Numerical Study}\label{sec5}
In this section, we validate the proposed method through two numerical studies.
First, we consider an illustrative example under a correctly specified data-generating process that exactly matches the modeling framework of the proposed method. This experiment is intended to examine its internal validity. In this setting, we do not include comparisons with other methods, as no directly comparable approach is currently available.
Second, to evaluate robustness to model misspecification and to assess the flexibility of the proposed covariate-dependent copula modeling, we compare our method with the approach of \citet{marra2017bivariate}, implemented in the R package \GJRM, under copula misspecification. Owing to the modeling limitations of \GJRM, however, the comparison is restricted to bivariate copulas with continuous responses.

\subsection{Illustrative Example}

\label{sec:example}

We consider numerical examples under a correctly specified data-generating mechanism to evaluate the proposed model and assess estimation accuracy.
Two datasets with sample sizes $n\in \{2{,}000,5{,}000\}$ are considered.
We set $m=6$ and assume that the $k$th response $y_{ik}$ follows the $k$th distribution listed in Table~\ref{tab:dist}, using the parameterization specified therein, with distribution-specific additional parameters.
Specifically, given the linear predictors specified below, $y_{i1}$ is Gaussian with variance $\xi_1=0.1^2$, $y_{i2}$ is gamma with shape parameter $\xi_2=10$, $y_{i3}$ is binary with a probit link, $y_{i4}$ is binomial for $10$ trials with a logit link, $y_{i5}$ is negative binomial with $\xi_5=5$ given successes, and $y_{i6}$ is ordinal probit for 4 categories, for which the cut points are $\boldsymbol{\xi}_6=(\xi_{61},\xi_{62},\xi_{63})^\top=(0,1,2)^\top$.
We set $p=2$, and the varying linear predictor is given by $\beta_{1k}(t_i)+x_i \beta_{2k}(t_i)$ for $k=1,\dots,6$, where $t_i\in[-1,1]$ and $x_i\in[-1,1]$.
For $i=1,\dots,n$, the covariate $t_i$ and the explanatory variable $x_i$ are drawn independently from $\mathrm{Unif}(-1,1)$. The varying coefficients $\beta_{1k},\beta_{2k}:[-1,1]\to\mathbb{R}$ are chosen as 
\begin{align}
\begin{split}
    \beta_{11}(t) & = {(3t+1.5)^3\over 400}+{(3t-2.5)^2\over 20e^{-3t-1.5}}\sin\!\left({(4t+1.5)^2\pi\over3}\right)\mathbbm{1}(-0.4\le t\le 0.85), \\
    \beta_{12}(t) & = \sin^3\!\left({\pi(t+2)^2\over 3}\right) +0.5,  \\
    \beta_{13}(t) & = t, \\
    \beta_{14}(t) & = \sin^3\!\left({\pi(t+1)^3\over 4}\right), \\
    \beta_{15}(t) & = 0.5(2t^5 + 3t^2 + \cos(3\pi t)-1) + 0.5, \\
    \beta_{16}(t) & = \sin({\pi(t+1)}), \\
    \beta_{21}(t) & = 0.5\sin^2(2\pi(t+1)), \\  
    \beta_{22}(t) & = -0.5(t+0.5)^2, \\
    \beta_{23}(t) & = -2t, \\
    \beta_{24}(t) & = -\sin^2(\pi(t+1)) + 0.5, \\
    \beta_{25}(t) & = \sin(0.8\pi(t+1)) - 0.5, \\
    \beta_{26}(t) & = \sin(0.5\pi(t+1)) - 0.5. 
\end{split}
\label{eqn:vcf}
\end{align}
To induce dependence across the responses, we employ a three–component mixture of Gaussian copulas with correlation matrices $(\mathbf{R}_1,\mathbf{R}_2,\mathbf{R}_3)$ and varying mixture weights $(\pi_1(t),\pi_2(t),\pi_3(t))$, where 
\begin{align*}
    (\mathbf{R}_1)_{k,k'} &=\mathbbm 1(k=k')+0.8 \mathbbm 1(k \ne k'), \\
    (\mathbf{R}_2)_{k,k'} &=\mathbbm 1(k=k')+ 0.9(-1)^{|k-k'|}  \mathbbm 1(k \ne k'), \\
    (\mathbf{R}_3)_{k,k'} &=\mathbbm 1(k=k')+ 0.3(-1)^{|k-k'|}  \mathbbm 1(k \ne k').
\end{align*}
and the weights are generated via the PSBP construction as
\begin{align*}
    \pi_1(t)& = \Phi(3t),\\
    \pi_2(t)& = \Phi(2\sin(3\pi(t+0.5)))[1-\pi_1(t)],\\
    \pi_3(t)& =1-\pi_1(t)-\pi_2(t).
\end{align*}

For the spline basis, we use $L=40$ interior knot candidates for each function. We set the prior inclusion probabilities for the number of selected bases to $\varpi_{jk}=0.2$ and $\varpi_h=0.2$, which induce reasonable prior decay. The mixture size is truncated at an upper bound of $H=30$ components. We run 10,000 MCMC iterations and discard the first 5,000 as burn-in.

%--------------------------------------------
% n = 2000
\begin{figure}[t]
    \centering
    \begin{subfigure}[t]{\linewidth}
        \centering
        \includegraphics[width=1\textwidth]{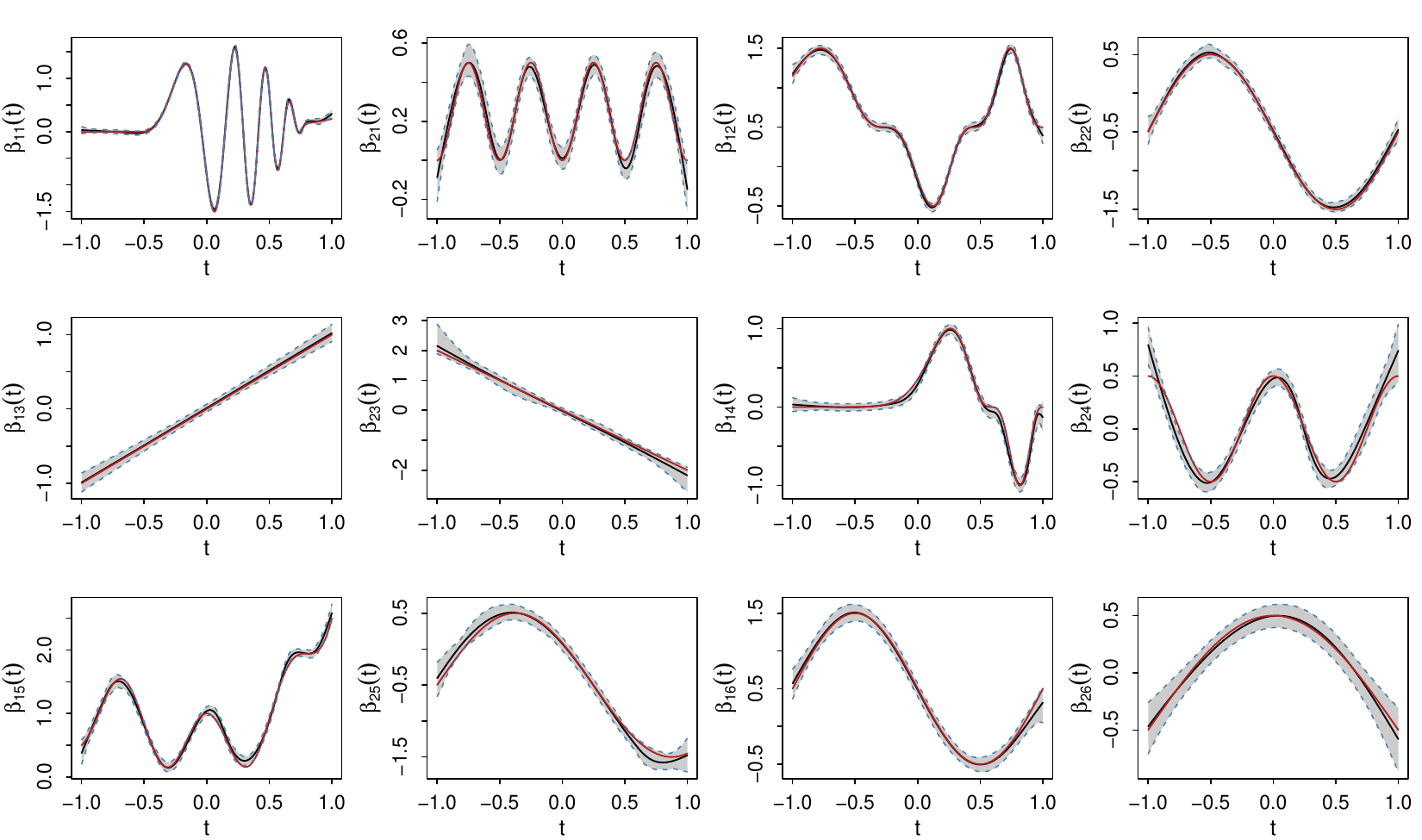}
        \caption{Results for the varying coefficients $\beta_{1k}(t)$ and $\beta_{2k}(t)$, $k=1,\dots,6$. 
        }
        \label{fig:marginal_n2000_a}
    \end{subfigure}
    \begin{subfigure}[t]{\linewidth}
        \centering
     \vspace{0.6em}
        \includegraphics[width=0.85\textwidth]{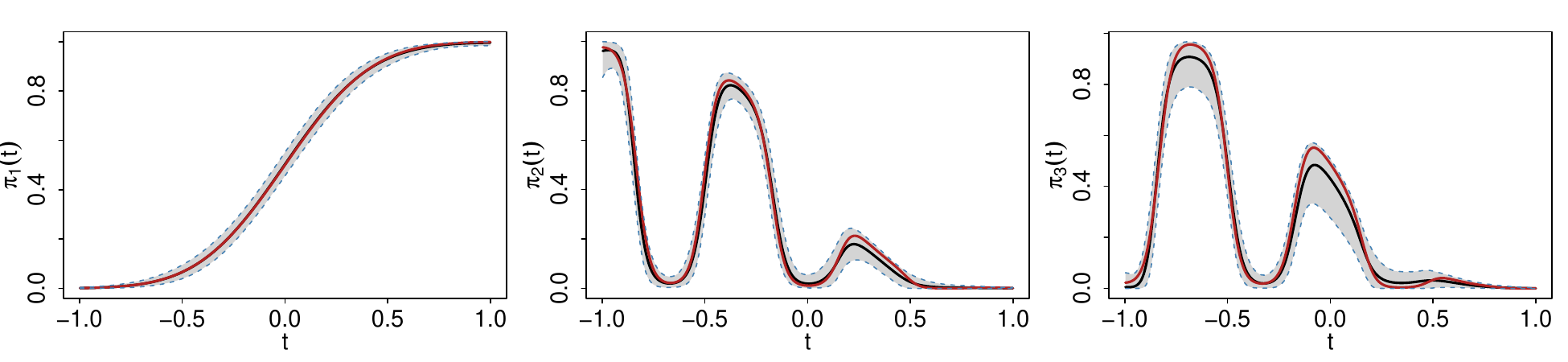}
        \caption{Results for the varying mixture weights $\pi_h(t)$, $h=1,2,3$.}
        \label{fig:pi_n2000_b}
    \end{subfigure}
    \caption{Results for $n=2{,}000$. Posterior means (black) and 95\% pointwise credible bands (gray), with the true curves shown in red.}
    \label{fig:sim_n2000}
\end{figure}
%--------------------------------------------

%--------------------------------------------
% n = 5000
\begin{figure}[t]
    \centering
    \begin{subfigure}[t]{\linewidth}
        \centering
        \includegraphics[width=1\linewidth]{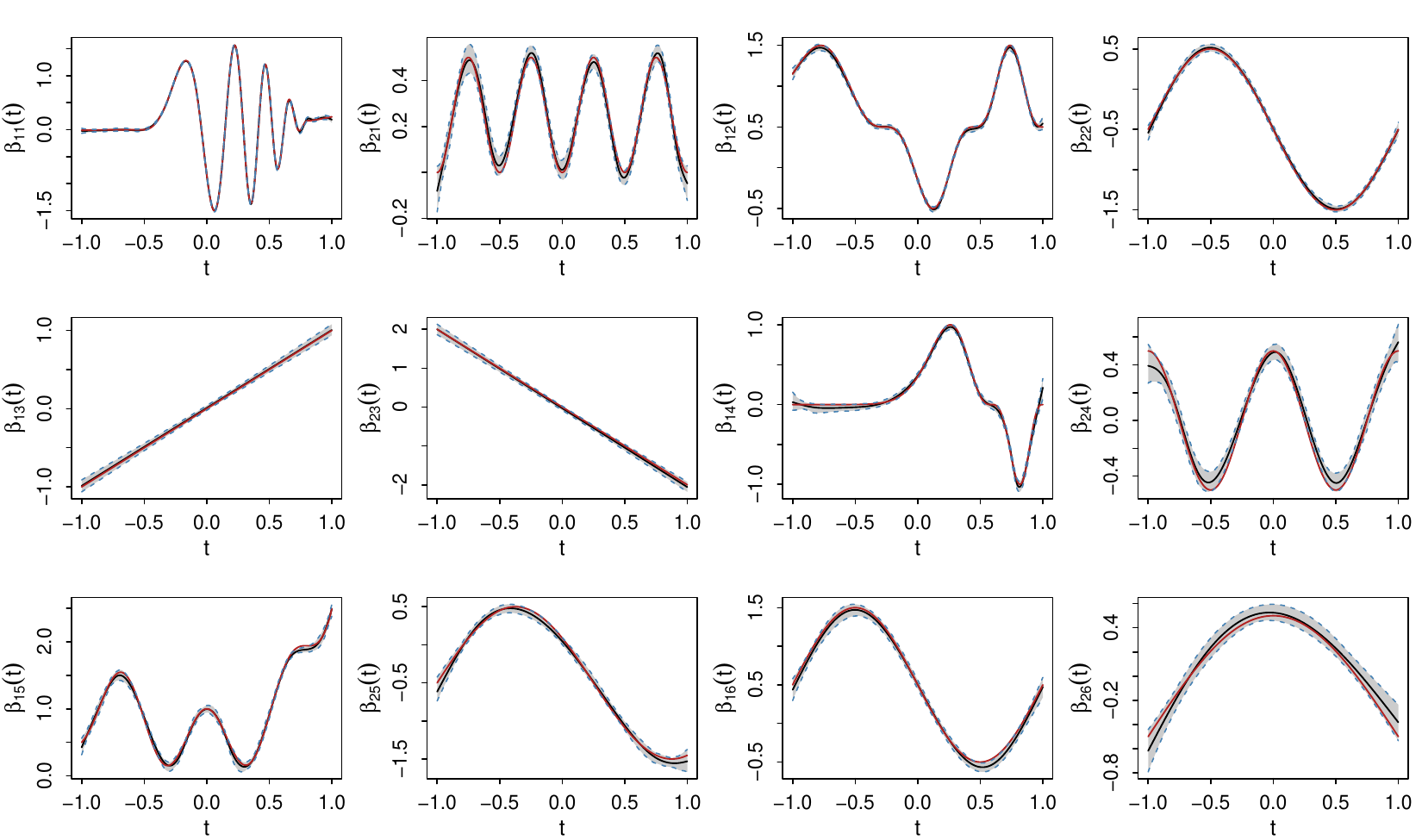}
        \caption{Results for the varying coefficients $\beta_{1k}(t)$ and $\beta_{2k}(t)$, $k=1,\dots,6$. 
       }
        \label{fig:marginal_n5000_a}
    \end{subfigure}
    \begin{subfigure}[t]{\linewidth}
        \centering
        \vspace{0.6em}
        \includegraphics[width=0.9\textwidth]{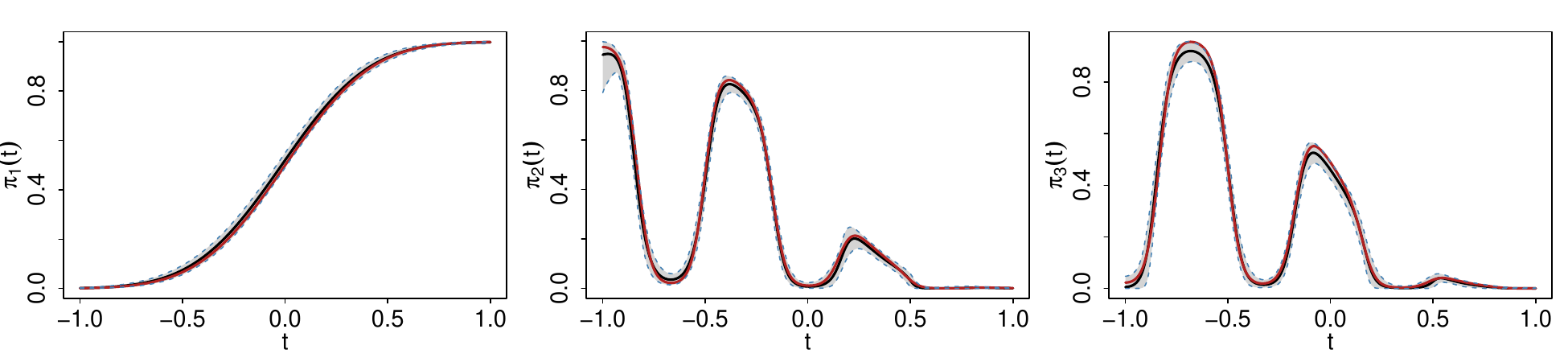}
        \caption{Results for the varying mixture weights $\pi_h(t)$, $h=1,2,3$.}
        \label{fig:pi_n5000_b}
    \end{subfigure}
    \caption{Results for $n=5{,}000$. Posterior means (black) and 95\% pointwise credible bands (gray), with the true curves shown in red.}
    \label{fig:sim_n5000}
\end{figure}
%--------------------------------------------

The results are summarized in Figures~\ref{fig:sim_n2000}--\ref{fig:sim_n5000} for $n=2{,}000$ and $n=5{,}000$, respectively.
Specifically, Figures~\ref{fig:marginal_n2000_a} and \ref{fig:marginal_n5000_a} compare the estimated varying coefficients with the true functions, while
Figures~\ref{fig:pi_n2000_b} and \ref{fig:pi_n5000_b} display the estimated mixture weights for the three clusters.
The results indicate that the proposed method identifies important knots among the candidates, thereby avoiding overfitting while achieving local adaptation. This behavior is consistent with findings in the knot-selection literature \citep[e.g.,][]{jeong2022bayesian}.
The true functions are well covered by the 95\% pointwise credible bands, supporting the validity of the proposed method.

\subsection{Simulation under Copula Misspecification}
\label{sec:simmis}

We next evaluate the flexibility of the proposed copula modeling by examining robustness to misspecification. For a meaningful comparison, we benchmark our method against the R package \GJRM. While our method captures covariate-varying dependence through a dependent mixture of Gaussian copulas, \GJRM adopts the modeling framework $t\mapsto C_G(\cdot\,;\tilde{\mathbf R}(t))$, discussed as an alternative approach in Section~\ref{sec:copula}. We consider two sample sizes, $n\in\{2{,}000,5{,}000\}$. Due to the more limited modeling scope of \GJRM, we restrict attention to bivariate continuous responses with Gaussian and gamma marginal distributions. Specifically, the data are generated in the same manner as $y_{i1}$ and $y_{i2}$ in Section~\ref{sec:example}, using additional parameters $\xi_1=0.1^2$ and $\xi_2=10$ and varying coefficients $\beta_{11}$, $\beta_{21}$, $\beta_{12}$, and $\beta_{22}$ provided in \eqref{eqn:vcf}. 
The explanatory variables $x_i$ and $t_i$ are drawn independently from $\mathrm{Unif}(-1,1)$. To generate misspecified dependence structures between $y_{i1}$ and $y_{i2}$, we consider the following four scenarios for the copula model $t\mapsto C(\cdot\,;t)$ used to generate the data.
\begin{itemize}
\item \textit{Scenario 1: $t$-copula with a functional correlation matrix}.
We consider a bivariate $t$-copula 
$$
C(\cdot\,;t)=C_{t,5}(\cdot\,;\tilde{\mathbf R}(t)),\quad
\tilde{\mathbf R}(t)_{1,2}=\sin\!\left(\frac{13\pi}{40}\sin^3\!\left({\pi(t+1)^3\over 4}\right)\right),
$$
where $C_{t,\nu}(\cdot \, ;\mathbf R)$ denotes the $t$-copula with $\nu$ degrees of freedom and correlation matrix $\mathbf R$.

\item \textit{Scenario 2: Gaussian copula with a functional correlation matrix}.
We consider a bivariate Gaussian copula 
$$
C(\cdot\,;t)=C_G(\cdot;\tilde{\mathbf R}(t)),
$$
with the same functional correlation matrix as in Scenario~1. This scenario corresponds to the correctly specified setting for \GJRM.

\item \textit{Scenario 3: Mixture of $t$-copulas with covariate-varying weights}.
We consider a two-component mixture of bivariate $t$-copulas 
$$
C(\cdot\,; t)=\sum_{h=1}^{2}\pi_h(t)\,C_{t,5}\big(\cdot;\mathbf R_h\big),\quad
\pi_1(t)=\Phi\!\left(1.5\sin^3\!\left(\frac{\pi(t+1)^3}{4}\right)\right),\quad
\pi_2(t)=1-\pi_1(t),
$$
with $({\mathbf R_1})_{1,2}=-0.8$ and $({\mathbf R_2})_{1,2}=0.8$.

\item \textit{Scenario 4: Mixture of Gaussian copulas with covariate-varying weights}.
We consider a two-component mixture of bivariate Gaussian copulas 
$$
C(\cdot\,;t)=\sum_{h=1}^{2}\pi_h(t)\,C_G\big(\cdot;\mathbf R_h\big),
$$
with the same  mixture weights $\pi_h(t)$ and correlation matrices $\mathbf R_h$ as in Scenario~3. This scenario corresponds to the correctly specified setting for the proposed method.
\end{itemize}
Other hyperparameter settings and MCMC specifications are identical to those in Section~\ref{sec:example}.
However, in contrast to Section~\ref{sec:example}, we generate 100 replicated datasets for each simulation setting
to properly account for sampling uncertainty in the data-generating mechanism.
All performance measures reported in this section are based on these 100 replications.

Because the two methods adopt different copula parameterizations, a direct comparison of copula parameters is not meaningful. To assess robustness to copula misspecification, we therefore compare pointwise estimates of the covariate-varying bivariate copula density $(u_1,u_2,t)\mapsto c(u_1,u_2;t)$ obtained from each method, where $c$ is the true copula density in Scenarios~1--4. Let $\hat c$ denote the corresponding pointwise estimate. We then define the $L^2$ estimation error for the copula density as
$
\{
\int_{-1}^1\int_0^1\int_0^1[c(u_1,u_2;t) - \hat c(u_1,u_2;t) ]^2 d u_1 d u_2 dt\}^{1/2}
$.
Similarly, we obtain pointwise estimates of the varying coefficients $\beta_{jk}$, denoted by $\hat\beta_{jk}$, and define the corresponding $L^2$ estimation error as
$
\{\int_{-1}^1[\beta_{jk}(t) - \hat\beta_{jk}(t)]^2 dt\}^{1/2}.
$
For the additional parameters $\xi_1$ and $\xi_2$, corresponding to the variance of $y_{i1}$ and the shape parameter $y_{i2}$, respectively, we define the estimation error as the absolute bias $|\xi_k-\hat \xi_k|$ for $k=1,2$, where $\hat \xi_k$ denotes the estimate.

\begin{figure}
    \centering
    \includegraphics[width=1\textwidth]{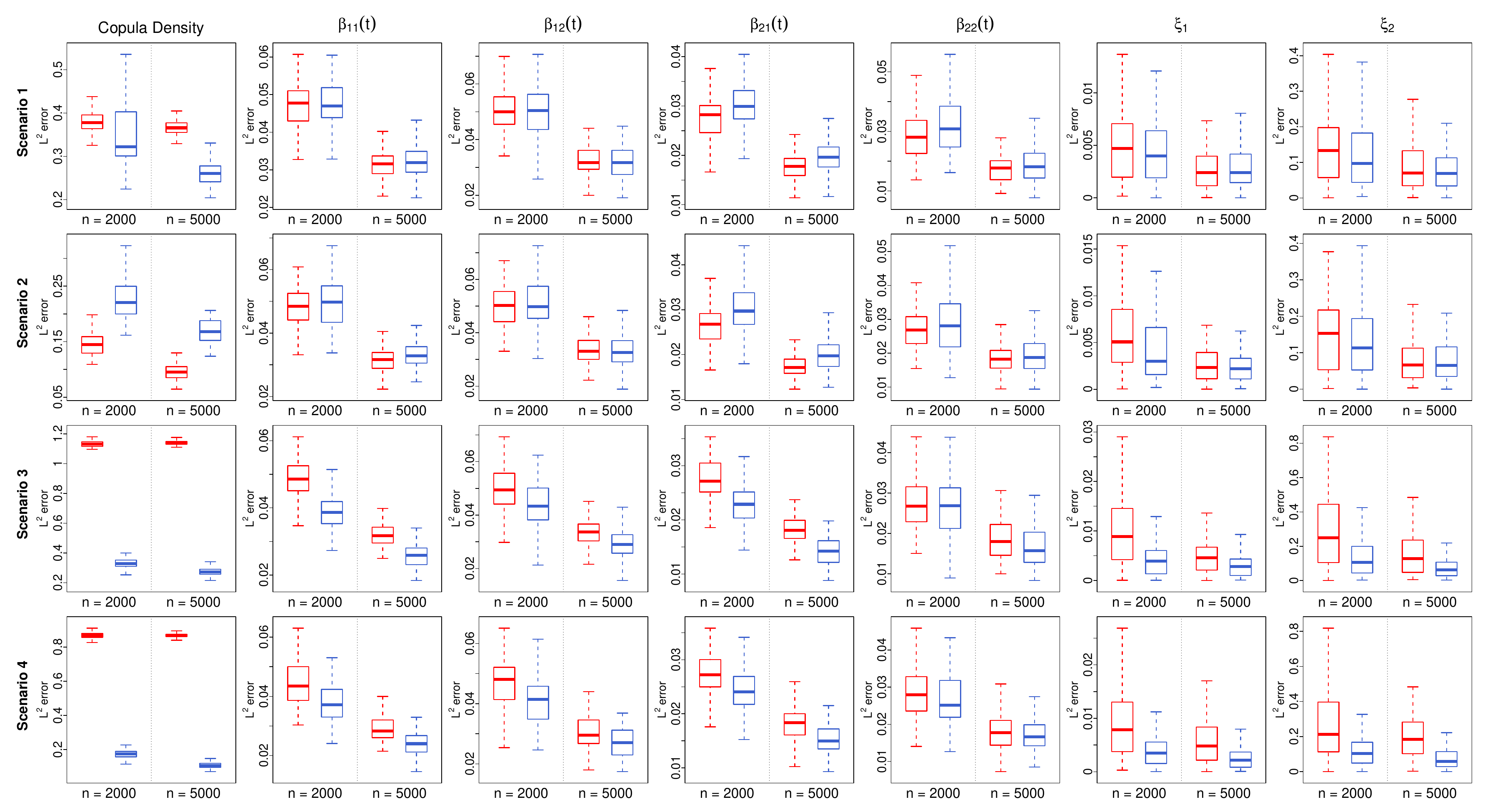}
    \caption{
Boxplots of estimation errors comparing \GJRM (red) and the proposed method (blue). Columns represent the target parameters: the copula density, the varying coefficients, and the two additional parameters (the variance for $y_{i1}$ and the shape parameter for $y_{i2}$). Rows represent Scenarios~1--4.
}

    \label{fig:RMSE}
\end{figure}

Figure~\ref{fig:RMSE} compares the estimation errors of the parameters across methods under each scenario. The results indicate that the relative performance depends primarily on the underlying dependence structure. When the data are generated from copulas with functional correlation matrices (Scenarios~1 and~2), both methods perform reasonably well, with only modest differences overall, except for copula density estimation. In particular, the proposed method achieves clearly better copula estimation accuracy in Scenario~1, whereas it performs worse in Scenario~2. Since Scenario~2 corresponds to the correctly specified setting for \GJRM, this behavior is expected. These results empirically suggest that the proposed mixture-based copula model can approximate copulas with functional correlation matrices reasonably well under misspecification.
In sharp contrast, when the data are generated from mixtures of copulas with covariate-dependent weights (Scenarios~3 and~4), the proposed method clearly outperforms \GJRM, especially in terms of copula density estimation. Notably, the proposed method attains substantially smaller copula density errors not only in Scenario~4, which is its correctly specified setting, but also under misspecification in Scenario~3. This indicates that modeling dependence solely through a covariate-varying correlation matrix is inadequate for capturing mixture-induced heterogeneity, whereas the proposed framework remains robust to the choice of copula family and yields stable estimation accuracy even for mixtures of $t$-copulas.
Although $t$-copulas are not exactly infinite mixtures of Gaussian copulas, the simulation results indicate that the proposed framework can nevertheless approximate such dependence structures well. Consequently, the results in this section provide empirical support for the approximation power of the proposed covariate-dependent copula modeling established in Theorems~\ref{thm:appGC} and~\ref{thm:appGCM}.

%-----------------------------------------------------------------------

\section{Application to BRFSS 2023 Data}
\label{sec6}
In this section, we apply the proposed model to the BRFSS 2023 data described in Section~\ref{sec:mot}. The dataset comprises self-reported information on health-related risk behaviors, chronic conditions, and preventive service use, collected through a dual-frame telephone survey of $n = 433{,}323$ U.S. adults conducted by the CDC. The core questionnaire covers domains including health status, lifestyle behaviors, chronic disease prevalence, mental health, and healthcare access. 
Table~\ref{tab:predictors} summarizes the explanatory variables $\mathbf x_i$ used in our analysis, together with their definitions and coding schemes. For each individual $i$, we define the multivariate response vector as $\mathbf{y}_i = (y_{i1}, \ldots, y_{i8})^\top$, as specified in Section~\ref{sec:mot}.
Age is reported in 5-year intervals from 25 to 80, with open-ended categories below 25 and above 80. We assign to each category a representative value, using the midpoint for the bounded intervals and the sample median within each open-ended category (21 and 82 for the lower and upper groups, respectively). The resulting 12 distinct values are used as the covariate $t_i$.
Individuals with missing observations are excluded from the analysis.
Given that only 12 distinct values of $t_i$ are available, we adopt a slightly conservative prior for knot selection. Specifically, we consider only $L=10$ candidate knots and set the decay hyperparameters to $\varpi_{jk}=0.4$ and $\varpi_h=0.4$.

\begin{table}[t!]
\caption{The explanatory variables used for each individual \(i\), with their possible value codes.}
\centering
\resizebox{1\textwidth}{!}{%
\begin{tabular}{@{} l  p{0.45\textwidth}  p{0.55\textwidth} @{}}
\toprule
{Variable}               & {Description}                                                & {Possible values}              \\ 
\midrule
$\texttt{Gender}_i$            & Gender indicator                  & 0 = Female; 1 = Male      \\
$\texttt{Race}_i$            &  Self-identified race indicator                      & 0 = Non-white or Hispanic; 1 = Non-Hispanic white \\
$\texttt{Marriage}_i$            &  Marital status indicator                    & 0 = Not married; 1 = Married \\
$\texttt{Education}_i$            & Highest level of education attained.                                & Ordered categories, 1--6       
\\
$\texttt{NumChild}_i$        & Number of children in the household.                                & Integers, 1--6                  \\
$\texttt{Employment}_i$        & Employment status indicator.                                & 0 = Not employed; 1 = Employed or self-employed                  \\
$\texttt{Income}_i$        & Income level.                                & Ordered categories, 1--7                  \\
$\texttt{OwnHome}_i$        & Home ownership indicator.                                & 0 = Rent or other arrangement; 1 = Own home                  \\
$\texttt{Smoking}_i$              & Indicator for current smoking status.                               & 0 = No; 1 = Yes \\
$\texttt{HeavyDrink}_i$              & Indicator for heavy drinker.                               & 0 = No; 1 = Yes \\
$\texttt{Insurance}_i$            & Indicator for health insurance coverage.                           & 0 = No; 1 = Yes                       \\
$\texttt{PhyRec}_i$        & Indicator for aerobic activity and muscle-strengthening activity. & 0 = None; 1 = One condition; 2 = Both conditions
                      \\
$\texttt{Urban}_i$   &  Urban residence indicator  
                     & 0 = Rural; 1 = Urban \\
\bottomrule
\end{tabular}
}
\label{tab:predictors}
\end{table}

%--------------------------------------------------- 
Figure~\ref{fig:plot_realmarg} displays the estimated varying coefficients $\beta_{jk}(t)$ for $j=1,\ldots,14$ and $k=1,\ldots,8$, along with pointwise 95\% credible bands.
The estimated varying coefficients illustrate how the effects of the explanatory variables vary with age. The overall age-adaptive patterns are similar to those in Figure~\ref{fig:motiv_marginal}. However, the coefficient magnitudes are not directly comparable to those in Figure~\ref{fig:motiv_marginal} because of the different response models: the present analysis employs probit and ordinal probit modeling, whereas the preliminary analysis in Section~\ref{sec:mot} was based on logit and ordinal logit models.
The baseline intercept effects are straightforward to interpret. As expected, for most health status outcomes, the baseline effects increase with age, indicating that health risks tend to rise as individuals grow older. The only exceptions are asthma and obesity (BMI), which exhibit inverted-U shapes. These patterns are consistent with existing findings that asthma and obesity are more prevalent in middle-aged adults than in younger or older groups \citep{busse2020asthma,flegal2016trends}.
Apart from the intercepts, the signs and magnitudes of the varying coefficients reflect the directions and strengths of the effects of the explanatory variables. For example, with the exception of arthritis and asthma, the estimated gender effect is generally positive for most outcomes, indicating higher disease risks for men than for women. In contrast, the higher prevalence of arthritis and asthma among women than men is well aligned with existing evidence \citep{shah2018sex,fallon2023prevalence}.
The age-varying gender effect is particularly clear for several outcomes. For example, diabetes and stroke show little or no gender difference at younger ages, whereas the disparity increases nearly linearly with age, so that middle-aged and older men exhibit higher risks. This pattern is consistent with findings from large cohort studies \citep{mnatzaganian2024sex,appelros2009sex}. Obesity (BMI) exhibits a similar age-increasing prevalence pattern for men. However, this result should be interpreted with caution because self-reported anthropometric measures are often misreported, which may lead to bias in obesity estimates \citep{ng2019biases}.
In contrast, high blood pressure exhibits a reversed and nonlinear pattern, indicating that men are more susceptible at younger ages, whereas the gender difference diminishes with age, consistent with previous studies \citep{sandberg2012sex}. Other responses and explanatory variables can be interpreted in a similar manner, but we omit further discussion for brevity.

%------------------------- marginal plots
\begin{figure}[p!]
  \centering

 \includegraphics[width=1\textwidth]{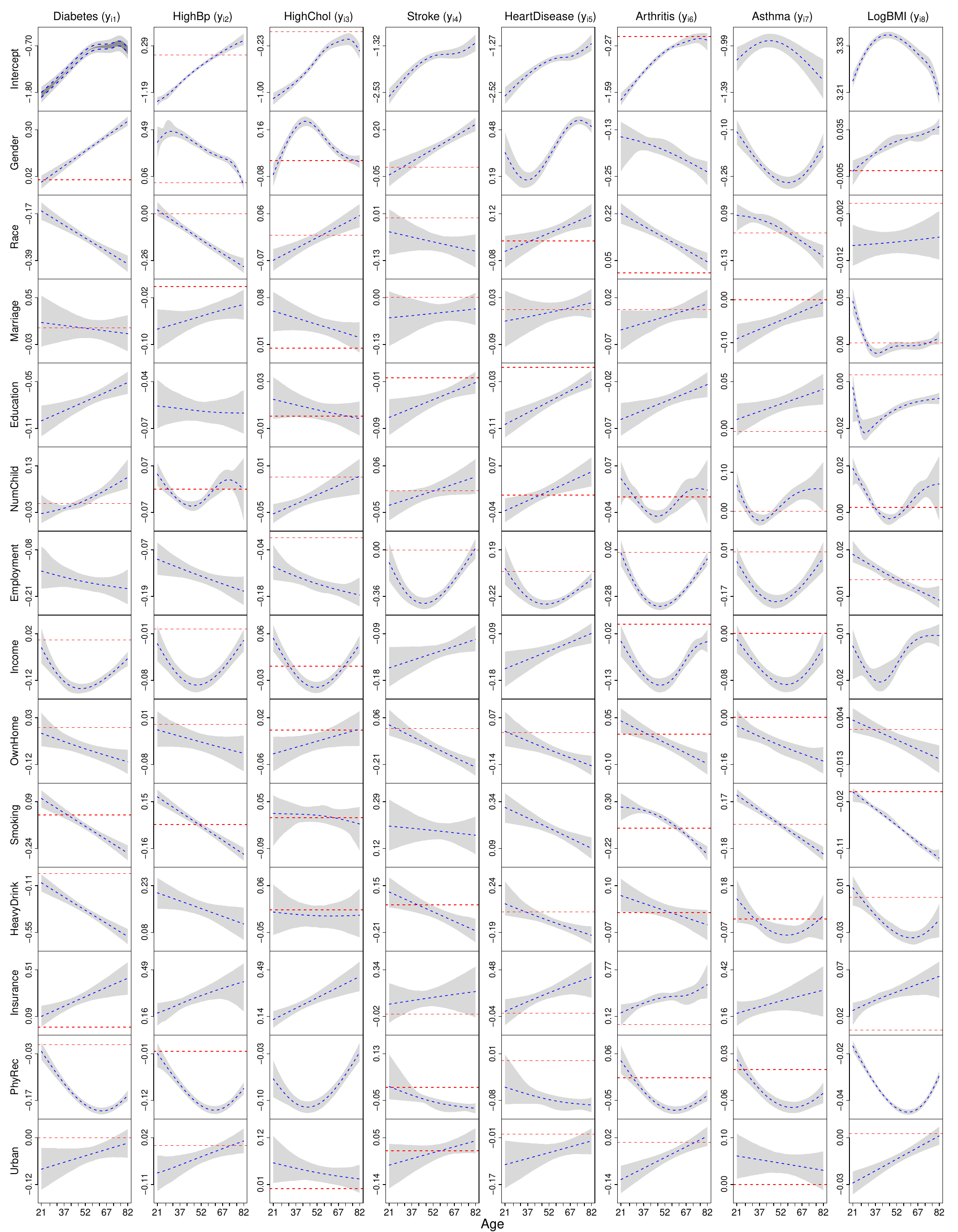}
%}
  \caption{The pointwise posterior means (blue dashed curve) and pointwise 95\% credible intervals (gray shade) from the Bayesian copula regression analyses of the BRFSS 2023 data. Each column corresponds to one of the eight health outcomes,  $k=1,\ldots,8$, and each row to one of the 14 demographic variables,  $j=1,\ldots,14$. The panel in row $j$ and column $k$ displays the varying coefficient $\beta_{jk}(t_i)$, showing how the effect of the explanatory variable $j$ on outcome $k$ evolves over $t$ under the joint model. 
  For the diabetes variable, the two intercept curves represent $-\Phi^{-1}(\Pr\{y_{i1}\le b \mid t_i\}) = \beta_{11}(t_i) - \xi_{1b}$, $b=1,2$, when $\mathbf{x}_i$ is zero, with $\xi_{11}=0$ for identifiability.}
  
  \label{fig:plot_realmarg}
\end{figure}

%------------------------- corr plot
\begin{figure}[t!]
    \centering
    
    \includegraphics[width=\textwidth]{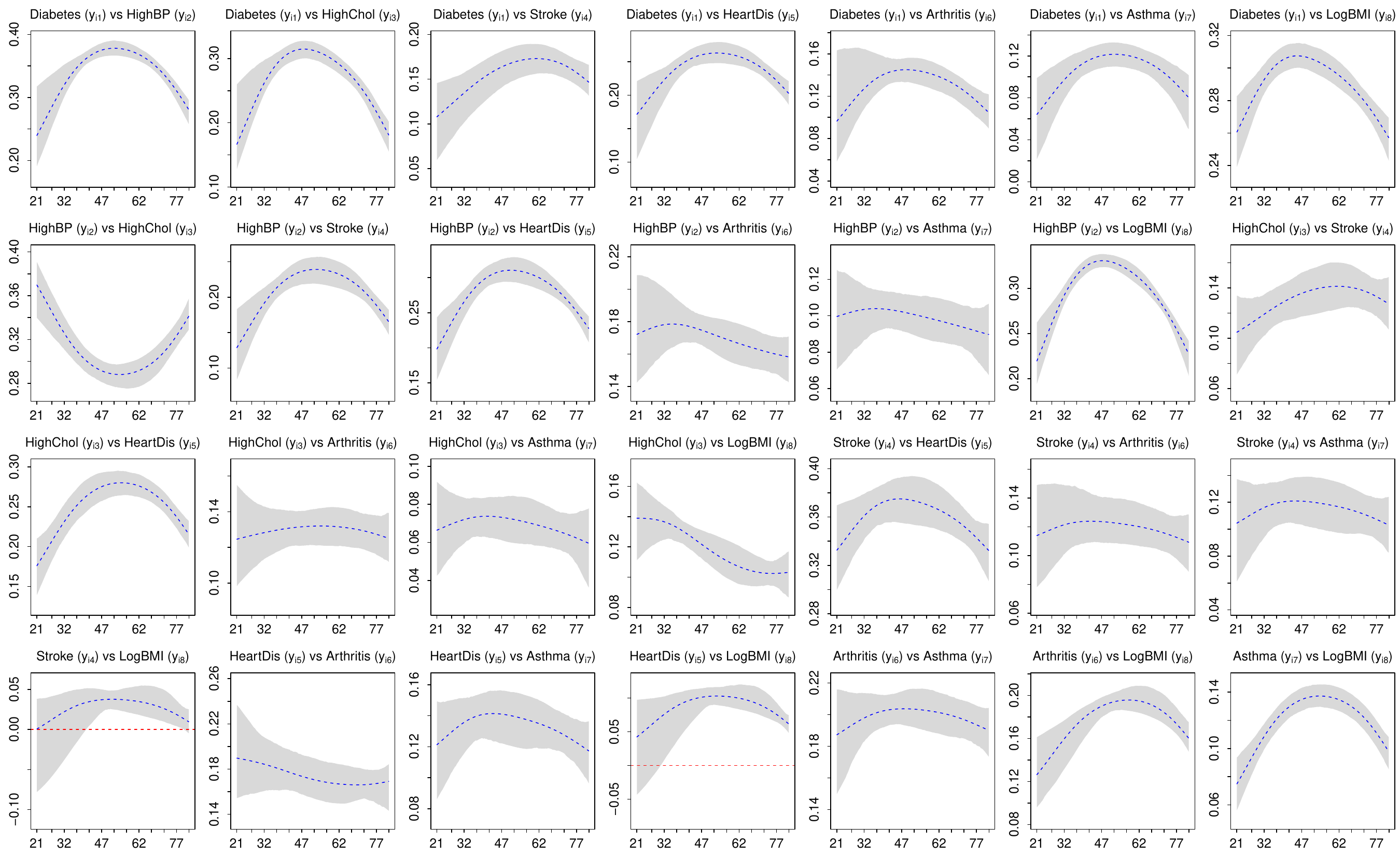}
    \caption{The pointwise posterior means (blue dashed curves) and 95\% pointwise credible intervals (gray shade) for the time-varying pairwise correlation functions among the latent Gaussian copula variables $\tilde z_{ik}$, $k=1,\ldots,8$.}
    \label{fig:real_corr}
\end{figure}

Figure~\ref{fig:real_corr} shows the estimated copula-latent correlation functions between the response models together with their pointwise 95\% credible bands. These correlations are not Pearson correlations between the responses and should therefore be interpreted with caution. Nevertheless, their signs and relative magnitudes remain informative for assessing the strength of dependence.
The estimated correlation patterns are closely related to those in Figure~\ref{fig:motiv_corr} from the preliminary analysis, although they are not directly comparable because the latter are based on residuals.
The cross-condition dependence varies systematically with age. For several pairs of health conditions, the latent correlations strengthen toward midlife and weaken again at older ages, forming an inverted-U pattern. Such patterns are consistent with the life-course evolution of multimorbidity, in which chronic conditions tend to accumulate and co-occur during middle adulthood \citep{chmiel2014spreading,hong2024age}. In contrast, the correlation between high blood pressure and high cholesterol exhibits a U-shaped pattern, reaching its minimum around midlife while remaining positive across all ages. This pair therefore exhibits a distinct age-dependent dependence pattern compared with the other condition combinations.

\section{Discussion}\label{sec7}
This study proposes a Bayesian nonparametric framework for multivariate conditional copula regression with varying coefficients. By combining adaptive spline-based marginals with a probit stick-breaking mixture of Gaussian copulas, the model accommodates heterogeneous response types and covariate-dependent dependence in a unified way. Theoretical and numerical results support its flexibility, showing good recovery under correct specification and robust performance under copula misspecification. In the BRFSS 2023 analysis, the fitted model yielded age-varying marginal effects and dependence patterns that were broadly consistent with existing empirical findings, while also providing a coherent joint view of multimorbidity that is not readily captured by separate marginal analyses alone. In particular, the proposed framework allows the strength and pattern of association among multiple health outcomes to evolve smoothly with age, thereby enabling a more nuanced assessment of how health risks and co-occurrence patterns vary with age. These features illustrate the practical value of the model for studying complex population health data with mixed outcome types. Future work may explore richer base copula families for representing more irregular, asymmetric, or tail-sensitive dependence structures.

\bibliographystyle{apalike}
\bibliography{ref}

\newgeometry{
 total={170mm,257mm},
 left=20mm,
 top=30mm,
 right=20mm,
 bottom=25mm
}

\subfile{supplement.tex}

%----------------------------------

\end{document}

%% file: supplement.tex
\title{Supplement to ``Bayesian Nonparametric Modeling  for Multivariate Conditional Copula Regression with Varying Coefficients''}

\date{}
\emptythanks
\maketitle
%----------------------------------

\appendix
\renewcommand{\thesection}{S\arabic{section}}
\renewcommand{\thesubsection}{S\arabic{section}.\arabic{subsection}}
\renewcommand{\thetable}{S\arabic{table}}
\renewcommand{\thefigure}{S\arabic{figure}}
\renewcommand{\theequation}{S\arabic{equation}}
\setcounter{section}{0}
\setcounter{subsection}{0}
\setcounter{table}{0}
\setcounter{figure}{0}
\setcounter{equation}{0}
% --------------------------------

\begingroup
  \fontsize{10}{11}\selectfont

%----------------------------------------------------
\section{Posterior Computation and MCMC Details} \label{app:sec1}

This section provides the details of the MCMC scheme used in the proposed model. 
For each response type, we specify the response-specific likelihood term and the copula-induced conditional density of the latent variables as needed for posterior computation. 
Using these building blocks, we derive the target posterior distribution
$p(\tilde{\boldsymbol{\alpha}}_{\tilde{\boldsymbol{\gamma}}},\tilde{\boldsymbol{\gamma}},\tilde{\mathbf{g}}, \boldsymbol{\xi}, \tilde{\mathbf{z}}, \boldsymbol{\alpha}^*_{\boldsymbol{\gamma}^*},\boldsymbol{\gamma}^*,\mathbf{g}^*,\mathbf{s},\mathbf{z}^*, \mathbf{D}, \mathbf{R}\mid \mathbf{y})$
and present the conditional posterior distributions of all parameters in the order in which they are sampled in the MCMC algorithm.
\paragraph{Step 1: Update of $(\tilde{\boldsymbol{\alpha}}_\gamma, \tilde{\boldsymbol{\gamma}})$.}
The MCMC updates for $(\tilde{\boldsymbol{\alpha}}_\gamma, \tilde{\boldsymbol{\gamma}})$ are constructed according to the response type. 
For Gaussian, binary, and ordinal responses, we adopt a Gaussian latent variable representation, which yields an efficient sampler. 
For response types that do not admit this representation, we use a Metropolis--Hastings update with an appropriately chosen proposal distribution.
In both cases, we introduce a latent variable $z_{ik}$ by transforming the Gaussian copula latent variable $\tilde z_{ik}$ as specified in \eqref{eq:reparamet}.
\begin{itemize}
    \item \textit{Gaussian, binary, and ordinal responses.} For these response types, introducing the latent variable $z_{ik}$ enables Gibbs updates for $(\tilde{\boldsymbol{\alpha}}_{k,\tilde{\boldsymbol{\gamma}}_k}, \tilde{\boldsymbol{\gamma}}_k)$. Accordingly, the joint posterior of $(\tilde{\boldsymbol{\alpha}}_{k,\tilde{\boldsymbol{\gamma}}_k}, \tilde{\boldsymbol{\gamma}}_k)$ can be decomposed as
    \begin{align}
        \label{app:jointposterior_gamma_alpha}p(\tilde{\boldsymbol{\alpha}}_{k,\tilde{\boldsymbol{\gamma}}_k}, \tilde{\boldsymbol{\gamma}}_k \mid \mathrm{rest}) \propto p( \tilde{\boldsymbol{\gamma}}_k \mid \mathrm{rest}\setminus \tilde{\boldsymbol{\alpha}}_{k,\tilde{\boldsymbol{\gamma}}_k}) p(\tilde{\boldsymbol{\alpha}}_{k,\tilde{\boldsymbol{\gamma}}_k} \mid \mathrm{rest}) 
    \end{align}
    We first update $\tilde{\boldsymbol{\gamma}}_k$ componentwise via Gibbs sampling. For the $\ell$th component, 
\begin{align}\label{eq:sup-gamma}
    p({\tilde{\gamma}_{k(\ell)}}\mid   \mathrm{rest}\setminus \tilde{\boldsymbol{\alpha}}_{k,\tilde{\boldsymbol{\gamma}}_k})
    & = p(\tilde{\gamma}_{k(\ell)}\mid \tilde{\boldsymbol{\gamma}}_{k\backslash (\ell )}) \int p(\mathbf{z}_{\cdot k}\mid   \tilde{\boldsymbol{\alpha}}_{k,\tilde{\boldsymbol{\gamma}}_k},\tilde{\boldsymbol{\gamma}}_k,\boldsymbol{{\xi}}_k,\tilde{\mathbf{z}}_{\cdot \backslash k},\mathbf{s},\mathbf{R})p(\tilde{\boldsymbol{\alpha}}_{k,\tilde{\boldsymbol{\gamma}}_k}\mid \tilde{\boldsymbol{\gamma}}_k,\tilde{\mathbf{g}}_k)d\tilde{\boldsymbol{\alpha}}_{k,\tilde{\boldsymbol{\gamma}}_k}.
\end{align}
For Gaussian responses, \eqref{eq:sup-gamma} holds because $\mathbf{z}_{\cdot k}=\mathbf{y}_{\cdot k}$. For binary and ordinal responses, it also holds because $\mathbf{y}_{\cdot k}$ is deterministic given $\mathbf{z}_{\cdot k}$ and $\boldsymbol{\xi}_k$.
Specifically, 
\begin{align*}
    & p({\tilde{\gamma}_{k(\ell)}}\mid   \mathrm{rest}\setminus \tilde{\boldsymbol{\alpha}}_{k,\tilde{\boldsymbol{\gamma}}_k})\nonumber\\& \propto p(\tilde{\gamma}_{k(\ell)}\mid  \tilde{\boldsymbol{\gamma}}_{k\backslash (\ell )})
    p(\mathbf{y}_{\cdot k},\mathbf{z}_{\cdot k}\mid {\tilde{\gamma}_{k(\ell)}},\tilde{\boldsymbol{ \gamma}}_{k\backslash (\ell)}, \boldsymbol{{\xi}}_k,\tilde{\mathbf{z}}_{\cdot \backslash k},\mathbf{s},\mathbf{R}),\nonumber\\
    & = p(\tilde{\gamma}_{k(\ell)}\mid \tilde{\boldsymbol{\gamma}}_{k\backslash (\ell )}) \int p(\mathbf{y}_{\cdot k},\mathbf{z}_{\cdot k}\mid  \tilde{\boldsymbol{\alpha}}_{k,\tilde{\boldsymbol{\gamma}}_k},\tilde{\boldsymbol{\gamma}}_k,\boldsymbol{{\xi}}_k,\tilde{\mathbf{z}}_{\cdot \backslash k},\mathbf{s},\mathbf{R})p(\tilde{\boldsymbol{\alpha}}_{k,\tilde{\boldsymbol{\gamma}}_k}\mid \tilde{\boldsymbol{\gamma}}_k,\tilde {\mathbf{g}}_k)d\tilde{\boldsymbol{\alpha}}_{k,\tilde{\boldsymbol{\gamma}}_k}, \nonumber\\
    & = p(\tilde{\gamma}_{k(\ell)}\mid \tilde{\boldsymbol{\gamma}}_{k\backslash (\ell )})p(\mathbf{y}_{\cdot k}\mid \mathbf{z}_{\cdot k},\boldsymbol{\xi}_k)\int p(\mathbf{z}_{\cdot k}\mid  \tilde{\boldsymbol{\alpha}}_{k,\tilde{\boldsymbol{\gamma}}_k},\tilde{\boldsymbol{\gamma}}_k, \boldsymbol{{\xi}}_k, \tilde{\mathbf{z}}_{\cdot \backslash k},\mathbf{s},\mathbf{R})p(\tilde{\boldsymbol{\alpha}}_{k,\tilde{\boldsymbol{\gamma}}_k}\mid \tilde{\boldsymbol{\gamma}}_k,\tilde {\mathbf{g}}_k) d\tilde{\boldsymbol{\alpha}}_{k,\tilde{\boldsymbol{\gamma}}_k}.\nonumber  
\end{align*}
We now derive the induced latent Gaussian kernel of $\mathbf{z}_{\cdot k}$ under the transformation \eqref{eq:reparamet}.
We begin by stating the full conditional kernel of $\tilde{\mathbf{z}}_{\cdot k}$:
\begin{align*}
     & p(\tilde{\mathbf{z}}_{\cdot k}\mid   \tilde{\boldsymbol{\alpha}}_{k,\tilde{\boldsymbol{\gamma}}_k},\tilde{\boldsymbol{\gamma}}_k,\boldsymbol{\xi}_k,\tilde{\mathbf{z}}_{\cdot \backslash k},\mathbf{s},\mathbf{R}), \\ & \propto \prod_{h=1}^H 
\prod_{i:s_i=h}\exp \!\left\{{1\over 2}\tilde{\mathbf{z}}_i^\top \!\left(\mathbf{I}_m - \mathbf{R}_{s_i}^{-1}\right)\tilde{\mathbf{z}}_i\right\}\prod_{i=1}^n\exp\!\left(-\frac{1}{2}\tilde z_{ik}^2\right),\\
& \propto \prod_{h=1}^H \exp\!\left\{{1\over 2}  \sum_{i:s_i =h} \!\left(1-\!\left(\mathbf R_{s_i}^{-1}\right)_{k,k}\right)\tilde z_{ik}^2 - \sum_{i:s_i =h} \sum_{k':k'\ne k}  \!\left(\mathbf R_{s_i}^{-1}\right)_{k,k'}\tilde z_{ik} \tilde z_{ik'}\right\}\prod_{i=1}^n\exp\!\left(-\frac{1}{2}\tilde z_{ik}^2\right),\\
& \propto \prod_{h=1}^H \exp\!\left\{-{1\over 2}  \sum_{i:s_i =h} \!\left(\mathbf R_{s_i}^{-1}\right)_{k,k}\tilde z_{ik}^2 - \sum_{i:s_i =h} \sum_{k':k'\ne k}  \!\left(\mathbf R_{s_i}^{-1}\right)_{k,k'}\tilde z_{ik} \tilde z_{ik'}\right\},
\end{align*}
where $\overline{\mathbf W}_{ik,\tilde{\boldsymbol{\gamma}}_k}$ denote the $i$th row of
$\overline{\mathbf W}_{k,\tilde{\boldsymbol{\gamma}}_k}$.
By the transformation of \eqref{eq:reparamet},
\begin{align*}
& p({\mathbf{z}}_{\cdot k}\mid \tilde{\boldsymbol{\alpha}}_{k,\tilde{\boldsymbol{\gamma}}_k},
\tilde{\boldsymbol{\gamma}}_k,\boldsymbol{\xi}_k,\tilde{\mathbf{z}}_{\cdot \backslash k},
\mathbf{s},\mathbf{R}) \\
& \propto \prod_{h=1}^H \exp\Biggl\{
-\frac{1}{2}\sum_{i:s_i=h}
\delta_{ik}^{-2}\!\left(\mathbf R_{s_i}^{-1}\right)_{k,k}
\!\left(
z_{ik}-\overline{\mathbf W}_{ik,\tilde{\boldsymbol{\gamma}}_k}
\tilde{\boldsymbol{\alpha}}_{k,\tilde{\boldsymbol{\gamma}}_k}
\right)^2 \\
& \quad 
-\sum_{i:s_i=h}\sum_{k':k'\ne k}
\delta_{ik}^{-1}\!\left(\mathbf R_{s_i}^{-1}\right)_{k,k'}
\!\left(
z_{ik}-\overline{\mathbf W}_{ik,\tilde{\boldsymbol{\gamma}}_k}
\tilde{\boldsymbol{\alpha}}_{k,\tilde{\boldsymbol{\gamma}}_k}
\right)\tilde z_{ik'}
\Biggr\}, \\
& \propto \exp\!\left\{
-\frac{1}{2}\!\left(\mathbf{z}_{\cdot k} - \overline{\mathbf{W}}_{k, \tilde{\boldsymbol{\gamma}}_k}\tilde{\boldsymbol{\alpha}}_{k,\tilde{\boldsymbol{\gamma}}_k}\right)^\top
\boldsymbol{\Psi}_k
\!\left(\mathbf{z}_{\cdot k} - \overline{\mathbf{W}}_{k, \tilde{\boldsymbol{\gamma}}_k}\tilde{\boldsymbol{\alpha}}_{k,\tilde{\boldsymbol{\gamma}}_k}\right)
-\!\left(\mathbf{z}_{\cdot k} - \overline{\mathbf{W}}_{k, \tilde{\boldsymbol{\gamma}}_k}\tilde{\boldsymbol{\alpha}}_{k,\tilde{\boldsymbol{\gamma}}_k}\right)^\top
\boldsymbol{\zeta}_k
\right\}.
\end{align*}
where $\delta_{ik}=1$ for binary and ordinal responses, and $\delta_{ik}=\sqrt{\xi_k}$ for Gaussian responses. The marginal likelihood of $\mathbf{z}_{\cdot k}$ in \eqref{eq:sup-gamma}, obtained by integrating out $\tilde{\boldsymbol{\alpha}}_{k,\tilde{\boldsymbol{\gamma}}_k}$, is given by
\begin{align}
\begin{split}
    & p(\mathbf{z}_{\cdot k}\mid \tilde{\boldsymbol{ \gamma}}_{k},  \tilde{\mathbf{g}}_k, \boldsymbol{\xi}_k, \tilde{\mathbf{z}}_{\cdot \backslash k}, \mathbf{s},\mathbf{R}) \\ 
    & = \int p(\mathbf{z}_{\cdot k}\mid   \tilde{\boldsymbol{\alpha}}_{k,\tilde{\boldsymbol{\gamma}}_k},\tilde{\boldsymbol{\gamma}}_k,\boldsymbol{\xi}_k,\tilde{\mathbf{z}}_{\cdot \backslash k},\mathbf{s},\mathbf{R})p(\tilde{\boldsymbol{\alpha}}_{k,\tilde{\boldsymbol{\gamma}}_k}\mid \tilde{\boldsymbol{\gamma}}_k,\tilde{\mathbf{g}}_k)d\tilde{\boldsymbol{\alpha}}_{k,\tilde{\boldsymbol{\gamma}}_k}, \\
    & \propto \int \exp\!\left\{-{1\over 2}\!\left(\mathbf{z}_{\cdot k} - \overline{\mathbf{W}}_{k, \tilde{\boldsymbol{\gamma}}_k}\tilde{\boldsymbol{\alpha}}_{k,\tilde{\boldsymbol{\gamma}}_k}\right)^\top\boldsymbol{\Psi}_k\!\left(\mathbf{z}_{\cdot k} - \overline{\mathbf{W}}_{k, \tilde{\boldsymbol{\gamma}}_k}\tilde{\boldsymbol{\alpha}}_{k,\tilde{\boldsymbol{\gamma}}_k}\right) -(\mathbf{z}_{\cdot k} - \overline{\mathbf{W}}_{k, \tilde{\boldsymbol{\gamma}}_k}\tilde{\boldsymbol{\alpha}}_{k,\tilde{\boldsymbol{\gamma}}_k})^\top  {\boldsymbol{\zeta}}_k \right\} \\& 
   \quad \times \!\left|\mathbf{V}_{{0}}\right|^{-{1\over 2}}\exp\!\left\{-{1\over 2} \tilde{\boldsymbol{\alpha}}_{k,\tilde{\boldsymbol{\gamma}}_k}^\top \mathbf{V}_{0}^{-1}  \tilde{\boldsymbol{\alpha}}_{k,\tilde{\boldsymbol{\gamma}}_k} \right\} \tilde{\boldsymbol{\alpha}}_{k,\tilde{\boldsymbol{\gamma}}_k}, \\
   & \propto \!\left|\overline{\mathbf{W}}_{k,\tilde{\boldsymbol{\gamma}}_k^\star}^\top \boldsymbol{\Psi}_k\overline{\mathbf{W}}_{k,\tilde{\boldsymbol{\gamma}}_k^\star}  + \mathbf{V}^{-1}_{0}\right|^{-{1\over 2}}\!\left|\mathbf{V}_{{0}}\right|^{-{1\over 2}}\\
   & \quad \times \exp\!\left\{{1\over2 }\!\left({\boldsymbol{\zeta}}_k+\boldsymbol{\Psi}_k\mathbf{z}_{\cdot k}\right)^\top \overline{\mathbf{W}}_{k, \tilde{\boldsymbol{\gamma}}_k}\!\left(\overline{\mathbf{W}}_{k,\tilde{\boldsymbol{\gamma}}_k^\star}^\top \boldsymbol{\Psi}_k\overline{\mathbf{W}}_{k,\tilde{\boldsymbol{\gamma}}_k^\star}  + \mathbf{V}^{-1}_{0}\right)^{-1}\overline{\mathbf{W}}_{k, \tilde{\boldsymbol{\gamma}}_k}^{\top}\!\left({\boldsymbol{\zeta}}_k+\boldsymbol{\Psi}_k\mathbf{z}_{\cdot k}\right)\right\}.
\end{split}
\label{eq:lhd_gamma}
\end{align}
By substituting \eqref{eq:lhd_gamma} into \eqref{eq:sup-gamma}, the $\ell$th entry of $\tilde{\boldsymbol{\gamma}}_k$ is updated according to the following posterior inclusion probability:
\begin{align*}
    & \mathrm{Pr}(\tilde{\gamma}_{k(\ell)}=1\mid  \mathrm{rest}\setminus \tilde{\boldsymbol{\alpha}}_{k,\tilde{\boldsymbol{\gamma}}_k})   = \!\left( 1 + {p(\mathbf{z}_{\cdot k}\mid  {\tilde{\gamma}_{k(\ell)}}=0, \tilde{\boldsymbol{ \gamma}}_{k\backslash (\ell)},\tilde{\mathbf{g}}_k, \boldsymbol{{\xi}}_k,\tilde{\mathbf{z}}_{\cdot \backslash k}, \mathbf{s},\mathbf{R})p(\tilde{\gamma}_{k(\ell)}=0 \mid \tilde{\boldsymbol{\gamma}}_{k\backslash (\ell )})\over p(\mathbf{z}_{\cdot k} \mid  {\tilde{\gamma}_{k(\ell)}}=1, \tilde{\boldsymbol{ \gamma}}_{k\backslash (\ell)},\tilde{\mathbf{g}}_k, \boldsymbol{{\xi}}_k,\tilde{\mathbf{z}}_{\cdot \backslash k}, \mathbf{s},\mathbf{R})p(\tilde{\gamma}_{k(\ell)}=1\mid \tilde{\boldsymbol{\gamma}}_{k\backslash (\ell )})}\right)^{-1},
\end{align*}
where 
$p(\tilde{\gamma}_{k(\ell)}=1\mid \tilde{\boldsymbol{\gamma}}_{k\backslash (\ell )})=(1-\varpi_{jk})(|\tilde{\boldsymbol\gamma}_{jk\backslash \tilde{\ell}}|+1) (L-|\tilde{\boldsymbol\gamma}_{jk\backslash \tilde{\ell}}|  + (1-\varpi_{jk})(|\tilde{\boldsymbol\gamma}_{jk\backslash \tilde{\ell}}|+1))^{-1}$. Here, \(j\) is the predictor-block index such that
\(\ell \in \{(j-1)L+1,\ldots,jL\}\), and
\(\tilde{\ell}=\ell-(j-1)L\) is the within-block index.
Accordingly, \(|\tilde{\boldsymbol\gamma}_{jk\backslash \tilde{\ell}}|\) denotes the number of selected knots
in the \(j\)th predictor block excluding the \(\tilde{\ell}\)th component.

Given the updated  $\tilde{\boldsymbol{\gamma}}_k$, the conditional posterior of $\tilde{\boldsymbol{\alpha}}_{k,\tilde{\boldsymbol{\gamma}}_k}$ is Gaussian, and is derived as follows.
\begin{align*}
    p(& \tilde{\boldsymbol{\alpha}}_{k,\tilde{\boldsymbol{\gamma}}_k}\mid   \mathrm{rest}), \\
    &  \propto p(\mathbf{z}_{\cdot k}\mid  \tilde{\boldsymbol{\alpha}}_{k,\tilde{\boldsymbol{\gamma}}_k},\tilde{\boldsymbol{\gamma}}_k,{\boldsymbol{\xi}}_k,\tilde{\mathbf{z}}_{\cdot \backslash k}, \mathbf{s},\mathbf{R})p(\tilde{\boldsymbol{\alpha}}_{k,\tilde{\boldsymbol{\gamma}}_k}\mid \tilde{\boldsymbol{\gamma}}_k, \tilde{\mathbf{g}}_k),\\
    & \propto \exp\!\left\{-{1\over 2}\!\left(\mathbf{z}_{\cdot k} - \overline{\mathbf{W}}_{k, \tilde{\boldsymbol{\gamma}}_k}\tilde{\boldsymbol{\alpha}}_{k,\tilde{\boldsymbol{\gamma}}_k}\right)^\top\boldsymbol{\Psi}_k\!\left( \mathbf{z}_{\cdot k} -  \overline{\mathbf{W}}_{k, \tilde{\boldsymbol{\gamma}}_k}\tilde{\boldsymbol{\alpha}}_{k,\tilde{\boldsymbol{\gamma}}_k}\right) -(\mathbf{z}_{\cdot k} - \overline{\mathbf{W}}_{k, \tilde{\boldsymbol{\gamma}}_k}\tilde{\boldsymbol{\alpha}}_{k,\tilde{\boldsymbol{\gamma}}_k})^\top {\boldsymbol{\zeta}}_k \right\}\\
    & \quad \times \exp\!\left\{-{1\over 2} \tilde{\boldsymbol{\alpha}}_{k,\tilde{\boldsymbol{\gamma}}_k}^\top \mathbf{V}_{0}^{-1}  \tilde{\boldsymbol{\alpha}}_{k,\tilde{\boldsymbol{\gamma}}_k} \right\}.
\end{align*}
Therefore, 
\begin{gather}
    \label{eq:condpost_alpha}\tilde{\boldsymbol{\alpha}}_{k,\tilde{\boldsymbol{\gamma}}_k}\mid   \mathrm{rest} \sim \text{N}_{2+\sum_j|\tilde{\boldsymbol{\gamma}}_{jk}|}(\boldsymbol{\mu}_{\tilde{\boldsymbol{\alpha}}_k},\mathbf{V}_{\tilde{\boldsymbol{\alpha}}_k}),\nonumber
\end{gather}
where
\begin{align*}
\mathbf{V}_{\tilde{\boldsymbol{\alpha}}_k} & = \!\left(\overline{\mathbf{W}}_{k,\tilde{\boldsymbol{\gamma}}_k^\star}^\top \boldsymbol{\Psi}_k\overline{\mathbf{W}}_{k,\tilde{\boldsymbol{\gamma}}_k}  + \mathbf{V}^{-1}_{0}\right)^{-1},\\
  \boldsymbol{\mu}_{\tilde{\boldsymbol{\alpha}}_k} & = \mathbf{V}_{\tilde{\boldsymbol{\alpha}}_k}\overline{\mathbf{W}}_{k,\tilde{\boldsymbol{\gamma}}_k}^\top\!\left({\boldsymbol{\zeta}}_k+ \boldsymbol{\Psi}_k \mathbf{z}_{\cdot k}\right).
\end{align*}

%-------------------

\item \textit{Binomial, negative binomial, and gamma responses.} We utilize a Metropolis-Hastings step to jointly sample $(\tilde{\boldsymbol{\alpha}}_{k,\tilde{\boldsymbol{\gamma}}_k}, \tilde{\boldsymbol{\gamma}}_k, \tilde{\mathbf{z}}_{\cdot k})$ since marginalization is intractable.  The target conditional posterior density is
\begin{align} 
\begin{split}
p(\tilde{\boldsymbol{\alpha}}_{k,\tilde{\boldsymbol{\gamma}}_k},\tilde{\boldsymbol{\gamma}}_k, \tilde{\mathbf{z}}_{\cdot k}\mid \mathrm{rest}) 
& \propto  p(\mathbf{y}_{\cdot k}\mid \tilde{\boldsymbol{\alpha}}_{k,\tilde{\boldsymbol{\gamma}}_k},\tilde{\boldsymbol{\gamma}}_k, \boldsymbol{\xi}_k, \tilde{\mathbf{z}}_{\cdot k})p(\tilde{\mathbf{z}}_{\cdot k} \mid\tilde{\mathbf{z}}_{\cdot \backslash k},\mathbf s, \mathbf R )p(\tilde{\boldsymbol{\alpha}}_{k,\tilde{\boldsymbol{\gamma}}_k}\mid\tilde{\boldsymbol{\gamma}}_k,\tilde{\mathbf{g}}_k)p(\tilde{\boldsymbol{\gamma}}_k).
\label{app:MH_target}
\end{split}
\end{align} 
After proposing $\tilde{\boldsymbol{\gamma}}_k$ by randomly selecting one entry and setting it to 0 or 1 with equal probability, we propose $\tilde{\boldsymbol{\alpha}}_{k,\tilde{\boldsymbol{\gamma}}_k}$ from the Gaussian distribution in \eqref{eqn:gaussalpha}, and $\tilde{\mathbf z}_{\cdot k}$ from $p(\tilde{\mathbf z}_{\cdot k}\mid \mathrm{rest})$, the full conditional distribution described in Step~\ref{par:tildez-update}.
 The acceptance probability is as follows:
 \begin{align}
 1\land \frac{p(\mathbf{y}_{\cdot k} \mid\tilde{\boldsymbol{\alpha}}_{k,\tilde{\boldsymbol{\gamma}}_k^\star}^\star,\tilde{\boldsymbol{\gamma}}_k^\star,  \boldsymbol{{\xi}}_k,\tilde{\mathbf{z}}_{\cdot \backslash k},\mathbf s,\mathbf R)p(\tilde{\boldsymbol{\alpha}}_{k,\tilde{\boldsymbol{\gamma}}_k^\star}^\star\mid\tilde{\boldsymbol{\gamma}}_k^\star,\tilde{\mathbf{g}}_k)
p(\tilde{\boldsymbol{\gamma}}_k^\star) q(\tilde{\boldsymbol{\alpha}}_{k,\tilde{\boldsymbol{\gamma}}_k}\mid\tilde{\boldsymbol{\gamma}}_k,\tilde{\mathbf{g}}_k, \boldsymbol{\xi}_k,  \mathbf{z}_{\cdot k}^\star,\tilde{\mathbf{z}}_{\cdot \backslash k},\mathbf s,\mathbf R)}{p(\mathbf{y}_{\cdot k} \mid\tilde{\boldsymbol{\alpha}}_{k,\tilde{\boldsymbol{\gamma}}_k},\tilde{\boldsymbol{\gamma}}_k, \boldsymbol{{\xi}}_k,\tilde{\mathbf{z}}_{\cdot \backslash k},\mathbf s,\mathbf R)p(\tilde{\boldsymbol{\alpha}}_{k,\tilde{\boldsymbol{\gamma}}_k}\mid\tilde{\boldsymbol{\gamma}}_k,\tilde{\mathbf{g}}_k)
p(\tilde{\boldsymbol{\gamma}}_k) q(\tilde{\boldsymbol{\alpha}}^\star_{k,\tilde{\boldsymbol{\gamma}}_k^\star}\mid\tilde{\boldsymbol{\gamma}}_k^\star,\tilde{\mathbf{g}}_k, \boldsymbol{\xi}_k,  \mathbf{z}_{\cdot k},\tilde{\mathbf{z}}_{\cdot \backslash k},\mathbf s,\mathbf R)},
\label{app:acceptratio}
\end{align}
where $\mathbf{z}_{\cdot k}^\star = (z^\star_{1k},\ldots, z^\star_{nk})^\top \in \mathbb R^n$ and ${z}^\star_{ik} = \delta_{ik}\tilde{z}^\star_{ik} + \overline{\mathbf{W}}_{ik, \boldsymbol{\tilde \gamma}^\star_k}\tilde{\boldsymbol{\alpha}}^\star_{k,\tilde{\boldsymbol{\gamma}}_{k}^\star}$.
The acceptance probability of the form given in \eqref{app:acceptratio} is obtained because 
\begin{align}
    \frac{p(\mathbf{y}_{\cdot k} \mid\tilde{\boldsymbol{\alpha}}_{k,\tilde{\boldsymbol{\gamma}}_k^\star}^\star,\tilde{\boldsymbol{\gamma}}_k^\star,  \boldsymbol{{\xi}}_k,\tilde{\mathbf{z}}_{\cdot  k}^\star)p(\tilde{\mathbf{z}}_{\cdot k}^\star \mid \tilde{\mathbf{z}}_{\cdot \backslash k} ,\mathbf{s}, \mathbf{R})p(\tilde{\mathbf{z}}_{\cdot k} \mid \mathrm{rest})
    }{p(\mathbf{y}_{\cdot k} \mid\tilde{\boldsymbol{\alpha}}_{k,\tilde{\boldsymbol{\gamma}}_k},\tilde{\boldsymbol{\gamma}}_k,  \boldsymbol{{\xi}}_k,\tilde{\mathbf{z}}_{\cdot  k})p(\tilde{\mathbf{z}}_{\cdot k} \mid \tilde{\mathbf{z}}_{\cdot \backslash k} ,\mathbf{s}, \mathbf{R})p(\tilde{\mathbf{z}}_{\cdot k}^\star \mid {\mathrm{rest}^\star})
} = \frac{
p(\mathbf{y}_{\cdot k} \mid\tilde{\boldsymbol{\alpha}}_{k,\tilde{\boldsymbol{\gamma}}_k^\star}^\star,\tilde{\boldsymbol{\gamma}}_k^\star,  \boldsymbol{{\xi}}_k,\tilde{\mathbf{z}}_{\cdot \backslash k},\mathbf s,\mathbf R)
}
{
p(\mathbf{y}_{\cdot k} \mid\tilde{\boldsymbol{\alpha}}_{k,\tilde{\boldsymbol{\gamma}}_k},\tilde{\boldsymbol{\gamma}}_k,  \boldsymbol{{\xi}}_k,\tilde{\mathbf{z}}_{\cdot \backslash k},\mathbf s,\mathbf R)
}.
\end{align}
Here, $\mathrm{rest}$ denotes the conditioning set in the full conditional distribution
of $\tilde{\mathbf z}_{\cdot k}$ at the current state, and ${\mathrm{rest}^\star}$
denotes the same conditioning set evaluated at the proposed state, that is, with
$(\tilde{\boldsymbol{\alpha}}_{k,\tilde{\boldsymbol{\gamma}}_k},
\tilde{\boldsymbol{\gamma}}_k)$ replaced by
$(\tilde{\boldsymbol{\alpha}}^\star_{k,\tilde{\boldsymbol{\gamma}}^\star_k},
\tilde{\boldsymbol{\gamma}}^\star_k)$. 
The conditional likelihood  $p(\mathbf{y}_{\cdot k} \mid\tilde{\boldsymbol{\alpha}}_{k,\tilde{\boldsymbol{\gamma}}_k},\tilde{\boldsymbol{\gamma}}_k,  \boldsymbol{{\xi}}_k,\tilde{\mathbf{z}}_{\cdot \backslash k},\mathbf s,\mathbf R)$ in \eqref{app:acceptratio} is specified separately for discrete and continuous responses. For discrete responses, the conditional likelihood is obtained by marginalizing out $\tilde{\boldsymbol{z}}_{\cdot k}$, which yields
\begin{align} \label{app:musigmaeq}
    p(\mathbf{y}_{\cdot k} \mid \tilde{\boldsymbol{\alpha}}_{k,\tilde{\boldsymbol{\gamma}}_k},\tilde{\boldsymbol{\gamma}}_k,  \boldsymbol{\xi}_k,\tilde{\mathbf{z}}_{\cdot \backslash k},\mathbf{s},\mathbf{R})   \propto \prod_{i=1}^n \!\left\{ \Phi\!\left( 
{u_{ik} -\tilde{\mu}_{i,s_i(k,\backslash k)} \over \tilde \sigma_{s_i(k,\backslash  k)}}\right) - \Phi \!\left( 
 {l_{ik} -\tilde{\mu}_{i,s_i(k,\backslash  k)} \over \tilde \sigma_{s_i(k,\backslash k)}} \right)\right\},
\end{align}
where $l_{ik} = \Phi^{-1}\{F_k(y_{ik}-1;{\boldsymbol{\beta}}_{k}, {\boldsymbol{\xi}}_k,\mathbf{x}_i,t_i )\}$ and $u_{ik} = \Phi^{-1}\{F_k(y_{ik};{\boldsymbol{\beta}}_{k}, {\boldsymbol{\xi}}_k,\mathbf{x}_i,t_i )\}$ with $\tilde{\mu}_{i,s_i(k,\backslash k)} = (\mathbf{R}_{s_i})_{k, \backslash  k}(\mathbf{R}_{s_i})^{-1}_{\backslash  k, \backslash  k}\tilde{\mathbf{z}}_{i,\backslash k}$ and $\tilde{\sigma}_{s_i(k, \backslash  k)}^2 = 1-(\mathbf{R}_{s_i})_{k, \backslash  k}(\mathbf{R}_{s_i})^{-1}_{\backslash  k,\backslash  k}(\mathbf{R}_{s_i})_{\backslash  k,k}$. Note that $(\mathbf{R}_{s_i})_{k,\backslash k}$ denotes the $k$th row of matrix $\mathbf{R}_{s_i}$ excluding its $k$th element,
$(\mathbf{R}_{s_i})^{-1}_{\backslash  k,\backslash  k}$ denotes the inverse of the submatrix of $\mathbf{R}_{s_i}$ obtained by removing its $k$-th row and $k$-th column,
$\tilde{\mathbf{z}}_{i,\backslash k}$ is the vector $\tilde{\mathbf{z}}_{i}$ excluding its $k$th component.
For continuous responses, it coincides with the usual conditional likelihood:
\begin{align}
\begin{split}\label{app:contilikelihood}
&p(\mathbf{y}_{\cdot k} \mid  \tilde{\boldsymbol{\alpha}}_{k,\tilde{\boldsymbol{\gamma}}_k},\tilde{\boldsymbol{\gamma}}_k,   \boldsymbol{{\xi}}_k, \tilde{\mathbf{z}}_{\cdot \backslash k},\mathbf{s},\mathbf{R}) \\ 
&  \propto \prod_{h=1}^H 
\!\left(\prod_{i:s_i=h}\exp \!\left\{{1\over 2}\tilde{\mathbf{z}}_i^\top\!\left(\mathbf{I}_m - \mathbf{R}_{s_i}^{-1}\right)\tilde{\mathbf{z}}_i\right\} \right)
\prod_{i=1}^n f_k(y_{ik}; {\boldsymbol{\beta}}_{k}, {\boldsymbol{\xi}}_k,\mathbf{x}_i,t_i), \\
&  \propto \prod_{h=1}^H \exp\!\left\{{1\over 2}  \sum_{i:s_i =h}\!\left(1-\!\left(\mathbf R_{s_i}^{-1}\right)_{k,k}\right) \tilde z_{ik}^2 - \sum_{i:s_i =h} \sum_{k':k'\ne k}  \!\left(\mathbf R_{s_i}^{-1}\right)_{k,k'}\tilde z_{ik} \tilde z_{ik'}\right\} \prod_{i=1}^n f_k(y_{ik}; {\boldsymbol{\beta}}_{k}, {\boldsymbol{\xi}}_k,\mathbf{x}_i,t_i),
\end{split}
\end{align}  
where $(\mathbf{R}_{s_i}^{-1})_{k,k'}$ denotes the $(k,k')$th element of the matrix $\mathbf{R}_{s_i}^{-1}$, $\tilde z_{ik} = \Phi^{-1}\{F_k( y_{ik};{\boldsymbol{\beta}}_{k}, {\boldsymbol{\xi}}_k,\mathbf{x}_i,t_i )\}$, and \(f_k(\cdot ;\boldsymbol \beta_k,\boldsymbol\xi_k,\mathbf x_i,t_i)\)
denotes the density for the $k$th response.
\end{itemize}

\paragraph{Step 2: Update of  $\tilde{\mathbf{g}}$.}
The update for $\tilde{\mathbf{g}}_k$ is identical across all response types and is carried out via Gibbs sampling by semi-conjugacy.
For $j=1,\ldots,p$ and $k=1,\ldots,m$,
\begin{align*}
    p(\tilde{{g}}_{jk}\mid   \mathrm{rest}) 
    & \propto p(\tilde{\boldsymbol{\alpha}}_{k,\tilde{\boldsymbol{\gamma}}_k}\mid \tilde{\boldsymbol{\gamma}}_k, \tilde{\mathbf{g}}_k)p(\tilde g_{jk}),\\
    & \propto (\tilde g_{jk})^{-{1\over 2}}\!\left|\tilde g_{jk}(\tilde{\mathbf{W}}^\top_{jk,\boldsymbol{\tilde \gamma}_{jk}}\tilde{\mathbf{W}}_{jk,\boldsymbol{\tilde \gamma}_{jk}})^{-1}\right|^{-{1\over 2}}\\
    & \quad \times \exp\!\left\{-{1\over 2\tilde g_{jk}}\tilde{\boldsymbol{\alpha}}_{jk\backslash 0, \tilde{\boldsymbol{\gamma}}_{jk}}^\top\tilde{\mathbf{W}}^\top_{jk,\boldsymbol{\tilde \gamma}_{jk}}\tilde{\mathbf{W}}_{jk,\boldsymbol{\tilde \gamma}_{jk}}\tilde{\boldsymbol{\alpha}}_{jk\backslash 0, \tilde{\boldsymbol{\gamma}}_{jk}}\right\}\exp\!\left\{-{n\over 2\tilde g_{jk}}\right\}.
\end{align*}
Therefore, 
\begin{align*}
   \tilde{g}_{jk}\mid \mathrm{rest} \sim \text{IG}\!\left({1\over 2}\!\left({|\boldsymbol{\tilde\gamma}_{jk}|+2}\right),{1\over 2}\!\left({\tilde{\boldsymbol{\alpha}}_{jk\backslash 0, \tilde{\boldsymbol{\gamma}}_{jk}}^\top\tilde{\mathbf{W}}^\top_{jk,\boldsymbol{\tilde \gamma}_{jk}}\tilde{\mathbf{W}}_{jk,\boldsymbol{\tilde \gamma}_{jk}}\tilde{\boldsymbol{\alpha}}_{jk\backslash 0, \tilde{\boldsymbol{\gamma}}_{jk}}+n}\right)\right).
\end{align*}

%------------------------------
\paragraph{Step 3: Update of  $\boldsymbol{\xi}$.}

Each marginal response is associated with an additional distribution-specific
parameter $\boldsymbol{\xi}_k$, which is updated by a Metropolis--Hastings step.
Given the current state, we draw a proposal $\boldsymbol{\xi}_k^\star$ from a
response-specific proposal density $q_k$. The proposed value is then accepted with
probability
\begin{align}
\label{app:xiupdate-likelihood}
1 \land
\frac{
p(\mathbf{y}_{\cdot k} \mid
\tilde{\boldsymbol{\alpha}}_{k,\tilde{\boldsymbol{\gamma}}_k},
\tilde{\boldsymbol{\gamma}}_k,
\boldsymbol{\xi}_k^\star,
\tilde{\mathbf z}_{\cdot\backslash k},
\mathbf s,\mathbf R)\,
p(\boldsymbol{\xi}_k^\star)\,
q_k(\boldsymbol{\xi}_k \mid \boldsymbol{\xi}_k^\star)
}{
p(\mathbf{y}_{\cdot k} \mid
\tilde{\boldsymbol{\alpha}}_{k,\tilde{\boldsymbol{\gamma}}_k},
\tilde{\boldsymbol{\gamma}}_k,
\boldsymbol{\xi}_k,
\tilde{\mathbf z}_{\cdot\backslash k},
\mathbf s,\mathbf R)\,
p(\boldsymbol{\xi}_k)\,
q_k(\boldsymbol{\xi}_k^\star \mid \boldsymbol{\xi}_k)
}.
\end{align}

For Gaussian marginals, $\xi_k$ is the variance parameter. We assign the prior $\xi_k \sim \mathrm{IG}({\varepsilon}, {\varepsilon})$,
where $\varepsilon > 0$ is a small positive hyperparameter, and use the independence proposal
\[
\xi_k^\star \sim q_k\!\left(\cdot \mid
\mathbf y_{\cdot k},
\tilde{\boldsymbol{\alpha}}_{k,\tilde{\boldsymbol{\gamma}}_k},
\tilde{\boldsymbol{\gamma}}_k\right)
=
\mathrm{IG}(a_{\xi_k},b_{\xi_k}),
\]
where
\begin{align*}
a_{\xi_k} &= \frac{1}{2}n+\varepsilon,\\
b_{\xi_k} &=
\frac{1}{2}
(\mathbf y_{\cdot k}
-\overline{\mathbf W}_{k,\tilde{\boldsymbol{\gamma}}_k}
\tilde{\boldsymbol{\alpha}}_{k,\tilde{\boldsymbol{\gamma}}_k})^\top
(\mathbf y_{\cdot k}
-\overline{\mathbf W}_{k,\tilde{\boldsymbol{\gamma}}_k}
\tilde{\boldsymbol{\alpha}}_{k,\tilde{\boldsymbol{\gamma}}_k})
+\varepsilon.
\end{align*}
This proposal is derived from the copula-free model, and the update is an independence-chain Metropolis-Hastings step.

For gamma and negative-binomial marginals, $\xi_k$ denotes the shape parameter and the size parameter, respectively. In both cases, we use the log-scale random-walk proposal
\[
\log \xi_k^\star \mid \log \xi_k \sim \mathrm N(\log \xi_k,\tau_{\xi_k}^2).
\]

For ordinal probit marginals with $B_k$ categories,
$\boldsymbol{\xi}_k=(\xi_{k1},\dots,\xi_{k,B_k-1})^\top$ denotes the threshold vector, with $\xi_{k1}=0$ fixed for identifiability. Let
\[
\xi^\ast_{kb}=\log(\xi_{k,b+1}-\xi_{kb}), \quad b=1,\dots,B_k-2.
\]
We update the transformed thresholds jointly using the proposal
\[
\boldsymbol{\xi}_k^{\ast\star} \mid \boldsymbol{\xi}_k^\ast
\sim \mathrm N\bigl(\boldsymbol{\xi}_k^\ast,\,\tau_{\xi_k}^2\,\mathbf I_{B_k-2}\bigr).
\]

The likelihood term in \eqref{app:xiupdate-likelihood} is evaluated via \eqref{app:contilikelihood} for Gaussian and gamma responses, whereas \eqref{app:musigmaeq} is used for the negative binomial and ordinal probit cases. For random-walk proposals, the tuning parameter $\tau_{\xi_k}$ is adaptively adjusted during burn-in according to the observed acceptance probability and is fixed thereafter.

%----------------------------
\paragraph{Step 4: Update of ${\tilde{\mathbf{z}}}$.}\label{app:paragraph_step_tildez}
The Gaussian copula latent variable $\tilde{\mathbf{z}}$ maintains a one‐to‐one correspondence with continuous responses; however, for discrete responses, this mapping becomes many‐to‐one, so we update $\tilde{\mathbf{z}}$ by stochastically sampling from its conditional posterior distribution, which is given by
\begin{align*}
     p(\tilde z_{ik} \mid \mathrm{rest}) 
    & \propto p(y_{ik} \mid \tilde{\boldsymbol{ \alpha}}_{k,\tilde{\boldsymbol{ \gamma}}_k}, \tilde{\boldsymbol{ \gamma}}_k ,\boldsymbol{\xi}_k, \tilde{{z}}_{ik})p(\tilde z_{ik}\mid \tilde{\mathbf{z}}_{i,\backslash k},{s}_i,\mathbf{R}_{s_i}),\\
    & \propto \mathbb{I}\{l_{ik}<\tilde{z}_{ik}\le u_{ik}\}p(\tilde z_{ik}\mid \tilde{\mathbf{z}}_{i,\backslash k},{s}_i,\mathbf{R}_{s_i}),\\
    & \propto \mathrm N(\tilde z_{ik};\tilde{\mu}_{is_i(k,\backslash k)}, \tilde{\sigma}^2_{s_i(k,\backslash k)})\mathbb{I}\{ \tilde z_{ik} \in (l_{ik},u_{ik}]\},
\end{align*}
where 
\begin{align} \label{eq:discrete_y}
    p(y_{ik}\mid \tilde{\boldsymbol{\alpha}}_{k,\tilde{\boldsymbol{\gamma}}_k},\tilde{\boldsymbol{\gamma}}_k,\boldsymbol{\xi}_k, \tilde z_{ik})  & =\begin{cases}
1, & \mbox{if } F_k(y_{ik}-1;{\boldsymbol{\beta}}_{k}, {\boldsymbol{\xi}}_k,\mathbf{x}_i,t_i )< \Phi(\tilde z_{ik})\le F_k(y_{ik};{\boldsymbol{\beta}}_{k}, {\boldsymbol{\xi}}_k,\mathbf{x}_i,t_i ), \\
0, & \mbox{otherwise}.
\end{cases}
\end{align}
Here, $l_{ik}$, $u_{ik}$, $\tilde{\mu}_{i,s_i(k,\backslash k)}$, and $\tilde{\sigma}_{s_i(k,\backslash k)}^2$ are defined as in \eqref{app:musigmaeq}.

\paragraph{Step 5: Update of $\mathbf{s}$.}
We derive the posterior allocation probabilities for the indicator variables $\mathbf{s}$. For each $i=1,\ldots,n$, 
\begin{align*}
\Pr(s_i=h \mid \mathrm{rest})
&\propto p(\mathbf{y}_i \mid \tilde{\boldsymbol{\alpha}}_{\tilde{\boldsymbol{\gamma}}}, \tilde{\boldsymbol{\gamma}}, \boldsymbol{\xi}, \tilde{\mathbf{z}}_i)\,
p(\tilde{\mathbf{z}}_i \mid s_i=h,\mathbf{R}_h)\,
\Pr(s_i=h \mid \boldsymbol{\alpha}^*_{h,\boldsymbol{\gamma}_h^*},\boldsymbol{\gamma}_h^*) \\
&\propto p(\tilde{\mathbf{z}}_i \mid s_i=h,\mathbf{R}_h)\,
\Pr(s_i=h \mid \boldsymbol{\alpha}^*_{h,\boldsymbol{\gamma}_h^*},\boldsymbol{\gamma}_h^*) \\
&\propto |\mathbf{R}_h|^{-1/2}
\exp\!\left\{-\frac{1}{2}\tilde{\mathbf{z}}_i^\top \mathbf{R}_h^{-1}\tilde{\mathbf{z}}_i\right\}\pi_h(t_i),
\end{align*}
where
$
\pi_h(t_i)=\Phi(f_h(t_i))\prod_{\eta<h}[1-\Phi(f_\eta(t_i))],
$
with
$
f_\eta(t_i)=\alpha_{\eta0}^\ast+\boldsymbol{\alpha}_{\eta\backslash0,\boldsymbol{\gamma}_\eta^*}^{\ast\top}\mathbf{b}_{\eta,\boldsymbol{\gamma}_\eta^*}(t_i)$. 
Therefore, for $h=1,\ldots,H$, the normalized posterior allocation probability is
\begin{align*}
\Pr(s_i=h \mid \mathrm{rest})
=
\frac{
|\mathbf{R}_h|^{-1/2}
\exp\!\left\{-\frac{1}{2}\tilde{\mathbf{z}}_i^\top \mathbf{R}_h^{-1}\tilde{\mathbf{z}}_i\right\}\pi_h(t_i)
}{
\sum_{\eta=1}^H
|\mathbf{R}_\eta|^{-1/2}
\exp\!\left\{-\frac{1}{2}\tilde{\mathbf{z}}_i^\top \mathbf{R}_\eta^{-1}\tilde{\mathbf{z}}_i\right\}\pi_\eta(t_i)
}.
\end{align*}
We then sample $s_i$ from the multinomial distribution with these probabilities.

\paragraph{Step 6: Update of $\mathbf{z}^*$.}
Conditional on the updated $\mathbf{s}$, the latent variables $\mathbf{z}^*$ are sampled independently from truncated normal distributions. Their conditional distribution factorizes as
\begin{align*}
    p(\mathbf{z}^* \mid \mathrm{rest})
    = \prod_{i=1}^n \prod_{\eta=1}^{s_i}
    p(z_{i\eta}^* \mid \boldsymbol{\alpha}^*_{\eta,\boldsymbol{\gamma}_\eta^*}, \boldsymbol{\gamma}_\eta^*, s_i, \mathbf{R}_{s_i}).
\end{align*}
For $\eta=s_i$, the conditional density is
\begin{align*}
    p(z_{i\eta}^* \mid \boldsymbol{\alpha}^*_{\eta,\boldsymbol{\gamma}_\eta^*}, \boldsymbol{\gamma}_\eta^*, s_i, \mathbf{R}_{s_i})
    \propto
    \exp\!\left\{-\frac{1}{2}(z_{i\eta}^* - f_\eta(t_i))^2\right\}\mathbb{I}(z_{i\eta}^* \ge 0).
\end{align*}
For $\eta < s_i$, the conditional density is
\begin{align*}
    p(z_{i\eta}^* \mid \boldsymbol{\alpha}^*_{\eta,\boldsymbol{\gamma}_\eta^*}, \boldsymbol{\gamma}_\eta^*, s_i, \mathbf{R}_{s_i})
    \propto
    \exp\!\left\{-\frac{1}{2}(z_{i\eta}^* - f_\eta(t_i))^2\right\}\mathbb{I}(z_{i\eta}^* < 0).
\end{align*}
Therefore, for $i=1,\ldots,n$ and $\eta=1,\ldots,s_i$,
\begin{gather*}
    z_{i\eta}^* \mid \mathrm{rest} \sim
    \begin{cases}
        \mathrm{TN}_{(0,\infty)}\bigl(f_\eta(t_i),1\bigr), & \text{if } \eta=s_i,\\
        \mathrm{TN}_{(-\infty,0]}\bigl(f_\eta(t_i),1\bigr), & \text{if } \eta<s_i.
    \end{cases}
\end{gather*}

\paragraph{Step 7: Update of $(\boldsymbol{\alpha}^*_{\boldsymbol{\gamma}^*},\boldsymbol{\gamma}^*)$.}
By allowing the latent variable $\mathbf{z}^*$, 
we update $(\boldsymbol{\alpha}_{h,\boldsymbol{\gamma}_h^*}^*,\boldsymbol{\gamma}^*_h)$ using Gibbs sampling. For $h = 1, \ldots, H - 1$, 
\begin{align}
    \label{app:posterior_alpha_gamma_star}p(\boldsymbol{\alpha}_{h,\boldsymbol{\gamma}_h^*}^*, \boldsymbol{\gamma}_h^*\mid \mathrm{rest}) \propto p(\boldsymbol{\alpha}_{h,\boldsymbol{\gamma}_h^*}^*\mid \mathrm{rest})p( \boldsymbol{\gamma}_h^*\mid \mathrm{rest} \setminus \boldsymbol{\alpha}_{h,\boldsymbol{\gamma}_h^*}^*).
\end{align}
The posterior inclusion probability for $\boldsymbol{\gamma}_{h\ell}^*$ in \eqref{app:posterior_alpha_gamma_star} is derived as follows:
For $\ell=1,\ldots,L$,
\begin{align*}
    p({\gamma}_{h\ell}^* \mid  \mathrm{rest}\setminus \boldsymbol{\alpha}^\ast_{h,\boldsymbol{\gamma}_h^\ast})
    & \propto p(\mathbf{z}_{\cdot h}^* \mid   \boldsymbol{\gamma}_{h}^*,{g}_h^*) p({\gamma}_{h\ell}^*\mid \boldsymbol{\gamma}_{h\backslash \ell}^*),
\end{align*}
where
\begin{align}
\begin{split}
   &  p(  \mathbf{z}_{\cdot h}^*\mid   \boldsymbol{\gamma}_{h}^*,{g}_h^* ) \\
    & = \int p(\mathbf{z}_{\cdot h}^* \mid  \boldsymbol{\alpha}_{h,\boldsymbol{\gamma}_{h}^*}^\ast, \boldsymbol{\gamma}_{h}^*)p(\boldsymbol{\alpha}_{h,\boldsymbol{\gamma}_{h}^*}^\ast\mid \boldsymbol{\gamma}_h^*,{g}_h^*)d\boldsymbol{\alpha}_{h,\boldsymbol{\gamma}_{h}^*}^\ast,\\
   & \propto \!\left|g_h^*(\mathbf{W}^{*\top}_{h,\boldsymbol{\gamma}^*_h}  \mathbf{W}^{*}_{h,\boldsymbol{\gamma}^*_h})^{-1}\right|^{-1/2} \\
   &\quad  \times \int  \exp\!\left\{ -{1\over 2}(\mathbf{z}_{\cdot h}^* - \alpha_{h0}^*\mathbf{1} - \widehat{\mathbf{W}}_{h,\boldsymbol{\gamma}^*_h}^{*}\boldsymbol{\alpha}^*_{h\backslash 0, \boldsymbol{\gamma}_h^*})^\top(\mathbf{z}_{\cdot h}^* - \alpha_{h0}^*\mathbf{1} - \widehat{\mathbf{W}}_{h,\boldsymbol{\gamma}^*_h}^{*}\boldsymbol{\alpha}^*_{h\backslash 0, \boldsymbol{\gamma}_h^*}) \right\} \\
   & \quad \times \exp\!\left\{-{1\over 2}\boldsymbol{\alpha}_{h,\boldsymbol{\gamma}_h^*}^{\ast\top}\begin{pmatrix}
             1 & \mathbf{0}\\
             \mathbf{0} & {g_h^{\ast }}^{-1}\mathbf{W}^{\ast\top}_{h,\boldsymbol{\gamma}^\ast_h}\mathbf{W}^{\ast}_{h,\boldsymbol{\gamma}^\ast_h}
         \end{pmatrix}\boldsymbol{\alpha}_{h,\boldsymbol{\gamma}_h^*}^{\ast}\right\}d\boldsymbol{\alpha}_{h,\boldsymbol{\gamma}_h^*}^*,\\
& \propto \!\left|{g_h^{\ast-1}}(\mathbf{W}^{*\top}_{h,\boldsymbol{\gamma}^*_h}\mathbf{W}^{*}_{h,\boldsymbol{\gamma}^*_h})\right|^{1/2} \!\left|\begin{pmatrix}
             \sum_{\eta\ge h}n_\eta + 1 & \mathbf{1}^\top\widehat{\mathbf{W}}^*_{h,\boldsymbol{\gamma}_h^*}\\
         \widehat{\mathbf{W}}^{*\top}_{h,\boldsymbol{\gamma}_h^{*}}\mathbf{1} & \widehat{\mathbf{W}}^{*\top}_{h,\boldsymbol{\gamma}_h^{*}} \widehat{\mathbf{W}}^*_{h,\boldsymbol{\gamma}_h^*}
             + {g_h^*}^{-1}\mathbf{W}^{*\top}_{h,\boldsymbol{\gamma}_h^*}\mathbf{W}^{*}_{h,\boldsymbol{\gamma}_h^*}
             \end{pmatrix}\right|^{-1/2}\\
         & \quad \times \exp\!\left\{{1\over 2} \mathbf{z}_{\cdot h}^{*\top}\begin{pmatrix}
             \mathbf{1} & \widehat{\mathbf{W}}_{h,\boldsymbol{\gamma}_h^*}^{*}
         \end{pmatrix}\begin{pmatrix}
             \sum_{\eta\ge h}n_\eta + 1 & \mathbf{1}^\top\widehat{\mathbf{W}}^*_{h,\boldsymbol{\gamma}_h^*}\\
         \widehat{\mathbf{W}}^{*\top}_{h,\boldsymbol{\gamma}_h^{*}}\mathbf{1} & \widehat{\mathbf{W}}^{*\top}_{h,\boldsymbol{\gamma}_h^{*}} \widehat{\mathbf{W}}^*_{h,\boldsymbol{\gamma}_h^*}
             + {g_h^*}^{-1}\mathbf{W}^{*\top}_{h,\boldsymbol{\gamma}_h^*}\mathbf{W}^{*}_{h,\boldsymbol{\gamma}_h^*}
             \end{pmatrix}^{-1}\begin{pmatrix}
             \mathbf{1}^\top \\ \widehat{\mathbf{W}}_{h,\boldsymbol{\gamma}_h^*}^{*\top}
         \end{pmatrix}\mathbf{z}_{\cdot h}^*\right\}.
\end{split}
\label{app:zstar-lhd}
\end{align}
Therefore, 
\begin{align*}
 \mathrm{Pr}\!\left({\gamma}_{h\ell }^*=1\mid \mathrm{rest}\setminus \boldsymbol{\alpha}^\ast_{h,\boldsymbol{\gamma}_h^\ast}\right)  &  =  \!\left( 1 + {p(\mathbf{z}^*_{\cdot h}\mid\gamma^*_{h\ell}=0, \boldsymbol{\gamma}^*_{h\backslash \ell},{g}_h^*)p(\gamma_{h\ell}^*=0\mid \boldsymbol{\gamma}^*_{h\backslash \ell})\over p(\mathbf{z}^*_{\cdot h}\mid \gamma^*_{h\ell}=1, \boldsymbol{\gamma}^*_{h\backslash \ell},{g}_h^*)p(\gamma_{h\ell}^*=1\mid \boldsymbol{\gamma}^*_{h\backslash \ell})} \right)^{-1},
\end{align*}
where 
$p(\gamma^*_{h\ell}=1\mid \boldsymbol{\gamma}^*_{h\backslash \ell})={(1-\varpi_h)(|\boldsymbol{\gamma}^*_{h\backslash \ell}|+1)( L-|\boldsymbol{\gamma}^*_{h\backslash \ell}| +(1-\varpi_h)(|\boldsymbol{\gamma}^*_{h\backslash \ell}|+1))^{-1}}$.
Given the updated $\boldsymbol{\gamma}_h^*$, the conditional posterior distribution of $\boldsymbol{\alpha}_{h,\boldsymbol{\gamma}_h^*}^*$ is Gaussian. The posterior form is derived as follows.
\begin{align*}
     p(\boldsymbol{\alpha}^*_{h,\boldsymbol{\gamma}^*_h}\mid \mathrm{rest})  &   \propto   p(\mathbf{y},\mathbf{z}_{\cdot h}^*\mid \tilde{\boldsymbol{\alpha}}_{\tilde{\boldsymbol{\gamma}}}, \tilde{\boldsymbol{\gamma}}, \boldsymbol{\xi},\tilde{\mathbf{z}}, \boldsymbol{\alpha}^*_{h,\boldsymbol{\gamma}_h^*},\boldsymbol{\gamma}^*_h)p(\boldsymbol{\alpha}^*_{h,\boldsymbol{\gamma}^*_h}|\boldsymbol{\gamma}^*_h,g^*_h),\\
    &  = p(\mathbf{y}\mid \tilde{\boldsymbol{\alpha}}_{\tilde{\boldsymbol{\gamma}}}, \tilde{\boldsymbol{\gamma}}, \boldsymbol{\xi},\tilde{\mathbf{z}})p(\mathbf{z}^*_{\cdot h} \mid   \boldsymbol{\alpha}^*_{h,\boldsymbol{\gamma}_h^*},\boldsymbol{\gamma}^*_h) p(\boldsymbol{\alpha}^*_{h,\boldsymbol{\gamma}^*_h}\mid \boldsymbol{\gamma}^*_h,g^*_h),\\
    & \propto \prod_{i:s_i\ge h}p({z}_{ih}^* \mid {\alpha}_{h,\boldsymbol{\gamma}_h^*}^*) p(\boldsymbol{\alpha}^*_{h,\boldsymbol{\gamma}^*_h}\mid \boldsymbol{\gamma}^*_h,g^*_h),\\ 
       & \propto \exp\!\left\{ -{1\over 2}(\mathbf{z}_{\cdot h}^* - \alpha_{h0}^*\mathbf{1} - \widehat{\mathbf{W}}_{h,\boldsymbol{\gamma}^*_h}^{*}\boldsymbol{\alpha}^*_{h\backslash 0,\boldsymbol{\gamma}_h^*})^\top(\mathbf{z}_{\cdot h}^* - \alpha_{h0}^*\mathbf{1} - \widehat{\mathbf{W}}_{h,\boldsymbol{\gamma}^*_h}^{*}\boldsymbol{\alpha}^*_{h\backslash 0,\boldsymbol{\gamma}_h^*}) \right\} \\
       & \quad \times \exp\!\left\{-{1\over 2}\boldsymbol{\alpha}_{h,\boldsymbol{\gamma}_h^*}^{\ast\top}\begin{pmatrix}
             1 & \mathbf{0}\\
             \mathbf{0} & {g_h^{\ast }}^{-1}\mathbf{W}^{\ast\top}_{h,\boldsymbol{\gamma}^\ast_h}\mathbf{W}^{\ast}_{h,\boldsymbol{\gamma}^\ast_h}
         \end{pmatrix}\boldsymbol{\alpha}_{h,\boldsymbol{\gamma}_h^*}^{\ast}\right\}.
         \end{align*}
Hence
\begin{align*}
\boldsymbol{\alpha}_{h,\boldsymbol{\gamma}_h^\ast}^\ast \mid \mathrm{rest}
          \sim\  \text{N}_{2+|\boldsymbol{\gamma}_h^*|}\!\left(\boldsymbol{\mu}_{\boldsymbol{\alpha}_h^*},\mathbf{V}_{\boldsymbol{\alpha}_h^*} \right),
\end{align*}
where \begin{align*}
    \mathbf{V}_{\boldsymbol{\alpha}^*_h}& =\begin{pmatrix}
             \sum_{\eta\ge h}n_\eta + 1 & \mathbf{1}^\top\widehat{\mathbf{W}}^*_{h,\boldsymbol{\gamma}_h^*}\\
             \widehat{\mathbf{W}}^{*\top}_{h,\boldsymbol{\gamma}_h^{*}}\mathbf{1} & \widehat{\mathbf{W}}^{*\top}_{h,\boldsymbol{\gamma}_h^{*}} \widehat{\mathbf{W}}^*_{h,\boldsymbol{\gamma}_h^*}
             + {g_h^*}^{-1}\mathbf{W}^{*\top}_{h,\boldsymbol{\gamma}_h^*}\mathbf{W}^{*}_{h,\boldsymbol{\gamma}_h^*}
         \end{pmatrix}^{-1}, \\
         \boldsymbol{\mu}_{\boldsymbol{\alpha}^*_h} & = \mathbf{V}_{\boldsymbol{\alpha}_h^*}\begin{pmatrix}
             \mathbf{1}^\top
             \\ \widehat{\mathbf{W}}^{*\top}_{h,\boldsymbol{\gamma}_h^*}
         \end{pmatrix}\mathbf{z}_{\cdot h}^*.
\end{align*}
\paragraph{Step 8: Update of  $\mathbf{g}^*$.}
For $h=1,\ldots,H-1$, the conditional posterior distribution of $g_h^*$ is derived as
\begin{align*}
   p(g_h^* \mid \mathrm{rest}) 
   & \propto p(\boldsymbol{\alpha}_{h,\boldsymbol{\gamma}_h^*}^*\mid \boldsymbol{\gamma}^*_h,g_h^*)p(g_h^*),\\
   & \propto (g_h^*)^{-{1\over 2}} \!\left|g_h^*(\mathbf{W}_{h,\boldsymbol{\gamma}_h^*}^{*\top}\mathbf{W}_{h,\boldsymbol{\gamma}_h^*}^{*})^{-1}\right|^{-{1\over 2}}\exp\!\left\{-{1\over 2g_h^*} \boldsymbol{\alpha}_{h\backslash 0,\boldsymbol{\gamma}_h^*}^{*\top}\mathbf{W}_{h,\boldsymbol{\gamma}_h^*}^{*\top}\mathbf{W}_{h,\boldsymbol{\gamma}_h^*}^{*}\boldsymbol{\alpha}_{h\backslash 0,\boldsymbol{\gamma}_h^*}^{*}\right\}\exp\!\left\{-{n\over 2g_h^*}\right\}.
\end{align*}
Therefore, 
\begin{gather*} 
   g_h^\ast  \sim \text{IG}\!\left({1\over 2}({|\boldsymbol{\gamma}_h^*|+2}),{1\over 2}\!\left(\boldsymbol{\alpha}_{h\backslash 0,\boldsymbol{\gamma}_h^*}^{*\top}\mathbf{W}_{h,\boldsymbol{\gamma}_h^*}^{*\top}\mathbf{W}_{h,\boldsymbol{\gamma}_h^*}^{*}\boldsymbol{\alpha}_{h\backslash 0,\boldsymbol{\gamma}_h^*}^{*} + n \right)\right).
\end{gather*}
% -------------
\paragraph{Step 9: Update of  $(\mathbf{D},\mathbf{R})$.}
The correlation matrix $\mathbf{R}_h$ does not admit a conjugate prior, and is therefore updated using either a parameter expansion method or a Metropolis--Hastings algorithm. In this work, we adopt the parameter expansion approach, which expands $\mathbf{R}_h$ to a covariance matrix $\boldsymbol{\Sigma}$ that allows for a conjugate inverse Wishart prior. 
Let $\tilde{\mathbf{z}}^\star_h = \{ \tilde{\mathbf{z}}_i^\top\mathbf{D}_{h}\}_{i:s_i=h}$. Then each row of $\tilde{\mathbf{z}}_h^\star$, i.e. $ \tilde{\mathbf{z}}_{ih}^{\star\top} = (\tilde{\mathbf{z}}_i^\top\mathbf{D}_{h})^\top$ follows $N_m(\mathbf{0},\boldsymbol{\Sigma}_h)$, where $\boldsymbol{\Sigma}_h=\mathbf{D}_h\mathbf{R}_h\mathbf{D}_h$. For $h=1,\ldots,H$,
\begin{align*}
     p(\boldsymbol{\Sigma}_h\mid \mathrm{rest})
    & \propto p(\mathbf{y}\mid \tilde{\boldsymbol{\alpha}}_{\tilde{\boldsymbol{\gamma}}}, \tilde{\boldsymbol{\gamma}},\boldsymbol{\xi},\tilde{\mathbf{z}})p(\tilde{\mathbf{z}}_h^\star\mid \mathbf{s}, \boldsymbol{\Sigma}_h)p(\boldsymbol{\Sigma}_h),\\
    & \propto \prod_{i:s_i=h}p(\tilde{\mathbf{z}}^\star_{ih}\mid  {s}_i, \boldsymbol{\Sigma}_h)p(\boldsymbol{\Sigma}_h),\\
   & \propto |\boldsymbol{\Sigma}_h|^{-{n_h/ 2}}\exp\!\left\{-{1\over 2}\sum_{i:s_i=h}\tilde{\mathbf{z}}^\star_{ih}\boldsymbol{\Sigma}^{-1}_h\tilde{\mathbf{z}}^{\star\top}_{ih}\right\}|\boldsymbol{\Sigma}_h|^{-(m+1+m+1)/2}\exp\!\left\{ -{1\over 2}\text{tr}\!\left(\boldsymbol{\Sigma}_h^{-1}\right)\right\},\\
   & \propto |\boldsymbol{\Sigma}_h|^{-(n_h+m+1+m+1)/2}\exp\!\left\{-{1\over 2}\text{tr}\!\left(\!\left(\mathbf{I}_m+\tilde{\mathbf{z}}_h^{\star \top}\tilde{\mathbf{z}}_h^{\star }\right)\boldsymbol{\Sigma}_h^{-1}\right)\right\}.
   \end{align*}
Therefore, 
\begin{align*}
  \boldsymbol{\Sigma}_h \mid \mathrm{rest} \sim \text{IW}_{n_h+m+1}\!\left(\!\left(\mathbf{I}_m+\tilde{\mathbf{z}}_h^{\star \top}\tilde{\mathbf{z}}_h^{\star }\right)^{-1}\right).
\end{align*}
After sampling $\boldsymbol{\Sigma}_h$, we recover the auxiliary scale matrix and the correlation matrix through
$\mathbf{D}_h = \sqrt{\mathrm{diag}(\boldsymbol{\Sigma}_h)}$
and
$\mathbf{R}_h = \mathbf{D}_h^{-1}\boldsymbol{\Sigma}_h\mathbf{D}_h^{-1}$.
The expanded latent variables are then transformed back to the original copula scale using
$\tilde{\mathbf{z}}_i = \tilde{\mathbf{z}}_i^\star \mathbf{D}_{s_i}^{-1}$.

\section{Proofs}  \label{app:sec2}

\begin{proof}[Proof of Theorem~\ref{thm:appGC}]
We first show that for any $\mathbf R_1,\mathbf R_2 \in \mathcal Q^m$, there exists a constant $C$ depending only on $m$ and the minimum eigenvalue in $\mathcal Q^m$ such that
\begin{align}
     \lVert c_G(\cdot\,;\mathbf R_1) - c_G(\cdot\,;\mathbf R_2)\rVert_{L^1[0,1]^m} \le C  \lVert\mathbf R_1-\mathbf R_2\rVert_F,
     \label{eqn:GCbo}
\end{align}
where $\lVert\cdot\rVert_F$ is the Frobenius norm. To verify this, note that the substitution $u_k=\Phi(z_k)$ yields
$$
c_G(u_1,\dots,u_m;\mathbf R_j) = \frac{\phi(z_1,\dots, z_m;\mathbf R_j)}{\prod_{k=1}^m \phi(z_k)},\quad du_1\dots du_m = \prod_{k=1}^m \phi(z_k)dz_1\dots dz_m,
$$
which gives
    \begin{align*}
        \lVert c_G(\cdot\,;\mathbf R_1) - c_G(\cdot\,;\mathbf R_2)\rVert_{L^1[0,1]^m} &=   \lVert \phi(\cdot\,;\mathbf R_1) - \phi(\cdot\,;\mathbf R_2)\rVert_{L^1(\mathbb R^m)}.
 %       \label{eqn:gascop}
    \end{align*}
Observe that the Kullback–Leibler divergence between two zero-mean multivariate normal distributions with covariance matrices $\mathbf V_1$ and $\mathbf V_2$ satisfies
\begin{align}
\begin{split}
    KL(\phi(\cdot\,;\mathbf V_1), \phi(\cdot\,;\mathbf V_2)) &=\frac{1}{2}[\text{tr}(\mathbf V_2^{-1}\mathbf V_1)-m-\log\det(\mathbf V_2^{-1}\mathbf V_1)]\\
    &=\frac{1}{2}\sum_{k=1}^m ( \lambda_k-1-\log \lambda_k),
    \end{split}
   \label{eqn:KLdiv}
\end{align}
where $\lambda_k$ denote the eigenvalues of
$\mathbf V_2^{-1/2}\mathbf V_1\mathbf V_2^{-1/2}$.
Using the inequality
$x-1-\log x \le(x-1)^2/(2\min\{x,1\})$ for every $x>0$, the rightmost side of \eqref{eqn:KLdiv} is bounded by $C_1\sum_{k=1}^m (\lambda_k-1)^2$, where $C_1$ depends only on the eigenvalues of $\mathbf V_1$ and $\mathbf V_2$. By Lemma~A.1 of \citet{banerjee2015bayesian}, this term is further bounded by $C_2\lVert\mathbf V_1-\mathbf V_2\rVert_F^2$ for some $C_2$ depending only on those eigenvalues. Therefore, we obtain \eqref{eqn:GCbo} by Pinsker's inequality.

Now set $\delta=\epsilon/(2C)$ and let $\{\mathbf R_h\}_{h=1}^H$ be a finite $\delta$-net of the range $\tilde{\mathbf R}(\mathcal T)$ with respect to the Frobenius norm. The finiteness of $H$ follows since the range of a compact space under a continuous map is compact (Theorem 26.5 of \citet{munkres2000topology}), hence $\tilde{\mathbf R}(\mathcal T)$ is totally bounded, which implies that $H$ is finite.
Define the bump function
    $$
\psi_h(t)=\max\!\left\{0,1-\frac{\lVert\tilde{\mathbf R}(t)-\mathbf R_h\rVert_F}{2\delta}\right\},
    $$
which is continuous and satisfies $\psi_h(t)>0$ only if $\lVert\tilde{\mathbf R}(t)-\mathbf R_h\rVert_F<2\delta$.
Since the collection $\{\mathbf R_h\}_{h=1}^H$ forms a $\delta$-net, for every given $t\in\mathcal T$ there exists at least one $h$ with $\lVert\tilde{\mathbf R}(t)-\mathbf R_h\rVert_F\le\delta$, and hence $\psi_h(t)> 0$ for that $h$.
   We therefore normalize $\psi_h(t)$ to form a partition of unity: $\pi_h(t)=\psi_h(t)/\sum_{j=1}^H \psi_j(t)$.
Using \eqref{eqn:GCbo}, observe that
\begin{align*}
    \bigg\lVert c_G(\cdot \,;\tilde{\mathbf R}(t))-\sum_{h=1}^H \pi_h(t)c_G(\cdot\,;\mathbf R_h)\bigg\rVert_{L^1[0,1]^m}&\le \sum_{h=1}^H \pi_h(t) \lVert c_G(\cdot \,;\tilde{\mathbf R}(t))-c_G(\cdot\,;\mathbf R_h)\rVert_{L^1[0,1]^m}\\
    &\le C \sum_{h=1}^H \pi_h(t) \lVert \tilde{\mathbf R}(t)-\mathbf R_h\rVert_F.
\end{align*}
Since $\pi_h(t)>0$ only if $\lVert \tilde{\mathbf R}(t)-\mathbf R_h\rVert_F < 2\delta$, the last expression is bounded by $2C\delta=\epsilon$.
\end{proof}

\begin{proof}[Proof of Theorem~\ref{thm:appGCM}]
    Note that $\mathcal Q^m(\lambda_0)$ is compact. Let $\delta=\epsilon/(2C)$ using $C$ in \eqref{eqn:GCbo}. Then there exist a finite $\delta$-net $\{\mathbf R_h\}_{h=1}^H$ of $\mathcal Q^m(\lambda_0)$ with respect to the Frobenius norm. Now, define the bump functions $\psi_h:\mathcal Q^m(\lambda_0)\rightarrow[0,1]$ and their normalized versions $\psi_h^\ast$ by
    $$
\psi_h(\mathbf R) = \max\!\left\{0,1-\frac{\lVert \mathbf R-\mathbf R_h \rVert_F}{2\delta}\right\}, \quad \psi_h^\ast(\mathbf R) = \frac{\psi_h(\mathbf R)}{\sum_{j=1}^H \psi_j(\mathbf R)}.
    $$
The normalization $\psi_h^\ast$ is well defined because, for every $\mathbf R\in \mathcal Q^m(\lambda_0)$, there exists at least one $h$ such that $\lVert \mathbf R-\mathbf R_h \rVert_F \le \delta$, which ensures $\psi_h(\mathbf R)>0$.
Define the weight functions $\pi_h(t) = \int_{\mathcal Q^m(\lambda_0)}\psi_h^\ast(\mathbf R)K(t,d\mathbf R)$, which satisfy $\sum_{h=1}^H \pi_h(t)=1$ since $K$ is a Markov kernel. Moreover, $\pi_h$ is continuous owing to the weak-Feller property.
Now observe that
\begin{align*}
       \bigg\lVert c_K(\cdot\,;t) - \sum_{h=1}^H \pi_h(t)c_G(\cdot\,; \mathbf R_h) \bigg\rVert_{L^1[0,1]^m}  &=  \bigg\lVert \int \sum_{h=1}^H \psi_h^\ast(\mathbf R)\Big[c_G(\cdot\,;\mathbf R)-c_G(\cdot\,;\mathbf R_h)\Big] K(t, d\mathbf R)\bigg\rVert_{L^1[0,1]^m} \\
       & \le\int \sum_{h=1}^H \psi_h^\ast(\mathbf R)\Big\lVert c_G(\cdot\,;\mathbf R)-c_G(\cdot\,;\mathbf R_h)\Big\rVert_{L^1[0,1]^m}  K(t, d\mathbf R)\\
        & \le C\int \sum_{h=1}^H \psi_h^\ast(\mathbf R)\lVert \mathbf R - \mathbf R_h\rVert_F K(t, d\mathbf R),
    \end{align*}
where the last inequality follows from \eqref{eqn:GCbo}.
Since $\psi_h^\ast(\mathbf R)>0$ only if $\lVert \mathbf R-\mathbf R_h\rVert_F < 2\delta$, the last expression is bounded by $2C\delta=\epsilon$.
\end{proof}

\endgroup

%-------------------------------------------------------------------